\documentclass[aps,%
twocolumn,%
floatfix,%
pra,%
amsfonts,%
groupedaddress,%
superscriptaddress]{revtex4-2}%
\usepackage[utf8]{inputenc}
\usepackage[T1]{fontenc}
\usepackage{mathrsfs}
\usepackage{mathtools}
\usepackage{amsmath}
\usepackage{amsfonts}
\usepackage{amssymb}
\usepackage[titletoc,title]{appendix}
\usepackage{algorithm,algorithmic}

\usepackage{bbm}
\usepackage{bm}
\usepackage{amsthm}
\usepackage{makecell}
\usepackage{array}
\usepackage[dvipsnames,xcdraw]{xcolor}
\usepackage{colortbl}
\makeatletter
\IfPackageAtLeastTF{array}{2024/06/01}{\let\switch@array\relax}{}
\makeatother

\definecolor{beamer@blendedblue}{rgb}{0.2,0.2,0.7}
\usepackage[osf,sc]{mathpazo}
\usepackage{caption}
\usepackage{subfig}
\usepackage{overpic}
\usepackage{booktabs}
\usepackage{multirow}
\usepackage[skins,breakable,most]{tcolorbox}
\newcolumntype{C}[1]{>{\centering\arraybackslash}p{#1}}
\usepackage[colorlinks=true,%
linkcolor=beamer@blendedblue,%
bookmarks=true,%
breaklinks=true,%
filecolor=beamer@blendedblue,%
anchorcolor=yellow,%
citecolor=beamer@blendedblue,%
urlcolor=beamer@blendedblue,%
pdfauthor={Zhenyuan Huang, Ping Xu, and Kun Wang},%
pdfsubject={Worst-Case Sample Complexity Bounds for Distributed Inner Product Estimation with Local Randomized Measurements},%
CJKbookmarks=true]{hyperref}
\usepackage{float}
\usepackage{pifont}

\newtheorem{definition}{Definition}
\newtheorem{proposition}[definition]{Proposition}

\newtheorem{lemma}[definition]{Lemma}
\newtheorem{theorem}[definition]{Theorem}
\newtheorem{corollary}[definition]{Corollary}

\mathchardef\ordinarycolon\mathcode`\:
\mathcode`\:=\string"8000
\def\vcentcolon{\mathrel{\mathop\ordinarycolon}}
\begingroup \catcode`\:=\active
  \lowercase{\endgroup
  \let :\vcentcolon
  }

\tcolorboxenvironment{proof}{%
  enhanced,
  boxrule=0.5pt,
  breakable,
  sharp corners,
  before upper={\setlength{\parindent}{1.5em}}, 
  upperbox=visible,
}

\DeclareFontFamily{U}{mathx}{\hyphenchar\font45}
\DeclareFontShape{U}{mathx}{m}{n}{<-> mathx10}{}
\DeclareSymbolFont{mathx}{U}{mathx}{m}{n}
\DeclareMathAccent{\widebar}{0}{mathx}{"73}

\newcommand{\wh}[1]{\widehat{#1}}

\newcommand{\ket}[1]{\vert{#1}\rangle}
\newcommand{\bra}[1]{\langle{#1}\vert}
\newcommand{\ketbra}[1]{\vert{#1}\rangle\!\langle{#1}\vert}

\newcommand\proj[1]{\vert{#1}\rangle\!\langle{#1}\vert}

\newcommand{\density}[1]{\mathscr{D}(#1)}

\newcommand{\abs}[1]{\left\vert{#1}\right\vert}

\DeclareMathOperator{\tr}{Tr}

\newcommand{\ox}{\otimes}
\def\ve{\varepsilon}
\newcommand{\norm}[2]{\ensuremath{\left\lVert#1\right\rVert_{#2}}}%

\newcommand{\id}{\operatorname{id}}

\newcommand{\spn}[1]{{\rm span}\{#1\}}
\newcommand{\Var}{\mathbb{V}}

\newcommand{\Cl}{{\rm Cl}}
\newcommand{\ba}{\bm{a}}
\newcommand{\bb}{\bm{b}}

\newcommand{\bsigma}{\bm{\sigma}}

\makeatletter
\newsavebox{\@brx}
\newcommand{\llangle}[1][]{\savebox{\@brx}{\(\m@th{#1\langle}\)}%
  \mathopen{\copy\@brx\kern-0.5\wd\@brx\usebox{\@brx}}}
\newcommand{\rrangle}[1][]{\savebox{\@brx}{\(\m@th{#1\rangle}\)}%
  \mathclose{\copy\@brx\kern-0.5\wd\@brx\usebox{\@brx}}}
\makeatother

\newcommand*{\cO}{\mathcal{O}}
\newcommand*{\cP}{\mathcal{P}}

\newcommand*{\cS}{\mathcal{S}}

\newcommand*{\cU}{\mathcal{U}}

\newcommand{\bE}{\mathbb{E}}
\newcommand{\bC}{\mathbb{C}}

\newcommand{\bR}{\mathbb{R}}
\newcommand{\bP}{\mathbb{P}}
\newcommand{\bF}{\mathbb{F}}

\newcommand{\bI}{\mathbb{I}}

\definecolor{wildstrawberry}{rgb}{1.0, 0.26, 0.64}
\definecolor{googleblue}{HTML}{4285F4}
\definecolor{googlered}{HTML}{DB4437}
\definecolor{googleyellow}{HTML}{F4B400}
\definecolor{googlegreen}{HTML}{0F9D58}
\definecolor{klevinblue}{HTML}{002FA7}
\definecolor{tiffanyblue}{HTML}{0ABAB5}

\begin{document}

\newcommand{\thetitle}{{Worst-Case Sample Complexity Bounds for{\\}
Distributed Inner Product Estimation with Local Randomized Measurements}}

\title{\large\thetitle}

\author{Zhenyuan Huang}
\affiliation{College of Computer Science and Technology, 
National University of Defense Technology, 
Changsha 410073, China}%

\author{Kun Wang}
\email{nju.wangkun@gmail.com}
\affiliation{College of Computer Science and Technology, 
National University of Defense Technology, 
Changsha 410073, China}%

\author{Ping Xu}
\email{pingxu520@nju.edu.cn}
\affiliation{College of Computer Science and Technology, 
National University of Defense Technology, 
Changsha 410073, China}%

\date{\today}

\begin{abstract}
We study distributed inner product estimation for $n$-qubit states using local
randomized measurements, for which rigorous worst-case guarantees are less understood.
We first reduce the minimax kernel optimization to Hamming-distance kernels.
Within this class, unbiasedness fixes a unique kernel.
For this kernel under local Clifford sampling, we prove a fourth-moment
bound for arbitrary inputs. This yields a worst-case sample complexity upper bound of
$\cO(\sqrt{5.25^n})$ copies per party at fixed additive error and failure probability.
For the same kernel under local Haar sampling, we prove a local twirling
identity that compares its fourth moment with the Clifford fourth moment. 
A direct bound on the Haar fourth moment shows that $\cO(\sqrt{4^n})$ copies per party suffice.
We also show that independent single-qubit Pauli shadows have a worst-case
sample complexity upper bound of $\cO(\sqrt{7.5^n})$ for large $n$.
\end{abstract}

\maketitle

\section{Introduction}

Distributed inner product estimation (DIPE)~\cite{elben2020cross} asks two remote parties, given copies of unknown $n$-qubit states $\rho$ and $\sigma$, to estimate the Hilbert--Schmidt inner product $\tr\left[\rho\sigma\right]$ 
using only local operations and classical communication (LOCC). This task lies at the core of cross-platform verification, distributed benchmarking, and certification of quantum devices prepared on separate platforms~\cite{eisert2020quantuma,kliesch2021theory,knorzer2025distributed},
and has motivated extensive theoretical and experimental work~\cite{Zhu2022,qian2024multimodal,zheng2025distributed,dalton2025resourceefficient,zheng2024crossplatform,hinsche2025efficient,gong2024sample,arunachalam2024distributed,zheng2025optimal}.
The optimal sample complexity of DIPE was established in~\cite{anshu2022distributed}, 
where it is achieved by a protocol based on global randomized measurements~\cite{elben2023randomized}. 
The problem addressed in this work is to establish rigorous worst-case sample
complexity bounds for experimentally feasible DIPE protocols based on
\emph{local randomized measurements (LRM)}.

LRM is a practically significant variant of the randomized measurement
toolbox~\cite{elben2023randomized}: one samples an $n$-qubit unitary $U=\bigotimes_{\ell=1}^n U_\ell$
in product form from single-qubit ensembles and then performs computational
basis measurements~\cite{elben2020cross,huang2020shadows}. 
This differs significantly from global randomized measurements, 
where one samples a single $n$-qubit unitary, 
which may be highly entangling across the full system. 
LRM is therefore the experimentally favorable setting for distributed architectures: 
it is compatible with local control and avoids the need for global random circuits 
across the full system~\cite{elben2023randomized,hu2025quantum,chen2025quantum}. 
Despite this practical appeal, the theoretical picture remains incomplete. 
Elben \emph{et al.}~\cite{elben2020cross} proposed the standard local Haar protocol,
and Zhu \emph{et al.}~\cite{Zhu2022} experimentally validated it, but neither work
provides a sample complexity analysis. 
Zheng \emph{et al.}~\cite{zheng2025distributed} obtained average-case and
state-dependent sample complexity bounds for structured random circuits and
related ensembles, while Wu \emph{et al.}~\cite{wu2025state} proved a loose
upper bound $\cO(\sqrt{6^n})$. 
These results do not determine the optimal post-processing kernel under
worst-case variance, nor do they identify the tight worst-case sample
complexity of DIPE with LRM. 
This gap substantially limits the theoretical understanding of DIPE with LRM
and, in turn, impedes its reliable practical use.

In this work, we address these gaps in three steps. 
First, we identify the unique unbiased Hamming-distance kernel 
selected by the worst-case symmetry reduction.
Then, we derive worst-case guarantees for local Clifford and Haar sampling 
based on this unique kernel.
Third, we compare these shared-randomness protocols with independent 
single-qubit Pauli shadows. 
The main results are summarized in Tab.~\ref{tab:summarization}.

\begin{table*}[t]
\centering
\colorlet{resultshade}{googlered!10}
\renewcommand{\arraystretch}{1.8} 
\setlength{\tabcolsep}{0pt} 
\setlength\heavyrulewidth{0.3ex}  
\begin{tabular}{@{}C{0.42\textwidth}C{0.26\textwidth}C{0.32\textwidth}@{}}
\toprule
\textbf{Protocol / ensemble} & \textbf{Scaling in $n$} & \textbf{Source} \\ \midrule
Single-qubit Haar ensemble & $\cO(\sqrt{6^n})$ 
& \cite[Prop.~2]{wu2025state} \\
\rowcolor{resultshade}
Single-qubit Clifford ensemble & $\cO(\sqrt{5.25^n})$
& Cor.~\ref{lem:clifford-complexity} \\
\rowcolor{resultshade}
Single-qubit Haar ensemble & $\cO(\sqrt{4^n})$
& Thm.~\ref{lem:haar-complexity} \\
\rowcolor{resultshade}
Independent Pauli shadows
& $\cO(\sqrt{7.5^n})$
& Prop.~\ref{prop:independent-pauli-shadow-complexity-main} \\ \bottomrule
\end{tabular}
\caption{\raggedright\textbf{Summary of worst-case sample complexity scalings for DIPE with LRM.}
The table reports the exponential dependence on the number of qubits $n$,
omitting accuracy and confidence prefactors. All entries are upper bounds on
the number of copies per party for arbitrary input pairs. The shaded block records the new
results obtained in this work.}
\label{tab:summarization}
\end{table*}

\textbf{Organization.} This work is organized as follows. 
Sec.~\ref{sec:general-framework} recalls the DIPE with LRM framework. 
Sec.~\ref{sec:main-results} presents our main results, including the Hamming distance kernel reduction, and the single-qubit Clifford 
and Haar analyses. 
Sec.~\ref{sec:independent-pauli-shadow-main} analyzes DIPE based on independent Pauli shadows for comparison.

\textbf{Notation and conventions.}
We write $\density{(\bC^2)^{\otimes n}}$ for the set of density operators on the $n$-qubit Hilbert space $(\bC^2)^{\otimes n}$ and $\bR$ for the real numbers. For a positive integer $m$, $S_m$ denotes the symmetric group on $m$ elements. Single-qubit binary measurement outcomes are denoted by $s,t\in\{0,1\}$, whereas bold symbols such as $\bm s,\bm t\in\{0,1\}^n$ denote $n$-qubit outcome strings. Single-qubit subsystems are labeled by letters such as $A$ and $B$, while $\boldsymbol{A}=A_1\cdots A_n$ and $\boldsymbol{B}=B_1\cdots B_n$ denote the corresponding multipartite registers. The identity operator on the joint register is denoted by $\bI_{AB}$, or equivalently $\bI_{\boldsymbol{A}\boldsymbol{B}}$ when the tensor product structure is emphasized. The full swap between the two global registers is denoted by $\bF_{AB}$, or equivalently $\bF_{\boldsymbol{AB}}$ when the tensor product structure is emphasized, and more generally $\bF_S$ denotes the partial swap on a subset $S\subseteq[n]$. If $S$ is a finite set, then $|S|$ denotes its cardinality. We write $\bE[\cdot]$ for expectation and $\Var[\cdot]$ for variance. 
We use $\Theta(\cdot)$ and $\cO(\cdot)$ to denote standard asymptotic notation.

\section{DIPE with Local Randomized Measurements}\label{sec:general-framework}

We recall the standard kernel estimator formalism for DIPE specialized to the LRM setting. 
This subsection reviews standard material; 
the new results begin in the following sections.

Let $\rho_{\boldsymbol{A}},\sigma_{\boldsymbol{B}}\in\density{(\bC^2)^{\otimes n}}$. 
In one round of local DIPE, Alice and Bob apply the same random product unitary
\begin{align}
U=\bigotimes_{\ell=1}^n U_\ell,
\end{align}
where each $U_\ell$ is sampled independently from a single-qubit ensemble $\cU_1$, and then measure in the computational basis, obtaining bit strings $\bm s,\bm t\in\{0,1\}^n$. A kernel estimator is specified by a function
\begin{align}
f:\{0,1\}^n\times\{0,1\}^n\to\bR.
\end{align}
Given $N_U$ independent unitary samples and $N_M$ repeated shots for each unitary sample, the empirical estimator reads
\begin{align}
\label{eq:estimator}
\wh X_f(\rho_{\boldsymbol{A}},\sigma_{\boldsymbol{B}})
=
\frac{1}{N_U N_M^2}
\sum_{l=1}^{N_U}
\sum_{j,k=1}^{N_M}
 f(\bm s_{l,j},\bm t_{l,k}).
\end{align}
Since each shot consumes one copy of $\rho$ and one copy of $\sigma$, the total number of copies used per party is $N:=N_U N_M$; throughout, we refer to this quantity as the sample complexity. The copy bounds below for the arithmetic mean follow from variance bounds via Chebyshev's inequality. In particular, we do not prove matching minimax lower bounds for these estimators or protocols.

The averaged estimator can be written as
\begin{align}
\bE[\wh X_f(\rho_{\boldsymbol{A}},\sigma_{\boldsymbol{B}})]
=\tr\left[(\rho_{\boldsymbol{A}}\ox\sigma_{\boldsymbol{B}})\,\bar\Omega_f\right],
\end{align}
where
\begin{align}
\bar\Omega_f
:=
\bE_U\!\left[(U^\dagger\ox U^\dagger)\Omega_f(U\ox U)\right],
\end{align}
with
\begin{align}
\Omega_f:=\sum_{\bm s,\bm t}f(\bm s,\bm t)\ketbra{\bm s}_{\boldsymbol{A}}\ox\ketbra{\bm t}_{\boldsymbol{B}}.
\end{align}
Thus unbiased estimation of $\tr\left[\rho\sigma\right]$ is equivalent to the operator identity
$\bar\Omega_f=\bF_{\boldsymbol{AB}}$, where $\bF_{\boldsymbol{AB}}$ 
denotes the full swap between the two $n$-qubit copies $\boldsymbol{A}$ and $\boldsymbol{B}$.

To expose the variance-based copy scaling, it is convenient to group the data by unitary sample. Define
\begin{align}
X_M^{(l)}:=\frac{1}{N_M^2}\sum_{j,k=1}^{N_M}f(\bm s_{l,j},\bm t_{l,k}),
\qquad
\wh X_f=\frac{1}{N_U}\sum_{l=1}^{N_U}X_M^{(l)}.
\end{align}
Let $X_M$ denote a generic copy of the independent and identically distributed (i.i.d.) random variables $X_M^{(1)},\dots,X_M^{(N_U)}$. Then
\begin{align}
\Var[\wh X_f]=\frac{1}{N_U}\Var[X_M].
\end{align}
For a symmetric, unbiased kernel $f$, the law of total variance and a count of coincident shot indices give the following four-term decomposition
\begin{align}
\label{eq:four-term-decomp}
\Var[X_M]=\mathbb{V}^{(1)}+\mathbb{V}^{(2)}+\mathbb{V}^{(3)}+\mathbb{V}^{(4)},
\end{align}
where
\begin{subequations}
\label{eq:general-vterms}
\begin{align}
\mathbb{V}^{(1)}&=-\tr\left[\rho\sigma\right]^2,\label{eq:general-vterms-v1}\\
\mathbb{V}^{(2)}&=\frac{1}{N_M^2}\bE_U\bE_{\bm s,\bm t\mid U}[f(\bm s,\bm t)^2],\label{eq:general-vterms-v2}\\
\mathbb{V}^{(3)}&=\frac{N_M-1}{N_M^2}\bE_U\!\left[\Xi_U(\rho,\sigma)+\Xi_U(\sigma,\rho)\right],\label{eq:general-vterms-v3}\\
\mathbb{V}^{(4)}&=\left(\frac{N_M-1}{N_M}\right)^2
\bE_U[(\bE_{\bm s,\bm t\mid U}[f(\bm s,\bm t)])^2].\label{eq:fourth-term}
\end{align}
\end{subequations}
In the above, for a fixed unitary sample $U$, $\bm s,\bm s'$ are independent outcomes from measuring $U\rho U^\dagger$ in the computational basis, $\bm t,\bm t'$ are independent outcomes from measuring $U\sigma U^\dagger$.
$\Xi_U(\rho,\sigma):=\bE_{\bm s,\bm t\mid U}\![f(\bm s,\bm t)\,\bE_{\bm s'\mid U}f(\bm s',\bm t)]$
and the quantity $\Xi_U(\sigma,\rho)$ is defined analogously, with the roles of $\rho$ and $\sigma$ exchanged.
Appx.~\ref{app:measurement-protocol} derives the formula for an arbitrary unbiased kernel.

For any unbiased estimator, Chebyshev's inequality gives the failure probability bound
\begin{align}
\label{eq:chebyshev}
\bP(|\wh X_f-\tr\left[\rho\sigma\right]|\ge \ve) \le \Var[\wh X_f]/\ve^2.
\end{align}
Therefore, for a target failure probability $\delta\in(0,1)$, using $\Var[\wh X_f]=\Var[X_M]/N_U$, 
a sufficient condition is
\begin{align}
\label{eq:sample complexity-from-variance}
N:=N_U N_M \geq N_M\Var[X_M]/(\delta\ve^2).
\end{align}
Thus worst-case sample complexity bounds follow by optimizing the variance term 
in Eq.~\eqref{eq:sample complexity-from-variance}.

A median-of-means estimator~\cite{huang2020shadows} averages the one-setting estimates within disjoint groups of independent unitary settings and returns the median of those averages. For any sampling ensemble for which $X_M$ is unbiased and has finite variance, this reduces the confidence factor in the sufficient copy bound from $1/\delta$ to $\cO(\log(2/\delta))$ (Appx.~\ref{app:median-of-means}).
Since both guarantees use the same bound on $\Var[X_M]$, we retain Eq.~\eqref{eq:sample complexity-from-variance}  
to study how the copy bound scales with the number of qubits $n$ at fixed accuracy and confidence.

\section{Main Results for DIPE with Local Randomized Measurements}\label{sec:main-results}

We now present the main theoretical results of this work. We first characterize the optimal kernel structure 
under worst-case variance optimization,
and then derive the resulting Chebyshev-based sufficient-copy guarantees
for the single-qubit Clifford and single-qubit Haar ensembles.

\subsection{Optimal Hamming distance kernels}\label{sec:optimal-kernel}

We first show that the kernel optimization problem admits a sharp symmetry reduction. All proofs in this subsection are collected in Appx.~\ref{app:hamming-sufficiency}.

\begin{lemma}
\label{lem:hamming-sufficiency}
Assume that the local sampling protocol is covariant under the action of
\begin{align}
\Gamma:=(S_2)^n\rtimes S_n,
\end{align}
where $\rtimes$ denotes the semidirect product for the natural action of $S_n$ on $(S_2)^n$ by permuting the $n$ factors. 
That is, $\Gamma$ is generated by local outcome relabelings and qubit permutations. 
Then, for every admissible kernel $f$, its symmetrization
\begin{align}
\label{eq:sym-kernel}
f^{\mathrm{sym}}(\bm s,\bm t)
:=
\bE_{\pi\sim S_n}\bE_{\tau\sim(S_2)^n}
\left[f(\tau(\pi(\bm s)),\tau(\pi(\bm t)))\right]
\end{align}
depends only on the Hamming distance $D(\bm s,\bm t)$ and satisfies
\begin{align}
\sup_{\rho,\sigma}\Var_{\rho,\sigma}[\wh X_{f^{\mathrm{sym}}}]
\le
\sup_{\rho,\sigma}\Var_{\rho,\sigma}[\wh X_f].
\end{align}
Consequently,
\begin{align}
\label{eq:minimax-reduction}
\inf_{f:\bar\Omega_f=\bF_{\boldsymbol{AB}}}\sup_{\rho,\sigma}\Var_{\rho,\sigma}[\wh X_f]
=
\inf_{\substack{g=g(D):\\\bar\Omega_g=\bF_{\boldsymbol{AB}}}}\sup_{\rho,\sigma}\Var_{\rho,\sigma}[\wh X_g].
\end{align}
\end{lemma}

The covariance requirement above is natural in the local i.i.d. setting. Qubit permutation covariance follows automatically from exchangeability of the independently sampled single-qubit unitaries. Local bit-flip covariance is ensured by invariance of the one qubit sampling law under left multiplication by Pauli $X$, which holds for both single-qubit Haar and uniform single-qubit Clifford sampling studied in this work.

Lem.~\ref{lem:hamming-sufficiency} reduces the kernel search space from $4^n$ entries to the $n+1$ values $g(0),\dots,g(n)$. We now show that unbiasedness removes even this residual freedom.

\begin{theorem}
\label{thm:unique-kernel}
Assume that the independently sampled single-qubit measurement axes have the Haar second moment, as for the Clifford and Haar ensembles considered here. Let $f(\bm s,\bm t)=g(D(\bm s,\bm t))$ be a Hamming distance kernel. Then the following are equivalent:
\begin{enumerate}
    \item $\wh X_f$ is unbiased for $\tr\left[\rho\sigma\right]$ for all states $\rho_{\boldsymbol{A}}$ 
          and $\sigma_{\boldsymbol{B}}$;
    \item $\bar\Omega_f=\bF_{\boldsymbol{AB}}$;
    \item $\forall d=0,1,\dots,n$, $g(d)=(-1)^d2^{n-d}$.
\end{enumerate}
As a consequence, the unique unbiased kernel in the Hamming distance class is
\begin{align}\label{eq:product-kernel-main}
f(\bm s,\bm t)=(-1)^{D(\bm s,\bm t)}2^{n-D(\bm s,\bm t)}
=\prod_{\ell=1}^n\left(3\delta_{s_\ell,t_\ell}-1\right).
\end{align}
\end{theorem}

The kernel in Eq.~\eqref{eq:product-kernel-main} is exactly the local estimator
introduced in~\cite{elben2020cross}. Thus Lem.~\ref{lem:hamming-sufficiency}
and Thm.~\ref{thm:unique-kernel} show that this estimator is not an ad hoc
choice: within the Hamming distance class singled out by the worst-case
symmetry reduction, it is the unique unbiased kernel.

\subsection{DIPE with single-qubit Clifford ensemble}

The single-qubit Clifford ensemble is experimentally attractive and forms an
exact unitary $3$-design~\cite{webb2016clifford}. 
Practically, single-qubit Clifford randomized measurements are implemented
by choosing, independently for each qubit, one of the Pauli $X$, $Y$, and $Z$
bases uniformly at random and measuring in that basis.
All discussions and proofs in this subsection 
are collected in Appx.~\ref{app:clifford}.

Since the single-qubit Clifford ensemble form an unitary $3$-design,
for the unique kernel in Eq.~\eqref{eq:product-kernel-main}, 
the terms defined in Eq.~\eqref{eq:general-vterms} can be analytically calculated as
\begin{subequations}
\label{eq:shared-vterms}
\begin{align}
\mathbb{V}^{(1)}_{\Cl}
&= -\tr\left[\rho\sigma\right]^2,\label{eq:shared-vterms-v1}\\
\mathbb{V}^{(2)}_{\Cl}
&=
\frac{1}{N_M^2}
\tr\left[(2\bI+\bF)^{\otimes n}(\rho\ox\sigma)\right],\label{eq:shared-vterms-v2}\\
\mathbb V^{(3)}_{\rm Cl}
&=
\frac{N_M-1}{N_M^2} \left(\operatorname{Tr}[R_{AA'B}^{\otimes n}(\rho\otimes\rho\otimes\sigma)]\right.\nonumber\\
&\qquad\qquad\; + \left.\operatorname{Tr}[R_{ABB'}^{\otimes n}(\rho\otimes\sigma\otimes\sigma)]\right),
\label{eq:shared-vterms-v3}
\end{align}
\end{subequations}
where $\Cl$ denotes Clifford sampling and
\begin{align}
R_{AA'B}
=
\frac32\bF_{AA'}+\frac12\bF_{AB}+\frac12\bF_{A'B}-\bI,\\
R_{ABB'}
=
\frac32\bF_{BB'}+\frac12\bF_{AB}+\frac12\bF_{AB'}-\bI.
\end{align}
The genuinely Clifford-specific input enters only through the fourth term, 
where the failure of the $4$-design property becomes essential.
To evaluate $\mathbb{V}^{(4)}_{\Cl}$, we invoke the recent complete characterization of the Clifford commutant \cite{bittel2025commutant}. For one qubit, uniform Clifford conjugation maps $Z$ to a uniformly random element of $\{\pm X,\pm Y,\pm Z\}$, so the single-qubit Clifford fourth moment operator can be written exactly as
\begin{align}
\label{eq:cliff-r4-main}
R_{4,\Cl}
=
-\bI
+\frac12\left(\bF_{A_1B_1}+\bF_{A_2B_2}\right)
+\frac32\,\Omega_4^{(1)},
\end{align}
where
\begin{align}
\Omega_4^{(1)}:=\frac12\left(\bI^{\otimes 4}+X^{\otimes 4}+Y^{\otimes 4}+Z^{\otimes 4}\right)
\end{align}
is the one qubit Clifford commutant generator beyond the permutation sector.  
Since the local Cliffords are sampled independently, we obtain
\begin{align}\label{eq:Clifford-vterms-v4}
\mathbb{V}^{(4)}_{\Cl}
=
\left(\frac{N_M-1}{N_M}\right)^2
\tr\left[R_{4,\Cl}^{\otimes n}
(\rho\ox\sigma\ox\rho\ox\sigma)\right].
\end{align}
We prove the following state independent upper bound.
\begin{theorem}
\label{thm:clifford-fourth-term-bound}
For every $n\ge 1$, the maximum over $n$-qubit states $\rho$ and $\sigma$ satisfies
\begin{align}\label{eq:Clifford-fourth-term-bound}
\left(\frac32\right)^n \le
\max_{\rho,\sigma}
\tr\left[R_{4,\Cl}^{\otimes n}
(\rho\ox\sigma\ox\rho\ox\sigma)\right]
\le
\left(\frac74\right)^n.
\end{align}
Identical pure product stabilizer states attain the lower endpoint, which is the maximum within the family of identical pure stabilizer inputs. The maximum over arbitrary states remains open.
\end{theorem}

In view of Eq.~\eqref{eq:Clifford-vterms-v4}, Thm.~\ref{thm:clifford-fourth-term-bound} immediately implies
\begin{align}
\label{eq:cliff-v4-main}
\mathbb{V}^{(4)}_{\Cl}
\le\left(\frac{N_M-1}{N_M}\right)^2\left(\frac74\right)^n.
\end{align}
Together with the state-independent upper bounds
\begin{align}\label{eq:v2-and-v3-upper-bounds}
\mathbb{V}^{(2)}_{\Cl}\le \frac{3^n}{N_M^2},
\qquad
\mathbb{V}^{(3)}_{\Cl}\le 2\,\frac{N_M-1}{N_M^2}\left(\frac74\right)^n,
\end{align}
this yields the following sufficient-copy corollary.

\begin{corollary}
\label{lem:clifford-complexity}
Let $\wh X_f$ be the estimator in Eq.~\eqref{eq:estimator} built from the unique kernel in Thm.~\ref{thm:unique-kernel}. 
Under single-qubit Clifford sampling, 
for arbitrary $n$-qubit states $\rho$ and $\sigma$, the guarantee
\begin{align}
\bP(|\wh X_f-\tr\left[\rho\sigma\right]|\ge \ve)\le \delta
\end{align}
can be achieved for $0<\ve,\delta<1$ with a per-party copy count of
\begin{align}
\label{eq:cliff-sample-main}
N_{\star,\Cl} = \cO\!\left(\frac{\sqrt{5.25^n}}{\delta\ve^2}\right).
\end{align}
Taking $N_M=\lceil(12/7)^{n/2}\rceil$ gives this bound.
\end{corollary}

Ref.~\cite[Thm.~13]{zheng2025distributed} established an $O(\sqrt{4.5^n})$ sufficient-copy bound for identical pure product stabilizer states under local Clifford sampling at fixed $\ve$ and $\delta$. Cor.~\ref{lem:clifford-complexity} gives a state-independent guarantee. Whether the product-stabilizer scaling suffices for arbitrary inputs remains open.
For identical pure stabilizer inputs, the sharper fourth-moment bound in Thm.~\ref{thm:clifford-fourth-term-bound} gives $\cO((9/2)^{n/2}/(\delta\ve^2))$ copies with $N_M=\lceil2^{n/2}\rceil$, by the same balance of variance terms.

\subsection{DIPE with single-qubit Haar ensemble}

The single-qubit Haar ensemble is an exact unitary $t$-design for every positive integer $t$~\cite{webb2016clifford}. 
Accordingly, the first three terms in the variance decomposition coincide with those of the single-qubit Clifford ensemble in Eq.~\eqref{eq:shared-vterms}:
\begin{align}
\mathbb{V}^{(r)}_{\mathrm{H}}=\mathbb{V}^{(r)}_{\Cl},\qquad r=1,2,3.
\end{align}
The only nontrivial change occurs in the fourth term, where the full $4$-design property becomes relevant. All proofs in this subsection 
are collected in Appx.~\ref{app:haar}.

The fourth term for the Haar ensemble is
\begin{align}
\mathbb{V}^{(4)}_{\mathrm{H}}
=
\left(\frac{N_M-1}{N_M}\right)^2
\tr\left[R_{4,\mathrm{H}}^{\otimes n}(\rho\ox\sigma\ox\rho\ox\sigma)\right],
\label{eq:haar-vterms-v4}
\end{align}
where $R_{4,\mathrm{H}}$ is the single-qubit Haar fourth moment operator whose explicit form is derived in Appx.~\ref{app:haar}.
We show that the single-qubit Haar fourth moment 
is the Haar average of conjugates of the single-qubit Clifford fourth moment.
Taking tensor products yields the desired $n$-qubit identity.

\begin{theorem}
\label{thm:haar-local-twirl-of-clifford}
For $n$ qubits,
\begin{align}\label{eq:haar-local-twirl-of-clifford-identity}
R_{4,\mathrm{H}}^{\otimes n}
=
\bE_{\bm W\sim\nu^{\otimes n}}\!\left[
(\bm W^{\otimes 4})^\dagger R_{4,\Cl}^{\otimes n}\bm W^{\otimes 4}
\right],
\end{align}
where $\nu$ denotes the single-qubit Haar measure and $\bm W:=W_1\ox\cdots\ox W_n$ with $W_1,\dots,W_n$ sampled independently from $\nu$.
Consequently, for arbitrary $n$-qubit states $\rho$ and $\sigma$,
\begin{align}\label{eq:haar-local-twirl-of-clifford-induced-bound}
\max_{\rho,\sigma}
\tr\left[R_{4,\mathrm{H}}^{\otimes n}(\rho\ox\sigma\ox\rho\ox\sigma)\right]
\le\left(\frac74\right)^n.
\end{align}
\end{theorem}

Thm.~\ref{thm:haar-local-twirl-of-clifford} expresses each Haar fourth moment as a $\nu^{\otimes n}$ average of the Clifford fourth moment over the local unitary orbit. Reaching the Clifford maximum would require every state pair on that orbit to be a Clifford maximizer. Appx.~\ref{appx:sec:universal-upper-bound-not-achievable} derives this condition and gives the exact Haar maxima $6/5$ for one qubit and $29/20$ for two qubits, both strictly below the corresponding Clifford product-stabilizer values $(3/2)^n$.
Since the first three variance terms also coincide, the same Chebyshev-based sufficient-copy upper bound given in Cor.~\ref{lem:clifford-complexity} follows as well. A direct bound on the Haar fourth moment gives a smaller copy count.

\begin{theorem}
\label{lem:haar-complexity}
Under single-qubit Haar sampling, for every pair of $n$-qubit states $\rho,\sigma$,
\begin{align}
\label{eq:haar-uniform-fourth-main}
\tr\!\left[R_{4,\mathrm H}^{\otimes n}(\rho\ox\sigma\ox\rho\ox\sigma)\right]
\le \left(\frac43\right)^n .
\end{align}
For $0<\ve,\delta<1$, the estimator in Eq.~\eqref{eq:estimator} with the kernel of Thm.~\ref{thm:unique-kernel} achieves additive error $\ve$ with failure probability at most $\delta$ using
\begin{align}
\label{eq:haar-sample-main}
N_{\star,\mathrm H} = \cO\!\left(\frac{2^n}{\delta\ve^2}\right)
\end{align}
copies per party. The allocation $N_M=\lceil(3/2)^n\rceil$ works for all input pairs.
\end{theorem}

Since a Clifford product stabilizer attains $(3/2)^n$, Eq.~\eqref{eq:haar-uniform-fourth-main} separates the two fourth-moment maxima for every $n\ge1$.
The sufficient copy bound in Eq.~\eqref{eq:haar-sample-main} is linear in the Hilbert-space dimension $2^n$ at fixed accuracy and confidence. It is not tight, as Appx.~\ref{app:haar-complexity-proof} gives the sharper all-input rate in Eq.~\eqref{eq:haar-sharper-bound}. For product pairs and certain entangled inputs, Appx.~\ref{app:haar-state-dependent} gives a smaller sufficient bound by choosing the number of shots per setting for the input pair. The optimal worst-case copy rates of the two ensembles remain open.

\section{DIPE with independent single-qubit classical shadow}
\label{sec:independent-pauli-shadow-main}

For comparison with the shared-randomness LRM protocols analyzed above, 
we also consider an independent local Pauli shadow estimator,
\begin{align}
\wh g_{\mathrm P}
:=
\tr\left[
\left(\frac1N\sum_{i=1}^N \hat{\rho}_i\right)
\left(\frac1N\sum_{j=1}^N \hat{\sigma}_j\right)
\right],
\end{align}
where $\hat{\rho}_i$ and $\hat{\sigma}_j$ are independent Pauli shadow
snapshots of $\rho$ and $\sigma$~\cite{huang2020shadows}, respectively. 
This estimator was experimentally investigated in~\cite{Zhu2022}, but without a variance-based copy analysis. 
We derive its Chebyshev-based worst-case sufficient copy bound below; the proof is deferred to Appx.~\ref{app:independent-pauli-shadow}.

\begin{proposition}
\label{prop:independent-pauli-shadow-complexity-main}
For arbitrary $n$-qubit states $\rho$ and $\sigma$, the estimator
$\wh g_{\mathrm P}$ is unbiased for $\tr\left[\rho\sigma\right]$. Moreover, 
\begin{align}
\bP\left(|\wh g_{\mathrm P}-\tr\left[\rho\sigma\right]|\ge \ve\right)\le \delta
\end{align}
can be achieved for any $\ve,\delta\in(0,1)$ using
\begin{align}
\label{eq:independent-pauli-shadow-sample-main}
N_{\star,\mathrm P} = \cO\!\left(\max\left\{
\frac{\sqrt{7.5^n}}{\sqrt{\delta}\,\ve},
\frac{2^n}{\delta\ve^2}
\right\}\right)
\end{align}
snapshots per party. For fixed $\ve$ and $\delta$, the leading
large $n$ scale of this sufficient bound is $\cO(\sqrt{7.5^n})$.
\end{proposition}

The Pauli shadow is constructed from the single-qubit Clifford ensemble~\cite{huang2020shadows}.
Using the single-qubit Haar ensemble gives the same variance for this estimator
and hence the same Chebyshev bound, because the proof of
Prop.~\ref{prop:independent-pauli-shadow-complexity-main} uses only the
single-qubit projective $3$-design identities. At fixed accuracy and confidence,
the sufficient-copy bound for independent shadows grows faster with $n$ than the
bounds derived here for shared Clifford or Haar settings.

\begin{figure*}[t]
\centering
\subfloat[Single-qubit Clifford ensemble.\label{fig:num-clifford}]{%
\includegraphics[width=0.48\textwidth]{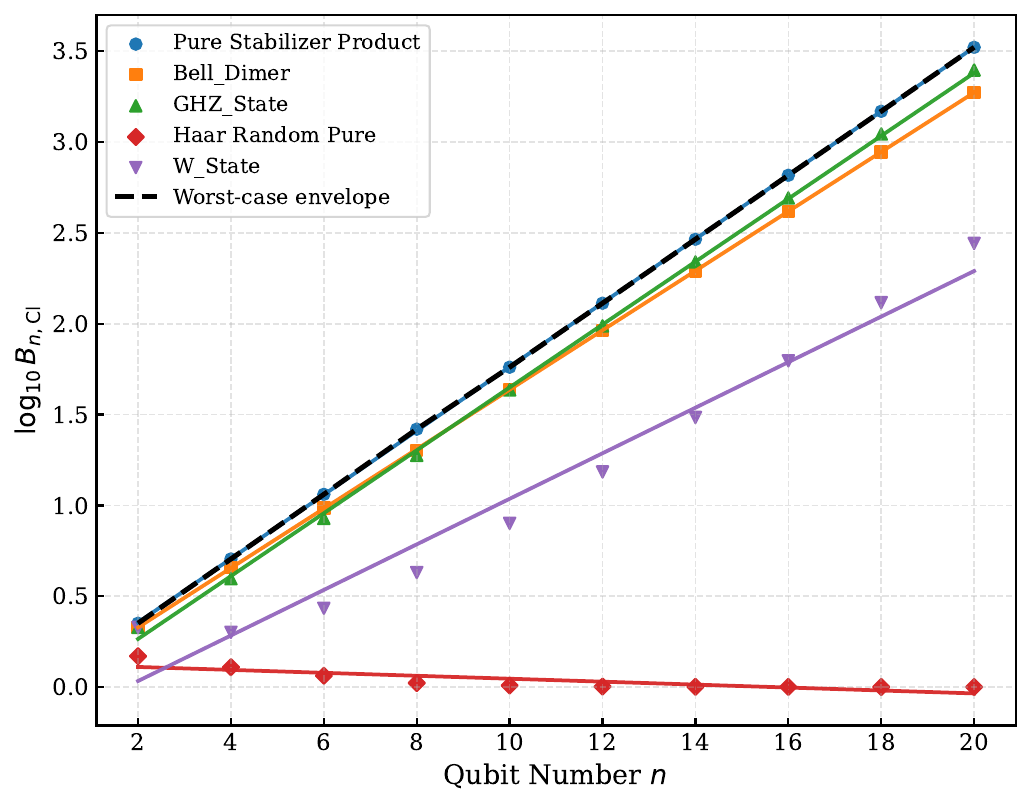}}
\hfill
\subfloat[Single-qubit Haar ensemble.\label{fig:num-haar}]{%
\includegraphics[width=0.48\textwidth]{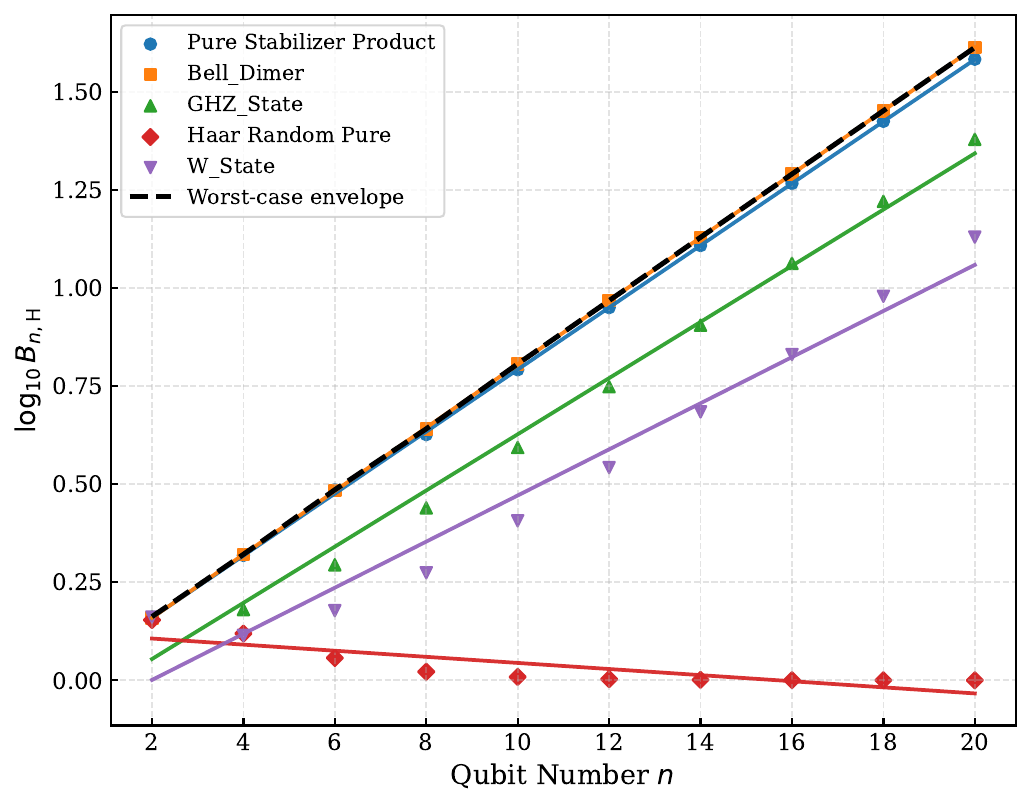}}
\caption{\raggedright\textbf{Fourth-moment quantities for DIPE with LRM.}
Both panels use the same benchmark state families and plot
$\log_{10}B_{n,\mathrm E}$ for (a) single-qubit Clifford and
(b) single-qubit Haar sampling. The markers give numerical values, and the fitted curves summarize finite-size trends. For the identical pure inputs shown here, $\bE_U[\mu_U]=\tr[\rho^2]=1$, so $B_{n,\mathrm E}=\bE_U[\mu_U^2]\ge1$. Portions of the fits below zero on the logarithmic axis therefore lie outside the physical range. The curves labeled ``Worst-case envelope'' are reference curves for the plotted state families.}
\label{fig:numerical-leading-scales}
\end{figure*}

\section{Numerical Simulations}\label{sec:numerical-simulation}

Our numerical study compares the fourth-moment quantities for local
Clifford and local Haar randomized measurements on the same benchmark families.
In every numerical experiment, Alice and Bob use identical input states, $\rho=\sigma$. For a Haar-random pure input, both parties use the same realization of the state.
We consider the product state $\ket{+}^{\otimes n}$, Haar-random pure states,
GHZ states, Bell-dimer states, and $W$ states. These examples probe the Clifford
stabilizer benchmark and the competition under local Haar sampling between fourth-moment enhancement and
subsystem-purity loss. More detailed decompositions of the variance
coefficients, together with variations in purity and entanglement
structure, are deferred to Appx.~\ref{app:numerics}.

For a benchmark pair $(\rho,\sigma)$ and unitary ensemble $\mathrm E\in\{\Cl,\mathrm H\}$,
we plot
\begin{align}
B_{n,\mathrm E}(\rho,\sigma)
&:=
\tr\left[R_{4,\mathrm E}^{\otimes n}
(\rho\ox\sigma\ox\rho\ox\sigma)\right],
\end{align}
which determines the fourth variance term through
\begin{align}
\mathbb{V}^{(4)}_{\mathrm E}
=
\left(\frac{N_M-1}{N_M}\right)^2
B_{n,\mathrm E}(\rho,\sigma).
\end{align}
Thus Fig.~\ref{fig:numerical-leading-scales} isolates the
$N_M$-independent ensemble-dependent fourth-moment coefficient.

In Fig.~\ref{fig:numerical-leading-scales}(a) the product stabilizer state lies
on the upper envelope of the benchmark families and attains the analytical
fourth-moment maximum within the stabilizer family.
For even $n$, Appx.~\ref{app:haar-state-dependent} gives $B_{n,\mathrm H}=(29/20)^{n/2}$ for Bell dimers and $(6/5)^n$ for pure product inputs. Their ratio is $(145/144)^{n/2}>1$. The markers at $n=6$ in Fig.~\ref{fig:numerical-leading-scales}(b) do not resolve this ordering. The Haar-random pure
state curve is a typical state benchmark. As the Hilbert-space dimension grows
exponentially with $n$, random sampling is unlikely to hit the
structured extremal regions; typical $n$-qubit Haar-random states are highly
entangled and have strongly suppressed local correlations. This behavior is
consistent with the discussion in Appx.~\ref{app:haar-state-dependent}, where entanglement can
enhance $B_{n,\mathrm H}$ though it may simultaneously reduce
second-moment contributions.

\section{Conclusions}\label{sec:conclusion}

We established worst-case sample complexity bounds for distributed inner product
estimation with local randomized measurements and identified the kernel selected
by the worst-case framework. The kernel optimization can be reduced to Hamming
distance kernels, and unbiasedness in this class fixes the unique kernel. 
The Clifford and Haar analyses share the same first three variance terms
because both single-qubit ensembles are unitary $3$-designs, but their
fourth-moment structures are different and must be treated separately. In the
single-qubit Clifford case, the fourth-moment maximum within the pure stabilizer
family is $(3/2)^n$, attained by identical product stabilizer states. For arbitrary
inputs the fourth-moment coefficient is at most $(7/4)^n$ for Clifford sampling and $(4/3)^n$ for Haar sampling.
These bounds give sufficient-copy scalings of $\cO(\sqrt{5.25^n})$ and $\cO(\sqrt{4^n})$, respectively, at fixed accuracy and confidence. Appx.~\ref{app:haar-complexity-proof} sharpens the Haar copy bound further.
The local twirling relation compares the Haar fourth moment with the Clifford one.
For comparison, we showed that the independent local Pauli-shadow protocol has a sufficient-copy
bound of $\cO(\sqrt{7.5^n})$ at fixed accuracy and confidence. 
The main results in this work are summarized in Tab.~\ref{tab:summarization} for reference and comparison.

Two questions remain open. The first is whether the Clifford fourth-moment
maximum over arbitrary inputs equals the stabilizer maximum.
The second question is to derive matching lower bounds for DIPE with LRM,
which would determine whether the present local protocols are worst-case optimal.

\section*{Acknowledgements}

This work was supported by 
the National Key R\&D Program of China (Grant No. 2022YFF0712800),
the Innovation Research Foundation of NUDT (Grant No. 202501-YJRC-ZZ-004), and
the National Natural Science Foundation of China (Grant No. 12504584).
\bibliographystyle{apsrev4-2}
\bibliography{references}

\makeatletter
\newcommand{\appendixtitle}[1]{\gdef\@title{#1}}
\makeatother

\makeatletter%
\newcommand{\appendixmaketitle}{%
\begin{center}%
\vspace{0.4in}%
{\large \@title \par}%
\end{center}%
\par%
}%
\makeatother%

\makeatletter
\newcommand{\appendixtableofcontents}{%
\begingroup
\setcounter{tocdepth}{4}%
\@starttoc{atoc}%
\endgroup}
\makeatother

\newcommand{\appendixtocdivider}{%
\par\medskip
\noindent\hbox to \linewidth{%
\leaders\hbox{\rule[0.6ex]{1pt}{0.4pt}}\hfill
\hspace{0.8em}\textsc{Table of Contents}\hspace{0.8em}%
\leaders\hbox{\rule[0.6ex]{1pt}{0.4pt}}\hfill}%
\par}

\makeatletter
\newcounter{subsubsubsection}[subsubsection]
\renewcommand{\thesubsubsubsection}{\thesubsubsection.\arabic{subsubsubsection}}
\providecommand{\theHsubsubsubsection}{}
\renewcommand{\theHsubsubsubsection}{\theHsubsubsection.\arabic{subsubsubsection}}
\providecommand*{\toclevel@subsubsubsection}{4}
\newcommand{\subsubsubsection}[1]{%
\refstepcounter{subsubsubsection}%
\paragraph*{\thesubsubsubsection\space #1}}
\newcommand*\l@subsubsubsection[2]{\@dottedtocline{4}{4.4em}{3.0em}{#1}{#2}}
\makeatother

\newcommand{\appsection}[1]{%
\section{#1}%
\addcontentsline{atoc}{section}{\protect\numberline{\thesection}#1}}

\newcommand{\appsubsection}[1]{%
\subsection{#1}%
\addcontentsline{atoc}{subsection}{\protect\numberline{\thesubsection}#1}}

\newcommand{\appsubsubsection}[1]{%
\subsubsection{#1}%
\addcontentsline{atoc}{subsubsection}{\protect\numberline{\thesubsubsection}#1}}

\newcommand{\appsubsubsubsection}[1]{%
\subsubsubsection{#1}}

\setcounter{secnumdepth}{4}
\appendix
\widetext
\newpage

\appendixtitle{\bf 
Supplemental Material for\\``\thetitle''}
\appendixmaketitle
\vspace{0.1in}

The appendices provide supporting details for the main text (MT) and are organized as follows. Appx.~\ref{app:general} derives the general variance decomposition for DIPE with shared local randomized measurements. Appx.~\ref{app:hamming-sufficiency} proves the reduction to Hamming-distance kernels and the uniqueness of the unbiased kernel in that class. Appx.~\ref{app:clifford} analyzes single-qubit Clifford sampling, including the common first three variance terms, the Clifford fourth-moment bound, and the resulting sufficient-copy corollary. Appx.~\ref{app:haar} analyzes single-qubit Haar sampling, proves the local twirling comparison with the Clifford fourth moment and the all-input Haar bound, explains why the comparison bound is not saturated, and establishes Prop.~\ref{conj:Haar-complexity} for the specified families. Appx.~\ref{app:independent-pauli-shadow} treats the independent single-qubit Pauli-shadow protocol. Finally, Appx.~\ref{app:numerics} records the numerical details and additional validation supporting the discussion in Sec.~\ref{sec:numerical-simulation} of MT.

\appendixtocdivider
{%
\appendixtableofcontents%
}%

\appsection{General framework of DIPE}
\label{app:general}

This appendix section complements Sec.~\ref{sec:general-framework} of MT.

\appsubsection{Measurement protocol and variance decomposition}\label{app:measurement-protocol}

We first spell out the local randomized measurement protocol used throughout 
the main text. It has the same block structure as the global shadow protocol in
Ref.~\cite{anshu2022distributed}, but each shared random unitary is restricted to be a
product of independently sampled single-qubit unitaries.

\begin{algorithm}[H]
\caption{DIPE with shared single-qubit randomized measurements}
\label{alg:lrm-dipe}
\begin{algorithmic}[1]
\REQUIRE $N_U$ unitary samples, $N_M$ repeated shots per unitary sample, a single-qubit ensemble $\cU_1$, 
and an unbiased kernel $f:\{0,1\}^n\times\{0,1\}^n\to\bR$.
\ENSURE An estimate of $\tr\left[\rho\sigma\right]$.
\FOR{$l=1,\ldots,N_U$}
\STATE Alice and Bob sample the shared product unitary
$U_l=\bigotimes_{\ell=1}^n U_{l,\ell}$, where $U_{l,\ell}\sim\cU_1$ independently.
\STATE Alice applies $U_l$ to $N_M$ fresh copies of $\rho$ and measures in the computational basis, obtaining $\bm s_{l,1},\ldots,\bm s_{l,N_M}$.
\STATE Bob applies $U_l$ to $N_M$ fresh copies of $\sigma$ and measures in the computational basis, obtaining $\bm t_{l,1},\ldots,\bm t_{l,N_M}$.
\ENDFOR
\STATE Return
\[
\wh X_f
:=
\frac{1}{N_U N_M^2}
\sum_{l=1}^{N_U}
\sum_{j,k=1}^{N_M}
f(\bm s_{l,j},\bm t_{l,k}).
\]
\end{algorithmic}
\end{algorithm}

Based on Alg.~\ref{alg:lrm-dipe}, we derive the four-term variance
decomposition for an arbitrary unbiased kernel $f$. For each unitary block
$l$, define
\begin{align}
X_M^{(l)}
:=
\frac{1}{N_M^2}\sum_{j,k=1}^{N_M}f(\bm s_{l,j},\bm t_{l,k}).
\end{align}
Then
\begin{align}
\wh X_f=\frac1{N_U}\sum_{l=1}^{N_U}X_M^{(l)}.
\end{align}
The random variables $X_M^{(1)},\ldots,X_M^{(N_U)}$ are i.i.d.; hence, if
$X_M$ denotes a generic block,
\begin{align}
\Var[\wh X_f]=\frac{1}{N_U}\Var[X_M].
\end{align}
It remains to compute the variance of this single-block estimator. For one
unitary sample, write
\begin{align}
X_M=\frac{1}{N_M^2}\sum_{j,k=1}^{N_M}f(\bm s_j,\bm t_k),
\end{align}
where, conditioned on the sampled unitary $U$, the strings $\bm s_1,\dots,\bm s_{N_M}$ are i.i.d. according to the measurement distribution of $\rho$, and $\bm t_1,\dots,\bm t_{N_M}$ are i.i.d. according to that of $\sigma$. Let $J_B(U)$ and $J_A(U)$ denote the contributions from sharing Bob's outcome and Alice's outcome, respectively. Set
\begin{align}
\mu_U &:= \bE[f(\bm s,\bm t)\mid U], \\
J_B(U) &:= \bE_{\bm t\mid U}\!\left[\left(\bE_{\bm s\mid U} f(\bm s,\bm t)\right)^2\right], \\
J_A(U) &:= \bE_{\bm s\mid U}\!\left[\left(\bE_{\bm t\mid U} f(\bm s,\bm t)\right)^2\right].
\end{align}
When the kernel $f$ is unbiased, we have $\bE[X_M]=\tr\left[\rho\sigma\right]$, and hence
\begin{align}
\Var[X_M]
=
\bE_U\!\left[\bE[X_M^2\mid U]\right]-\tr\left[\rho\sigma\right]^2.
\end{align}
It therefore suffices to evaluate the conditional second moment. Expanding the square gives
\begin{align}
\bE[X_M^2\mid U]
=
\frac{1}{N_M^4}
\sum_{j,k,j',k'=1}^{N_M}
\bE\!\left[f(\bm s_j,\bm t_k)f(\bm s_{j'},\bm t_{k'})\mid U\right].
\end{align}
The quadruple of indices $(j,k,j',k')$ falls into four disjoint coincidence classes.

\vspace*{0.1in}
\textbf{1. Full coincidence: $(j,k)=(j',k')$.}
There are $N_M^2$ such terms, and each contributes
\begin{align}
\bE\!\left[f(\bm s_j,\bm t_k)^2\mid U\right]
=
\bE_{\bm s,\bm t\mid U}[f(\bm s,\bm t)^2].
\end{align}
Hence this class contributes
\begin{align}
\frac{N_M^2}{N_M^4}\bE_{\bm s,\bm t\mid U}[f(\bm s,\bm t)^2]
=
\frac{1}{N_M^2}\bE_{\bm s,\bm t\mid U}[f(\bm s,\bm t)^2].
\end{align}

\vspace*{0.1in}
\textbf{2. Coincidence in Bob's index only: $j\neq j'$ and $k=k'$.}
There are $N_M^2(N_M-1)$ such terms. Conditioned on $U$, the variables $\bm s_j$ and $\bm s_{j'}$ are independent, while the same outcome $\bm t_k$ is shared. Therefore each term equals
\begin{align}
\bE_{\bm s,\bm t,\bm s'\mid U}[f(\bm s,\bm t)f(\bm s',\bm t)]
=
J_B(U).
\end{align}
Thus this class contributes
\begin{align}
\frac{N_M^2(N_M-1)}{N_M^4}J_B(U)
=
\frac{N_M-1}{N_M^2}J_B(U).
\end{align}

\vspace*{0.1in}
\textbf{3. Coincidence in Alice's index only: $j=j'$ and $k\neq k'$.}
Again there are $N_M^2(N_M-1)$ such terms. Now the shared variable is $\bm s_j$, while $\bm t_k$ and $\bm t_{k'}$ are independent conditioned on $U$. Hence each term equals
\begin{align}
\bE_{\bm s,\bm t,\bm t'\mid U}[f(\bm s,\bm t)f(\bm s,\bm t')]
=
J_A(U),
\end{align}
and the total contribution is
\begin{align}
\frac{N_M^2(N_M-1)}{N_M^4}J_A(U)
=
\frac{N_M-1}{N_M^2}J_A(U).
\end{align}

\vspace*{0.1in}
\textbf{4. No coincidence: $j\neq j'$ and $k\neq k'$.}
There are $N_M^2(N_M-1)^2$ such terms. In this case all sampled outcomes are independent conditioned on $U$, so each term factorizes as
\begin{align}
\bE_{\bm s,\bm t,\bm s',\bm t'\mid U}[f(\bm s,\bm t)f(\bm s',\bm t')]
=
\mu_U^2.
\end{align}
Therefore this class contributes
\begin{align}
\frac{N_M^2(N_M-1)^2}{N_M^4}\mu_U^2
=
\left(\frac{N_M-1}{N_M}\right)^2\mu_U^2.
\end{align}

Collecting the four classes, we obtain
\begin{align}
\bE[X_M^2\mid U]
=
\frac{1}{N_M^2}\bE_{\bm s,\bm t\mid U}[f(\bm s,\bm t)^2]
+
\frac{N_M-1}{N_M^2}\bigl(J_B(U)+J_A(U)\bigr)
+
\left(\frac{N_M-1}{N_M}\right)^2\mu_U^2.
\end{align}
Substituting this into the formula for $\Var[X_M]$ yields
\begin{align}
\Var[X_M]
=
-\tr\left[\rho\sigma\right]^2
+
\frac{1}{N_M^2}\bE_U\bE_{\bm s,\bm t\mid U}[f(\bm s,\bm t)^2]
+
\frac{N_M-1}{N_M^2}\bE_U\!\left[J_B(U)+J_A(U)\right]
+
\left(\frac{N_M-1}{N_M}\right)^2\bE_U[\mu_U^2].
\end{align}
When $f(\bm s,\bm t)=f(\bm t,\bm s)$, exchanging the states in $\Xi_U(\rho,\sigma)$ also exchanges the shared outcome. Thus $J_B(U)=\Xi_U(\rho,\sigma)$ and $J_A(U)=\Xi_U(\sigma,\rho)$, which gives Eq.~\eqref{eq:general-vterms} of MT.

The terms $\mathbb{V}^{(2)}$, $\mathbb{V}^{(3)}$, and $\mathbb{V}^{(4)}$ correspond, respectively, to the coincidence patterns $(j,k)=(j',k')$, exactly one of the equalities $j=j'$ or $k=k'$ holding, and $j\neq j'$, $k\neq k'$.
This decomposition makes the moment requirements transparent: $\mathbb{V}^{(2)}$ depends only on the second moment of the unitary ensemble, $\mathbb{V}^{(3)}$ on the third moment, and $\mathbb{V}^{(4)}$ on the fourth moment.
Since the $N_U$ unitary blocks are independent, $\Var[\wh X_f]=\Var[X_M]/N_U$.

\appsubsection{Median-of-means estimator}
\label{app:median-of-means}

Fix $N_M$ and let $V$ be an upper bound on $\Var[X_M]$. Divide $N_U=Kb$ independent unitary settings into $K$ disjoint groups of $b$ settings, keeping all shots from each setting together. Let $Y_j$ be the mean of the $b$ one-setting estimates in group $j$. For an unbiased kernel, $\bE[Y_j]=\theta:=\tr[\rho\sigma]$ and $\Var[Y_j]\le V/b$. Choosing $b=\max\{1,\lceil 4V/\ve^2\rceil\}$ gives $\bP(|Y_j-\theta|\ge\ve)\le 1/4$ by Chebyshev's inequality.

Choose $K$ to be the smallest positive odd integer with $K\ge 8\ln(1/\delta)$. If the median of the $Y_j$ differs from $\theta$ by at least $\ve$, at least half the groups have failed. Hoeffding's inequality for the independent failure indicators gives
\[
\bP\!\left(\left|\operatorname{med}(Y_1,\ldots,Y_K)-\theta\right|\ge\ve\right)
\le \exp\!\left[-2K(1/2-1/4)^2\right]=e^{-K/8}\le\delta .
\]
The cost per party is $KbN_M\le K N_M(1+4V/\ve^2)$, so $\cO\!\left(N_M(1+V/\ve^2)\log(2/\delta)\right)$ copies suffice for any such ensemble. For $0<\ve,\delta<1$, the uniform variance bounds and shot allocations used in Cor.~\ref{lem:clifford-complexity} and Thm.~\ref{lem:haar-complexity} give sufficient copy counts of
\[
\cO\!\left(\frac{\sqrt{5.25^n}\log(2/\delta)}{\ve^2}\right)
\quad\text{(Clifford)},\qquad
\cO\!\left(\frac{\sqrt{4^n}\log(2/\delta)}{\ve^2}\right)
\quad\text{(Haar)}.
\]

\appsection{Optimal Hamming distance kernels}\label{app:hamming-sufficiency}

This appendix section proves the results presented in Sec.~\ref{sec:optimal-kernel} of MT.

\appsubsection{Proof of Lem.~\ref{lem:hamming-sufficiency}}

Let $\Gamma=(S_2)^n\rtimes S_n$ act on pairs $(\bm s,\bm t)$ by simultaneous relabeling and permutation. For $\gamma\in\Gamma$, write
\begin{align}
f^\gamma(\bm s,\bm t):=f(\gamma(\bm s),\gamma(\bm t)).
\end{align}
Then $f^{\mathrm{sym}}=\bE_{\gamma\in\Gamma} f^\gamma$. Two pairs lie in the same orbit under $\Gamma$ if and only if they have the same Hamming distance. Hence the orbit average in Eq.~\eqref{eq:sym-kernel} depends only on $D(\bm s,\bm t)$.

The covariance property implies that for every $\gamma\in\Gamma$ there exists a unitary $W_\gamma$ such that
\begin{align}
\wh X_{f^\gamma}(\rho,\sigma)\stackrel{d}{=}\wh X_f(W_\gamma\rho W_\gamma^\dagger,W_\gamma\sigma W_\gamma^\dagger).
\end{align}
Assume now that $f$ is admissible. Taking expectations in the preceding distributional identity gives, for all states $\rho$ and $\sigma$,
\begin{align}
\tr\left[(\rho\ox\sigma)\bar\Omega_{f^\gamma}\right]
&=
\tr\left[((W_\gamma\rho W_\gamma^\dagger)\ox(W_\gamma\sigma W_\gamma^\dagger))\bar\Omega_f\right]  \\
&=
\tr\left[(\rho\ox\sigma)(W_\gamma^\dagger\ox W_\gamma^\dagger)
\bar\Omega_f(W_\gamma\ox W_\gamma)\right]  \\
&=
\tr\left[(\rho\ox\sigma)\bF_{\boldsymbol{AB}}\right],
\end{align}
where the last step uses $\bar\Omega_f=\bF_{\boldsymbol{AB}}$ and the invariance of the full swap under simultaneous conjugation by $W_\gamma\ox W_\gamma$. Since this holds for all $\rho$ and $\sigma$, $\bar\Omega_{f^\gamma}=\bF_{\boldsymbol{AB}}$, and by linearity $\bar\Omega_{f^{\mathrm{sym}}}=\bF_{\boldsymbol{AB}}$. Thus symmetrization maps admissible kernels to admissible Hamming-distance kernels.

Variance is convex under averaging of estimators, so for any fixed $(\rho,\sigma)$,
\begin{align}
\Var_{\rho,\sigma}[\wh X_{f^{\mathrm{sym}}}]
\le \bE_{\gamma\in\Gamma}\Var_{\rho,\sigma}[\wh X_{f^\gamma}].
\end{align}
Taking the supremum and using the bijectivity of $(\rho,\sigma)\mapsto (W_\gamma\rho W_\gamma^\dagger,W_\gamma\sigma W_\gamma^\dagger)$ gives
\begin{align}
\sup_{\rho,\sigma}\Var_{\rho,\sigma}[\wh X_{f^{\mathrm{sym}}}]
\le
\bE_{\gamma\in\Gamma}
\sup_{\rho,\sigma}\Var_{\rho,\sigma}[\wh X_{f^\gamma}]
=
\sup_{\rho,\sigma}\Var_{\rho,\sigma}[\wh X_f],
\end{align}
where the last step again uses the covariance identity above. Let
\begin{align}
V(f):=\sup_{\rho,\sigma}\Var_{\rho,\sigma}[\wh X_f].
\end{align}
For every admissible $f$, we have constructed an admissible Hamming-distance kernel $f^{\mathrm{sym}}$ with $V(f^{\mathrm{sym}})\le V(f)$, so
\begin{align}
\inf_{g=g(D):\bar\Omega_g=\bF_{\boldsymbol{AB}}}V(g)
\le
\inf_{f:\bar\Omega_f=\bF_{\boldsymbol{AB}}}V(f).
\end{align}
The reverse inequality is immediate because admissible Hamming-distance kernels are a subclass of all admissible kernels. This proves Eq.~\eqref{eq:minimax-reduction} and the lemma.

\appsubsection{Proof of Thm.~\ref{thm:unique-kernel}}
\label{app:unique-kernel}

For a single qubit let $P_s=U^\dagger\ketbra{s}U$. Writing $L=U^\dagger ZU$ and $P_s=(\bI+(-1)^sL)/2$, the identity $\bE[L^{\otimes2}]=(2\bF-\bI)/3$ gives $\bE_U\big[\sum_s P_s\ox P_s\big]=(\bI+\bF)/3$. Since $\sum_{s,t}P_s\ox P_t=\bI$, the unequal-outcome sum averages to $(2\bI-\bF)/3$. Give each unequal pair a factor $z$. The averaged generating operator on $n$ sites is then
\[
3^{-n}\bigotimes_{\ell=1}^n\big[(1+2z)\bI_\ell+(1-z)\bF_\ell\big].
\]
For a fixed partial swap $\bF_S$, $|S|=k$, the coefficient of $z^d$ is $3^{-n}[z^d](1+2z)^{n-k}(1-z)^k$. For a Hamming distance kernel $f(\bm s,\bm t)=g(D(\bm s,\bm t))$, the averaged operator can be expanded as
\begin{align}
\bar\Omega_f=\sum_{k=0}^n\alpha_k(g)\sum_{|S|=k}\bF_S,
\end{align}
where
\begin{align}
\alpha_k(g)=\frac{1}{3^n}\sum_{d=0}^n g(d)K_d(k;n,3)
\end{align}
is the $q=3$ Krawtchouk transform. Here
\begin{align}
K_d(k;n,3)=\sum_{j=0}^d(-1)^j2^{d-j}\binom{k}{j}\binom{n-k}{d-j}.
\end{align}

Unbiasedness requires $\bar\Omega_f=\bF_{\boldsymbol{AB}}$, so all partial swap sectors must vanish:
\begin{align}
\alpha_k(g)=0,\qquad k=0,1,\dots,n-1,
\qquad
\alpha_n(g)=1.
\end{align}
More generally, the condition
\begin{align}
\bar\Omega_f\in\spn{\bI_{\boldsymbol{AB}},\bF_{\boldsymbol{AB}}}
\end{align}
is equivalent to
\begin{align}
\alpha_k(g)=0,\qquad k=1,\dots,n-1.
\end{align}
Since the Krawtchouk transform
\begin{align}
g\longmapsto \left(\alpha_0(g),\alpha_1(g),\dots,\alpha_n(g)\right)
\end{align}
is invertible on the $(n+1)$-dimensional space of functions $g:\{0,\dots,n\}\to\mathbb{R}$, the subspace cut out by the $n-1$ homogeneous constraints $\alpha_k(g)=0$ for $1\le k\le n-1$ has dimension exactly $2$.

To construct a basis, we find the kernels whose averaged operators are the identity and the full swap. Let $Q_0=(\bI+\bF)/3$ and $Q_1=(2\bI-\bF)/3$ denote the one-site averages for equal and unequal outcomes derived above. Then
\[
\bI=Q_0+Q_1,\qquad \bF=2Q_0-Q_1.
\]
For the identity, equal and unequal pairs both have weight $1$, so multiplying over all $n$ sites gives the constant kernel $g_{\mathrm I}(d)=1$. For the swap, each equal pair has weight $2$ and each unequal pair has weight $-1$. Two strings at Hamming distance $d$ agree at $n-d$ sites and differ at $d$ sites, so the full-swap kernel is $g_{\mathrm F}(d)=2^{n-d}(-1)^d=2^n(-1/2)^d$.

For $g_{\mathrm I}(d)=1$, evaluating $3^{-n}(1+2z)^{n-k}(1-z)^k$ at $z=1$ gives $\alpha_k(g_{\mathrm I})=\delta_{k,0}$. For $g_{\mathrm F}(d)=2^n(-1/2)^d$, evaluating the same polynomial at $z=-1/2$ and multiplying by $2^n$ gives $\alpha_k(g_{\mathrm F})=\delta_{k,n}$. These two solutions of the homogeneous constraints are linearly independent and therefore form a basis. Therefore every Hamming distance kernel satisfying $\bar\Omega_f\in\spn{\bI_{\boldsymbol{AB}},\bF_{\boldsymbol{AB}}}$ has the form
\begin{align}
g(d)=a+b(-1)^d2^{n-d}.
\end{align}
To impose unbiasedness, we require $\bar\Omega_f=\bF_{\boldsymbol{AB}}$, equivalently
\begin{align}
\alpha_0(g)=0,\qquad \alpha_n(g)=1.
\end{align}
It follows that unbiasedness fixes $a=0$ and $b=1$, which proves Eq.~\eqref{eq:product-kernel-main}.

\appsection{DIPE with single-qubit Clifford ensemble}
\label{app:clifford}

Because the single-qubit Clifford group forms an exact unitary $3$-design \cite{webb2016clifford}, this section begins with the common analysis of $\mathbb{V}^{(1)}$, $\mathbb{V}^{(2)}$, and $\mathbb{V}^{(3)}$, which also applies to the single-qubit Haar ensemble. We then turn to the genuinely Clifford-specific contribution, namely the fourth term, prove the fourth-term bounds stated in Thm.~\ref{thm:clifford-fourth-term-bound}, and finally combine these ingredients to derive Cor.~\ref{lem:clifford-complexity}.

\appsubsection{Calculations of the first three terms}\label{app:clifford-first-three-terms}

We derive the common first three variance formulas stated in Eqs.~\eqref{eq:shared-vterms-v1}--\eqref{eq:shared-vterms-v3} of MT from the general formulas in Eq.~\eqref{eq:general-vterms} of MT, and then establish the corresponding upper bounds. Since the single-qubit Clifford group is an exact unitary $3$-design, these identities apply equally to single-qubit Clifford and single-qubit Haar sampling.

For one qubit, let
\begin{align}
P_b(V):=V^\dagger\ketbra{b}V=\frac12\left(\bI+(-1)^bA(V)\right),
\qquad
A(V):=\bm n(V)\cdot\bsigma,
\end{align}
where $V$ is sampled from the single-qubit Clifford ensemble and $\bm n(V)$ is uniformly distributed over the six Pauli directions. Then $A(V)^2=\bI$, $\bE_V[A(V)]=0$, and
\begin{align}
\label{eq:clifford-second-moment-A}
\bE_V[A(V)\ox A(V)]=\frac{2\bF-\bI}{3}.
\end{align}
For the one qubit kernel $h(a,b)=3\delta_{a,b}-1$, define
\begin{align}\label{eq:O-1-V}
O_1(V)
:=
\sum_{a,b\in\{0,1\}}h(a,b)\,P_a(V)\ox P_b(V)
=
\frac12\bI+\frac32A(V)\ox A(V).
\end{align}

\vspace*{0.1in}
\textbf{The first term $\mathbb{V}^{(1)}_{\Cl}$.}
This term is given by 
$\mathbb{V}^{(1)}_{\Cl} = - \tr\left[\rho\sigma\right]^2$ and is ensemble independent. 
It holds that $\mathbb{V}^{(1)}_{\Cl}\le 0$
with equality attained whenever $\rho$ and $\sigma$ have orthogonal supports, for example for
\begin{align}
\rho=\ketbra{0}^{\otimes n},
\qquad
\sigma=\ketbra{1}^{\otimes n}.
\end{align}

\vspace*{0.1in}
\textbf{The second term $\mathbb{V}^{(2)}_{\Cl}$.}
We derive the operator form of $\mathbb{V}^{(2)}_{\Cl}$ from Eq.~\eqref{eq:general-vterms-v2}.
Let $U=\bigotimes_{\ell=1}^n U_\ell$ be a product Clifford sample. Since
\begin{align}
f(\bm s,\bm t)^2=\prod_{\ell=1}^n h(s_\ell,t_\ell)^2,
\end{align}
the general formula in Eq.~\eqref{eq:general-vterms-v2} becomes
\begin{align}
\mathbb{V}^{(2)}_{\Cl}
&=
\frac{1}{N_M^2}
\tr\left[
\bE_U
\sum_{\bm s,\bm t\in\{0,1\}^n}
f(\bm s,\bm t)^2
\left(U^\dagger\ketbra{\bm s}U\right)\ox\left(U^\dagger\ketbra{\bm t}U\right)
(\rho\ox\sigma)
\right]
\nonumber\\
&=
\frac{1}{N_M^2}
\tr\left[
\bE_U\!\left[\bigotimes_{\ell=1}^n O_1(U_\ell)^2\right]
(\rho\ox\sigma)
\right]
\nonumber\\
&=
\frac{1}{N_M^2}
\tr\left[
\left(\bE_V[O_1(V)^2]\right)^{\otimes n}
(\rho\ox\sigma)
\right].
\end{align}
Now
\begin{align}
O_1(V)^2
=
\left(\frac12\bI+\frac32A(V)\ox A(V)\right)^2
=
\frac52\bI+\frac32A(V)\ox A(V),
\end{align}
because $A(V)^2=\bI$. Averaging and using Eq.~\eqref{eq:clifford-second-moment-A} gives
\begin{align}
\bE_V[O_1(V)^2]
=
\frac52\bI+\frac32\cdot\frac{2\bF-\bI}{3}
=
2\bI+\bF.
\end{align}
Hence
\begin{align}
\mathbb{V}^{(2)}_{\Cl}
=
\frac{1}{N_M^2}
\tr\left[(2\bI+\bF)^{\otimes n}(\rho\ox\sigma)\right],
\end{align}
which is Eq.~\eqref{eq:shared-vterms-v2}. Since the eigenvalues of $2\bI+\bF$ are $3,3,3,1$, we have $\|2\bI+\bF\|_\infty=3$ and therefore
\begin{align}
\mathbb{V}^{(2)}_{\Cl}\le \frac{3^n}{N_M^2}.
\end{align}
This is the bound stated in Eq.~\eqref{eq:v2-and-v3-upper-bounds}, 
and is attained by any identical pure product state; for example,
$\rho=\sigma=\ketbra{0}^{\otimes n}$ gives equality because $(2\bI+\bF)\ket{00}=3\ket{00}$.

\vspace*{0.1in}
\textbf{The third term $\mathbb{V}^{(3)}_{\Cl}$.}
We derive the operator form of $\mathbb{V}^{(3)}_{\Cl}$ directly from
Eq.~\eqref{eq:general-vterms-v3}, namely
\begin{align}
\mathbb{V}^{(3)}_{\Cl}
=
\frac{N_M-1}{N_M^2}
\bE_U\!\left[\Xi_U(\rho,\sigma)+\Xi_U(\sigma,\rho)\right].
\end{align}
Define
\begin{align}
K_b(V)
:=
\sum_{a\in\{0,1\}} h(a,b)\,P_a(V)
=
\frac12\bI+\frac32(-1)^bA(V).
\end{align}
Since $h(a,b)=h(b,a)$, the same operators also satisfy
\begin{align}
K_a(V)=\sum_{b\in\{0,1\}} h(a,b)\,P_b(V).
\end{align}
For a product sample $U=\bigotimes_{\ell=1}^n U_\ell$ and a bit string $\bm b\in\{0,1\}^n$, set
\begin{align}
K_{\bm b}(U):=\bigotimes_{\ell=1}^n K_{b_\ell}(U_\ell),
\qquad
P_{\bm b}(U):=U^\dagger\ketbra{\bm b}U.
\end{align}
Then the first ordered contraction is
\begin{align}
\bE_U[\Xi_U(\rho,\sigma)]
&=
\bE_U\sum_{\bm b\in\{0,1\}^n}
\tr\left[P_{\bm b}(U)\sigma\right]
\left(
\sum_{\bm a\in\{0,1\}^n}
f(\bm a,\bm b)\,
\tr\left[P_{\bm a}(U)\rho\right]
\right)^2
\nonumber\\
&=
\bE_U\sum_{\bm b\in\{0,1\}^n}
\tr\left[K_{\bm b}(U)\rho\right]^2
\tr\left[P_{\bm b}(U)\sigma\right]
\nonumber\\
&=
\tr\left[
R_{AA'B}^{\otimes n}
(\rho\ox\rho\ox\sigma)
\right],
\end{align}
where
\begin{align}
R_{AA'B}
:=
\bE_V\sum_{b=0}^1
K_b(V)_{A}\ox K_b(V)_{A'}\ox P_b(V)_{B}.
\end{align}
Likewise, the second ordered contraction is
\begin{align}
\bE_U[\Xi_U(\sigma,\rho)]
&=
\bE_U\sum_{\bm a\in\{0,1\}^n}
\tr\left[P_{\bm a}(U)\rho\right]
\left(
\sum_{\bm b\in\{0,1\}^n}
f(\bm a,\bm b)\,
\tr\left[P_{\bm b}(U)\sigma\right]
\right)^2
\nonumber\\
&=
\bE_U\sum_{\bm a\in\{0,1\}^n}
\tr\left[P_{\bm a}(U)\rho\right]
\tr\left[K_{\bm a}(U)\sigma\right]^2
\nonumber\\
&=
\tr\left[
R_{ABB'}^{\otimes n}
(\rho\ox\sigma\ox\sigma)
\right],
\end{align}
with
\begin{align}
R_{ABB'}
:=
\bE_V\sum_{a=0}^1
P_a(V)_{A}\ox K_a(V)_{B}\ox K_a(V)_{B'}.
\end{align}
Therefore
\begin{align}
\mathbb{V}^{(3)}_{\Cl}
=
\frac{N_M-1}{N_M^2}\left(
\tr\left[R_{AA'B}^{\otimes n}(\rho\ox\rho\ox\sigma)\right]
+
\tr\left[R_{ABB'}^{\otimes n}(\rho\ox\sigma\ox\sigma)\right]
\right),
\label{eq:appx-shared-vterms-v3}
\end{align}
which is Eq.~\eqref{eq:shared-vterms-v3} in MT.

Now we compute state independent upper bound on $\mathbb{V}^{(3)}_{\Cl}$.
We first compute the one qubit operators $R_{AA'B}$ and $R_{ABB'}$ explicitly. 
Note that
\begin{align}
\sum_{b=0}^1 K_b(V)\ox K_b(V)\ox P_b(V)
=
\frac14\bI^{\otimes 3}
+\frac94A(V)\ox A(V)\ox\bI
+\frac34A(V)\ox\bI\ox A(V)
+\frac34\bI\ox A(V)\ox A(V).
\end{align}
Averaging and using Eq.~\eqref{eq:clifford-second-moment-A} gives
\begin{align}
R_{AA'B}
=
-\bI+\frac32\bF_{AA'}+\frac12\bF_{AB}+\frac12\bF_{A'B},
\end{align}
and, by the same calculation with the roles of $A$ and $B$ exchanged,
\begin{align}
R_{ABB'}
=
-\bI+\frac32\bF_{BB'}+\frac12\bF_{AB}+\frac12\bF_{AB'}.
\end{align}
A direct diagonalization yields the eigenvalues
\begin{align}
\{3/2,3/2,3/2,3/2,0,0,-2,-2\}
\end{align}
for both $R_{AA'B}$ and $R_{ABB'}$. 

The operator norms of $R_{AA'B}$ and $R_{ABB'}$ are both $2$.
This yields a loose universal upper bound $2^n$ on $\tr\left[R_{AA'B}^{\otimes n}(\rho\ox\rho\ox\sigma)\right]$. 
We can further tighten this bound by exploiting the state structure.

\begin{proposition}
\label{prop:third-term-bound}
For any $\rho,\sigma\in\density{(\bC^2)^{\otimes n}}$, it holds that
\begin{align}
\tr\left[R_{AA'B}^{\otimes n}(\rho_{\boldsymbol{A}}\ox\rho_{\boldsymbol{A}'}\ox\sigma_{\boldsymbol{B}})\right]
&\le
\left(\frac74\right)^n,
\label{eq:third-term-ordered-bound-AAB}\\
\tr\left[R_{ABB'}^{\otimes n}(\rho_{\boldsymbol{A}}\ox\sigma_{\boldsymbol{B}}\ox\sigma_{\boldsymbol{B}'})\right]
&\le
\left(\frac74\right)^n.
\label{eq:third-term-ordered-bound-ABB}
\end{align}
Consequently,
\begin{align}
\mathbb{V}^{(3)}_{\Cl}
\le
2\,\frac{N_M-1}{N_M^2}\left(\frac74\right)^n.
\label{eq:third-term-bound}
\end{align}
\end{proposition}

\begin{proof}
It suffices to prove Eq.~\eqref{eq:third-term-ordered-bound-AAB}; the proof of
Eq.~\eqref{eq:third-term-ordered-bound-ABB} is identical after exchanging the
roles of the two repeated registers.

Let
\begin{align}
\Pi_-^{AA'}:=\frac{\bI_{AA'}-\bF_{AA'}}{2}
\end{align}
be the antisymmetric projector on the one-qubit pair $AA'$, and let
$\Pi_{\mathrm{sym}}^{AA'B}$ be the projector onto the fully symmetric subspace
of the three one-qubit registers $AA'B$. Since the local Hilbert space is
two-dimensional, the fully antisymmetric three-qubit representation is absent,
and
\begin{align}
\Pi_{\mathrm{sym}}^{AA'B}
=
\frac13\left(
\bF_{AA'}+\bF_{AB}+\bF_{A'B}
\right).
\end{align}
Using
\begin{align}
R_{AAB}
=
-\bI+\frac32\bF_{AA'}+\frac12\bF_{AB}+\frac12\bF_{A'B},
\end{align}
we obtain the spectral decomposition
\begin{align}
R_{AAB}
=
\frac32\,\Pi_{\mathrm{sym}}^{AA'B}
-
2\,\Pi_-^{AA'}\ox\bI_B.
\label{eq:R-AAB-spectral-decomp}
\end{align}
The two projectors in Eq.~\eqref{eq:R-AAB-spectral-decomp} have orthogonal
supports, because the fully symmetric subspace is symmetric under
$A\leftrightarrow A'$, whereas $\Pi_-^{AA'}\ox\bI_B$ is antisymmetric under
$A\leftrightarrow A'$. Therefore
\begin{align}
\abs{R_{AAB}}
=
\frac32\,\Pi_{\mathrm{sym}}^{AA'B}
+
2\,\Pi_-^{AA'}\ox\bI_B.
\end{align}
Moreover,
\begin{align}
\Pi_{\mathrm{sym}}^{AA'B}
\le
\bI_{AA'B}-\Pi_-^{AA'}\ox\bI_B,
\end{align}
and hence
\begin{align}
\abs{R_{AAB}}
\le
\frac32\bI_{AA'B}
+
\frac12\,\Pi_-^{AA'}\ox\bI_B.
\label{eq:R-AAB-abs-bound}
\end{align}

Now set
\begin{align}
Q_\ell:=\Pi_-^{A_\ell A'_\ell}\ox\bI_{B_\ell},
\qquad \ell=1,\ldots,n.
\end{align}
For the positive operator
$X:=\rho_{\boldsymbol{A}}\ox\rho_{\boldsymbol{A}'}\ox\sigma_{\boldsymbol{B}}$,
we have
\begin{align}
\abs{
\tr\left[R_{AAB}^{\otimes n}X\right]
}
\le
\tr\left[
\abs{R_{AAB}^{\otimes n}}X
\right]
=
\tr\left[
\abs{R_{AAB}}^{\otimes n}X
\right]
\le
\tr\left[
\bigotimes_{\ell=1}^n
\left(
\frac32\bI_{A_\ell A'_\ell B_\ell}
+
\frac12 Q_\ell
\right)
X
\right],
\label{eq:third-term-abs-expansion-start}
\end{align}
where the last inequality uses Eq.~\eqref{eq:R-AAB-abs-bound} and positivity
under tensor products.

Expanding the product in Eq.~\eqref{eq:third-term-abs-expansion-start} gives
\begin{align}
\abs{
\tr\left[R_{AAB}^{\otimes n}X\right]
}
\le
\sum_{S\subseteq[n]}
\left(\frac32\right)^{n-|S|}
\left(\frac12\right)^{|S|}
\tr\left[
\Pi_{-,S}^{\boldsymbol{A}\boldsymbol{A}'}
(\rho_{\boldsymbol{A}}\ox\rho_{\boldsymbol{A}'})
\right],
\label{eq:third-term-subset-expansion}
\end{align}
where
\begin{align}
\Pi_{-,S}^{\boldsymbol{A}\boldsymbol{A}'}
:=
\bigotimes_{\ell\in S}\Pi_-^{A_\ell A'_\ell}
\end{align}
with identity operators on the complementary tensor factors. If
$\rho_S:=\tr_{S^c}[\rho]$ denotes the reduced state of $\rho$ on the subsystem
$S$, then
\begin{align}
\tr\left[
\Pi_{-,S}^{\boldsymbol{A}\boldsymbol{A}'}
(\rho_{\boldsymbol{A}}\ox\rho_{\boldsymbol{A}'})
\right]
=
\tr\left[
\left(\bigotimes_{\ell\in S}\Pi_-^{A_\ell A'_\ell}\right)
(\rho_S\ox\rho_S)
\right].
\label{eq:third-term-reduced-state}
\end{align}

Let $k:=|S|$. The projector
$\bigotimes_{\ell\in S}\Pi_-^{A_\ell A'_\ell}$ is a rank-one maximally
entangled projector across the bipartition $A_S|A'_S$. Thus there exists a
unitary $W_S$ on $A'_S$ such that, with
\begin{align}
\ket{\Phi_S}
:=
2^{-k/2}\sum_{\bm x\in\{0,1\}^k}
\ket{\bm x}_{A_S}\ket{\bm x}_{A'_S},
\end{align}
one has
\begin{align}
\bigotimes_{\ell\in S}\Pi_-^{A_\ell A'_\ell}
=
(\bI_{A_S}\ox W_S)\ketbra{\Phi_S}(\bI_{A_S}\ox W_S^\dagger).
\end{align}
Therefore,
\begin{align}
\tr\left[
\left(\bigotimes_{\ell\in S}\Pi_-^{A_\ell A'_\ell}\right)
(\rho_S\ox\rho_S)
\right]
=
2^{-k}
\tr\left[
\rho_S\,W_S\rho_S^T W_S^\dagger
\right],
\label{eq:third-term-max-ent-overlap}
\end{align}
where the transpose is taken in the computational basis. Since
$W_S\rho_S^T W_S^\dagger$ is again a density operator,
\begin{align}
\tr\left[
\rho_S\,W_S\rho_S^T W_S^\dagger
\right]
\le 1.
\end{align}
Combining this with Eq.~\eqref{eq:third-term-reduced-state} gives
\begin{align}
\tr\left[
\Pi_{-,S}^{\boldsymbol{A}\boldsymbol{A}'}
(\rho_{\boldsymbol{A}}\ox\rho_{\boldsymbol{A}'})
\right]
\le
2^{-|S|}.
\label{eq:third-term-singlet-overlap-bound}
\end{align}

Substituting Eq.~\eqref{eq:third-term-singlet-overlap-bound} into
Eq.~\eqref{eq:third-term-subset-expansion}, we obtain
\begin{align}
\abs{
\tr\left[R_{AAB}^{\otimes n}X\right]
}
\le
\sum_{S\subseteq[n]}
\left(\frac32\right)^{n-|S|}
\left(\frac12\right)^{|S|}
2^{-|S|}
=
\sum_{k=0}^n
\binom{n}{k}
\left(\frac32\right)^{n-k}
\left(\frac14\right)^k
=
\left(\frac74\right)^n.
\end{align}
This proves Eq.~\eqref{eq:third-term-ordered-bound-AAB}. The same argument with
$\rho\ox\sigma\ox\sigma$ and the antisymmetric projector on $BB'$ proves
Eq.~\eqref{eq:third-term-ordered-bound-ABB}. Finally,
Eq.~\eqref{eq:third-term-bound} follows immediately from
Eq.~\eqref{eq:shared-vterms-v3}.
\end{proof}

\appsubsection{Calculations of the fourth term}

We compute the analytic forms of the single-qubit Clifford fourth-moment operator $R_{4,\Cl}$ in Eq.~\eqref{eq:cliff-r4-main} of MT.

\vspace*{0.1in}
\textbf{Clifford fourth-moment operator.}
For a product Clifford unitary $C=\bigotimes_{\ell=1}^n C_\ell$, define
\begin{align}
O(C)
:=
\sum_{\bm s,\bm t\in\{0,1\}^n}
f(\bm s,\bm t)
(C^\dagger\ketbra{\bm s}C)
\ox
(C^\dagger\ketbra{\bm t}C).
\end{align}
Since $f(\bm s,\bm t)=\prod_{\ell=1}^n h(s_\ell,t_\ell)$ and $C$ is a product unitary, this factorizes as
\begin{align}
O(C)=\bigotimes_{\ell=1}^n O_1(C_\ell),
\end{align}
where $O_1(\cdot)$ is defined in Eq.~\eqref{eq:O-1-V}.
Therefore
\begin{align}
\bE_{\bm s,\bm t\mid C}[f(\bm s,\bm t)]
=
\sum_{\bm s,\bm t}
f(\bm s,\bm t)
\tr\left[C^\dagger\ketbra{\bm s}C\,\rho\right]
\tr\left[C^\dagger\ketbra{\bm t}C\,\sigma\right]
=
\tr\left[O(C)(\rho\ox\sigma)\right],
\end{align}
and hence
\begin{align}
\bE_C\!\left[
\left(\bE_{\bm s,\bm t\mid C}[f(\bm s,\bm t)]\right)^2
\right]
=
\tr\left[
\bE_C\!\left[O(C)_{A_1B_1}\ox O(C)_{A_2B_2}\right]
(\rho\ox\sigma\ox\rho\ox\sigma)
\right].
\end{align}
By independence of the single-qubit Clifford unitaries, the averaged operator factorizes as
\begin{align}
\bE_C\!\left[O(C)_{A_1B_1}\ox O(C)_{A_2B_2}\right]
=
R_{4,\Cl}^{\otimes n},
\end{align}
where the single-qubit local Clifford fourth-moment operator $R_{4,\Cl}$ is defined as
\begin{align}
R_{4,\Cl}
:=
\bE_D\!\left[O_1(D)_{A_1B_1}\ox O_1(D)_{A_2B_2}\right]
\end{align}
with $D$ a uniformly random single-qubit Clifford unitary. 
For one qubit, uniform Clifford conjugation maps $Z$ to a uniformly random Pauli from $\{\pm X,\pm Y,\pm Z\}$. Hence
\begin{align}
O_1(C)=\frac12\bI+\frac32P(C)\ox P(C),
\end{align}
where $P(C)$ is uniformly distributed over $\{\pm X,\pm Y,\pm Z\}$.
Averaging $O_1(C)\ox O_1(C)$ over the one-qubit Clifford group yields the
one-qubit operator $R_{4,\Cl}$ appearing above; explicitly,
\begin{align}
\label{eq:cliff-r4-app}
R_{4,\Cl} = \mathbb{E}_{\Cl}[O_1(C)\otimes O_1(C)] = 
-\bI
+\frac12\left(\bF_{A_1B_1}+\bF_{A_2B_2}\right)
+\frac32\,\Omega_4^{(1)},
\end{align}
where
\begin{align}
\Omega_4^{(1)}:=\frac12\left(\bI^{\otimes 4}+X^{\otimes 4}+Y^{\otimes 4}+Z^{\otimes 4}\right).
\end{align}
Note that the appearance of $\Omega_4^{(1)}$ reflects the nontrivial $k=4$ Clifford commutant described in Ref.~\cite{bittel2025commutant}.

\appsubsection{Proof of Thm.~\ref{thm:clifford-fourth-term-bound}}
\label{app:clifford-fourth-term-bound}

We now prove Thm.~\ref{thm:clifford-fourth-term-bound} of MT
via the following three steps:

\begin{enumerate}
  \item Sec.~\ref{appendix:sec:optimality-identical-pure} (containing Prop.~\ref{prop:cliff-v4-identical-pure}) 
      shows that the maximum in Eq.~\eqref{eq:Clifford-fourth-term-bound} 
      is attained by a pair of identical pure states.
  \item Sec.~\ref{appendix:sec:extremizers-pure} (containing Prop.~\ref{prop:cliff-pure-stabilizer})  
      bounds the fourth moment by $(7/4)^n$ for arbitrary inputs using the ordered third moment in Prop.~\ref{prop:third-term-bound}.
  \item Finally, Sec.~\ref{appx:sec:product-pure-stabilizer-upper-bound} (containing Prop.~\ref{prop:cliff-stabilizer-bound}) proves the sharp upper bound $(3/2)^n$ over the stabilizer family, with equality attained by pure product stabilizer states.
\end{enumerate}

We now establish these three ingredients in turn.

\appsubsubsection{Fourth-moment optimum occurs for identical pure states}
\label{appendix:sec:optimality-identical-pure}

\begin{proposition}
\label{prop:cliff-v4-identical-pure}
It holds that
\begin{align}
\max_{\rho,\sigma\in\density{(\bC^2)^{\otimes n}}}
\tr\left[
R_{4,\Cl}^{\otimes n}
(\rho\ox\sigma\ox\rho\ox\sigma)
\right]
=
\max_{\substack{\ket{\psi}\in(\bC^2)^{\otimes n}\\ \langle\psi|\psi\rangle=1}}
\tr\left[
R_{4,\Cl}^{\otimes n}\psi^{\ox4}
\right].
\end{align}
\end{proposition}

\begin{proof}
Define 
\begin{align}
    \mu_C(\rho,\sigma) := \bE_{\bm s,\bm t\mid C}[f(\bm s,\bm t)].
\end{align}
Then 
\begin{align}
\tr\left[R_{4,\Cl}^{\otimes n}(\rho\ox\sigma\ox\rho\ox\sigma)\right]
=
\bE_C[\mu_C(\rho,\sigma)^2].
\end{align}
For every fixed product Clifford sample $C$ and every fixed state $\sigma$, 
the quantity $\mu_C(\rho,\sigma)$ is affine in $\rho$. 
Therefore, for $\rho_1,\rho_2\in\density{(\bC^2)^{\otimes n}}$ and $\lambda\in[0,1]$,
\begin{align}
&\; \tr\left[
R_{4,\Cl}^{\otimes n}
((\lambda\rho_1+(1-\lambda)\rho_2)\ox\sigma
\ox(\lambda\rho_1+(1-\lambda)\rho_2)\ox\sigma)
\right] \\
=&\;
\bE_C\!\left[
\left(
\lambda\mu_C(\rho_1,\sigma)+(1-\lambda)\mu_C(\rho_2,\sigma)
\right)^2
\right] \\
\leq&\;
\lambda\tr\left[
R_{4,\Cl}^{\otimes n}(\rho_1\ox\sigma\ox\rho_1\ox\sigma)
\right]
+(1-\lambda)\tr\left[
R_{4,\Cl}^{\otimes n}(\rho_2\ox\sigma\ox\rho_2\ox\sigma)
\right],
\end{align}
because $x\mapsto x^2$ is convex. Hence, for fixed $\sigma$, the fourth-moment
objective is convex in $\rho$ on the compact convex set of density operators,
so its maximum is attained at a pure state. The same argument applies to the
second argument. Therefore a maximizing pair may be chosen pure.

It remains to show that the two pure states may be chosen identical. For a fixed product Clifford sample $C$ and state $\tau$, define
\begin{align}
p_C^\tau(\bm s):=\tr\left[C^\dagger\ketbra{\bm s}C\,\tau\right],
\qquad
\wh p_C^\tau(S):=\sum_{\bm s\in\{0,1\}^n}
(-1)^{\sum_{\ell\in S}s_\ell}p_C^\tau(\bm s).
\end{align}
Here $p_C^\tau$ is the computational basis measurement distribution of $\tau$ after applying $C$, and $\wh p_C^\tau(S)$ is its Walsh coefficient on $S\subseteq[n]$. Expanding the product kernel in Eq.~\eqref{eq:product-kernel-main} in the Walsh basis gives
\begin{align}
\mu_C(\rho,\sigma)
=
\sum_{S\subseteq[n]}
\frac{3^{|S|}}{2^n}\,
\wh p_C^\rho(S)\,\wh p_C^\sigma(S),
\end{align}
where $p_C^\tau$ and $\wh p_C^\tau$ are defined above. Since all weights $3^{|S|}/2^n$ are nonnegative, the weighted Cauchy--Schwarz inequality and the arithmetic-geometric mean inequality imply
\begin{align}
\mu_C(\rho,\sigma)^2
\leq
\mu_C(\rho,\rho)\,\mu_C(\sigma,\sigma)
\leq
\frac{\mu_C(\rho,\rho)^2+\mu_C(\sigma,\sigma)^2}{2}.
\end{align}
Averaging over $C$ yields
\begin{align}
\tr\left[R_{4,\Cl}^{\otimes n}(\rho\ox\sigma\ox\rho\ox\sigma)\right]
\le
\frac{
\tr\left[R_{4,\Cl}^{\otimes n}\rho^{\ox4}\right]
+
\tr\left[R_{4,\Cl}^{\otimes n}\sigma^{\ox4}\right]
}{2}
\le
\max\left\{
\tr\left[R_{4,\Cl}^{\otimes n}\rho^{\ox4}\right],
\tr\left[R_{4,\Cl}^{\otimes n}\sigma^{\ox4}\right]
\right\}.
\end{align}
Applying this to a maximizing pure state pair shows that at least one identical pure state pair also attains the global maximum.
\end{proof}

\appsubsubsection{A state-independent fourth-moment bound}
\label{appendix:sec:extremizers-pure}

We record the Pauli formula for the optimization in Prop.~\ref{prop:cliff-v4-identical-pure} and derive a state-independent upper bound from the conditional measurement probabilities.

Let $\cP_n$ be the set of $n$-fold tensor products of elements of $\cP_1=\{I,X,Y,Z\}$. For one qubit, define
\begin{align}
\kappa_1(a,b):=
\begin{cases}
1, & a=b=I,\\
1, & a=I,\ b\in\{X,Y,Z\}\ \text{or}\ b=I,\ a\in\{X,Y,Z\},\\
3, & a=b\in\{X,Y,Z\},\\
0, & a,b\in\{X,Y,Z\},\ a\neq b,
\end{cases}
\end{align}
and for $P,Q\in\cP_n$ set
\begin{align}
\kappa_n(P,Q):=\prod_{\ell=1}^n \kappa_1(P_\ell,Q_\ell).
\end{align}
The exact Pauli basis formula for the Clifford fourth moment can be written as
\begin{align}
\tr\left[R_{4,\Cl}^{\otimes n}(\rho\ox\sigma\ox\rho\ox\sigma)\right]
=
\frac{1}{4^n}
\sum_{P,Q\in\cP_n}
\kappa_n(P,Q)
\wh\rho(P)\wh\sigma(P)\wh\rho(Q)\wh\sigma(Q).
\label{eq:cliff-kappa-form}
\end{align}

Specialize to identical pure inputs, $\rho=\sigma=\proj{\psi}$. Define
\begin{align}
X_\psi(P):=\wh\rho(P)^2,\qquad \Xi_\psi(P):=2^{-n}X_\psi(P).
\end{align}
Here $\Xi_\psi$ is the Pauli squared characteristic distribution of $\ket{\psi}$~\cite[Eq.~(A10)]{hinsche2025efficient}. By Pauli orthogonality, purity, and the bound $|\wh\rho(P)|\le 1$, one has
\begin{align}
0\le \Xi_\psi(P)\le 2^{-n},
\qquad
\sum_{P\in\cP_n}\Xi_\psi(P)=1.
\end{align}
Hence
\begin{align}
\tr\left[R_{4,\Cl}^{\otimes n}(\rho^{\ox4})\right]
=
\sum_{P,Q\in\cP_n}\Xi_\psi(P)\Xi_\psi(Q)\kappa_n(P,Q).
\label{eq:cliff-pure-quadratic}
\end{align}
Since the one qubit kernel matrix associated with $\kappa_1$ is positive semidefinite, so is $\kappa_n=\kappa_1^{\otimes n}$; therefore Eq.~\eqref{eq:cliff-pure-quadratic} is a convex quadratic functional of the probability vector $\Xi_\psi$.

\begin{proposition}
\label{prop:cliff-pure-stabilizer}
For arbitrary $n$-qubit states $\rho$ and $\sigma$,
\begin{align}
\tr\left[R_{4,\Cl}^{\otimes n}(\rho\ox\sigma\ox\rho\ox\sigma)\right]
\le \left(\frac74\right)^n.
\end{align}
\end{proposition}

\begin{proof}
For a fixed product Clifford sample $C$ and Bob's outcome $\bm t$, let $g_C(\bm t):=\bE_{\bm s\mid C}[f(\bm s,\bm t)]$, where $\bm s$ is sampled from Alice's measurement distribution. Conditional independence of the two parties' outcomes gives $\mu_C(\rho,\sigma)=\bE_{\bm t\mid C}[g_C(\bm t)]$. Jensen's inequality yields
\begin{align}
\mu_C(\rho,\sigma)^2 \le \bE_{\bm t\mid C}[g_C(\bm t)^2]=\Xi_C(\rho,\sigma).
\end{align}
The last equality follows by using two independent outcomes from Alice for the same $\bm t$. Averaging over $C$ and applying Eq.~\eqref{eq:third-term-ordered-bound-AAB} gives
\begin{align}
\bE_C[\mu_C(\rho,\sigma)^2]
\le \bE_C[\Xi_C(\rho,\sigma)] 
=\tr\left[R_{AA'B}^{\otimes n}(\rho\ox\rho\ox\sigma)\right]
\le \left(\frac74\right)^n.
\end{align}
The left-hand side is the fourth-moment coefficient in Eq.~\eqref{eq:Clifford-vterms-v4}.
\end{proof}

\appsubsubsection{Sharp bound within the stabilizer family}
\label{appx:sec:product-pure-stabilizer-upper-bound}

We now restrict Eq.~\eqref{eq:cliff-kappa-form} to identical pure stabilizer inputs to obtain the lower endpoint in Eq.~\eqref{eq:Clifford-fourth-term-bound}. Specializing the exact Pauli basis formula to identical inputs gives
\begin{align}
\tr\left[R_{4,\Cl}^{\otimes n}(\psi^{\ox 4})\right]
=
\frac{1}{4^n}
\sum_{\substack{P,Q\in\cP_n\\P\sim Q}}
3^{c(P,Q)}\,\wh\psi(P)^2\,\wh\psi(Q)^2,
\end{align}
where $\wh\psi(P):=\tr\left[\psi P\right]$ is the Pauli coefficient of $\psi$, $P\sim Q$ means that on each qubit the local Paulis are either identical or one of them is identity, and $c(P,Q)$ counts the positions where the two local Paulis coincide nontrivially.

\begin{proposition}
\label{prop:cliff-stabilizer-bound}
Let $\psi=\proj{\psi}$ be a pure stabilizer state on $n$ qubits. Then
\begin{align}
\tr\left[R_{4,\Cl}^{\otimes n}(\psi^{\ox 4})\right]
\le
\left(\frac32\right)^n.
\label{eq:cliff-stabilizer-bound}
\end{align}
Equality is attained by product stabilizer states, equivalently by identical local Pauli eigenstates.
\end{proposition}

\begin{proof}
Let
\begin{align}
\psi=\frac{1}{2^n}\sum_{P\in S}\ve(P)P,
\end{align}
where $S\subset\cP_n$ is the phase-free stabilizer support of $\psi$ and
$\ve(P)\in\{\pm1\}$. Since $\psi$ is a pure stabilizer state, $S$ has
cardinality $2^n$ and may be identified with an $n$-dimensional totally
isotropic linear subspace of $(\mathbb F_2^2)^n$. Thus
\begin{align}
\tr\left[R_{4,\Cl}^{\otimes n}(\psi^{\ox 4})\right]
=
\frac{1}{4^n}
\sum_{P,Q\in S}
\prod_{j=1}^n\kappa(P_j,Q_j),
\label{eq:cliff-stabilizer-exact-repaired}
\end{align}
where the local kernel is
\begin{align}
\kappa(a,b)
:=
\begin{cases}
1, & a=b=0,\\
1, & a=0,\ b\neq 0\ \text{or}\ a\neq 0,\ b=0,\\
3, & a=b\neq 0,\\
0, & a\neq 0,\ b\neq 0,\ a\neq b.
\end{cases}
\end{align}
Here $0$ denotes the identity Pauli and the three nonzero elements of
$\mathbb F_2^2$ denote the three phase-free nonidentity Paulis.

For $r\ge 0$, write $V_r:=(\mathbb F_2^2)^r$ and
\begin{align}
\kappa_r(u,v):=\prod_{j=1}^r\kappa(u_j,v_j),
\qquad
F_r(A):=\frac{1}{4^r}\sum_{u,v\in A}\kappa_r(u,v)
\end{align}
for $A\subseteq V_r$. An affine set $A\subseteq V_r$ will be called affine
isotropic if $A=u+K$, where $K\subseteq V_r$ is a totally isotropic linear
subspace and $u\in K^\perp$, with orthogonality taken with respect to the
standard symplectic form. Equivalently, every two elements of $A$ are
symplectically orthogonal.

We prove the following induction statement. If $A\subseteq V_r$ is
affine isotropic of affine dimension $m$, then
\begin{align}
F_r(A)
\le
B_{r,m}
:=
\left(\frac34\right)^{r-m}\left(\frac32\right)^m.
\label{eq:cliff-affine-claim-repaired}
\end{align}
The desired stabilizer bound follows by taking $A=S$, $r=m=n$.

The proof is by induction on $r$, uniformly over all possible affine
dimensions $m$. The case $r=0$ is immediate. Assume the claim has been proved
for $V_{r-1}$ and all affine dimensions. Put $s=r-1$, and let
$\pi:V_r\to\mathbb F_2^2$ be the projection onto the last coordinate.
For each $a\in\pi(A)$ define the fiber
\begin{align}
A_a:=\{x\in V_s:(x,a)\in A\}.
\end{align}
If $A=u+K$, then all nonempty fibers have common direction
\begin{align}
L:=\{x\in V_s:(x,0)\in K\}.
\end{align}
Moreover, each fiber is affine isotropic in $V_s$. Indeed, if
$x,y\in A_a$, then $(x,a),(y,a)\in A$, so
$\langle x,y\rangle_{\mathrm{sp}}+\langle a,a\rangle_{\mathrm{sp}}=0$, and
$\langle a,a\rangle_{\mathrm{sp}}=0$.

We shall also use the following consequence of the induction hypothesis. Let
$B\subseteq V_s$ be affine isotropic of affine dimension $d$, and suppose
$h\notin L_B$, where $L_B$ is the direction of $B$, is such that
$B\cup(B+h)$ is affine isotropic. Define
\begin{align}
C(B,B+h):=\frac{1}{4^s}\sum_{x\in B,\,y\in B+h}\kappa_s(x,y).
\end{align}
Then
\begin{align}
\frac14\left(F_s(B)+2C(B,B+h)+3F_s(B+h)\right)
\le B_{s,d}.
\label{eq:proper-two-fiber-estimate}
\end{align}
Indeed, $B\cup(B+h)$ is affine isotropic of dimension $d+1$, so by the
induction hypothesis
\begin{align}
F_s\bigl(B\cup(B+h)\bigr)
=
F_s(B)+2C(B,B+h)+F_s(B+h)
\le B_{s,d+1}=2B_{s,d}.
\end{align}
Also $F_s(B+h)\le B_{s,d}$. Therefore
\begin{align}
\frac14\left(F_s(B)+2C(B,B+h)+3F_s(B+h)\right)
=
\frac14\left(F_s\bigl(B\cup(B+h)\bigr)+2F_s(B+h)\right)
\le B_{s,d}.
\end{align}
This is the estimate that rules out the missing full-projection case.

We now distinguish the possible affine projections $\pi(A)\subseteq\mathbb F_2^2$.

\emph{Case 1: $\pi(A)$ is a single point.}
Write $\pi(A)=\{a\}$. Then $A=A_a\times\{a\}$ and $A_a$ has affine dimension
$m$. Hence
\begin{align}
F_r(A)=\frac{\kappa(a,a)}{4}F_s(A_a).
\end{align}
If $a=0$, then $\kappa(a,a)=1$; if $a\neq0$, then $\kappa(a,a)=3$. In either
case,
\begin{align}
F_r(A)\le \frac34 B_{s,m}=B_{r,m}.
\end{align}

\emph{Case 2: $\pi(A)$ has two points.}
Let $\pi(A)=\{a,b\}$, where $a\neq b$. Put $A_b=A_a+h$. The common fiber
direction has dimension $d=m-1$.

If one of $a,b$ is zero, relabel so that $a=0$ and $b\neq0$. In this subcase
$A_a\cup A_b$ is affine isotropic in $V_s$, because the two last-coordinate
labels $0$ and $b$ are symplectically orthogonal. If $h\in L$, then
$A_b=A_a$ and
\begin{align}
F_r(A)
=
\frac14\left(F_s(A_a)+2F_s(A_a)+3F_s(A_a)\right)
=
\frac32 F_s(A_a)
\le
\frac32 B_{s,d}
=
B_{r,m}.
\end{align}
If $h\notin L$, Eq.~\eqref{eq:proper-two-fiber-estimate} gives the stronger
bound
\begin{align}
F_r(A)
=
\frac14\left(F_s(A_a)+2C(A_a,A_b)+3F_s(A_b)\right)
\le B_{s,d}
\le B_{r,m}.
\end{align}

It remains in this case to consider the possibility that both $a$ and $b$ are
nonzero. Since $a\neq b$, the local factor $\kappa(a,b)$ is zero. Therefore
\begin{align}
F_r(A)
=
\frac14\left(3F_s(A_a)+3F_s(A_b)\right)
\le
\frac32 B_{s,d}
=
B_{r,m}.
\end{align}

\emph{Case 3: $\pi(A)=\mathbb F_2^2$.}
In this case the projection of the direction $K$ has dimension $2$, and hence
$L$ has dimension $d=m-2$. Write the four fibers as
\begin{align}
A_0,
\qquad
A_a=A_0+h_a\quad(a\in\mathbb F_2^2\setminus\{0\}).
\end{align}
For every nonzero $a$, the vector $h_a$ is not in $L$. To see this, suppose
$h_a\in L$. Since $(h_a,a)\in K$ and $(h_a,0)\in K$, their difference
$(0,a)$ would belong to $K$. Since $\pi(K)=\mathbb F_2^2$, there is
$b\in\mathbb F_2^2$ with $\langle a,b\rangle_{\mathrm{sp}}=1$ and some
$(h_b,b)\in K$. But then
\begin{align}
\left\langle (0,a),(h_b,b)\right\rangle_{\mathrm{sp}}
=
\langle a,b\rangle_{\mathrm{sp}}
=
1,
\end{align}
contradicting the isotropy of $K$.

For $a,b\in\mathbb F_2^2$, set
\begin{align}
C_{ab}:=\frac{1}{4^s}\sum_{x\in A_a,\,y\in A_b}\kappa_s(x,y).
\end{align}
The local factor $\kappa(a,b)$ vanishes whenever $a$ and $b$ are distinct
nonzero elements of $\mathbb F_2^2$. Therefore
\begin{align}
F_r(A)
=
\frac14\left(
C_{00}+2\sum_{a\neq0}C_{0a}+3\sum_{a\neq0}C_{aa}
\right).
\label{eq:full-projection-expansion-repaired}
\end{align}
For each nonzero $a$, the set $A_0\cup A_a$ is affine isotropic and
$h_a\notin L$, so Eq.~\eqref{eq:proper-two-fiber-estimate} gives
\begin{align}
\frac14\left(C_{00}+2C_{0a}+3C_{aa}\right)
\le
B_{s,d}.
\end{align}
Summing this inequality over the three nonzero values of $a$ yields
\begin{align}
\frac14\left(
3C_{00}+2\sum_{a\neq0}C_{0a}+3\sum_{a\neq0}C_{aa}
\right)
\le
3B_{s,d}.
\end{align}
Since $C_{00}\ge0$, Eq.~\eqref{eq:full-projection-expansion-repaired} implies
\begin{align}
F_r(A)
\le
3B_{s,d}.
\end{align}
Finally, because $s=r-1$ and $d=m-2$,
\begin{align}
3B_{s,d}
=
3\left(\frac34\right)^{r-m+1}\left(\frac32\right)^{m-2}
=
\left(\frac34\right)^{r-m}\left(\frac32\right)^m
=
B_{r,m}.
\end{align}
This completes the induction and proves Eq.~\eqref{eq:cliff-affine-claim-repaired}.

Applying the claim to the stabilizer support $S\subseteq V_n$, which has
$r=m=n$, gives
\begin{align}
\tr\left[R_{4,\Cl}^{\otimes n}(\psi^{\ox 4})\right]
=
F_n(S)
\le
\left(\frac32\right)^n.
\end{align}
Equality is attained by product stabilizer states. Indeed, if
$S=\{0,p_1\}\times\cdots\times\{0,p_n\}$ with each $p_j\neq0$, then the sum
factorizes into the one-qubit value
\begin{align}
\frac14\left(\kappa(0,0)+2\kappa(0,p_j)+\kappa(p_j,p_j)\right)
=
\frac14(1+2+3)
=
\frac32
\end{align}
on every qubit. Therefore $F_n(S)=(3/2)^n$ for every product stabilizer state.
This gives the lower endpoint in Eq.~\eqref{eq:Clifford-fourth-term-bound}. Prop.~\ref{prop:cliff-pure-stabilizer} gives its upper endpoint, completing the proof of Thm.~\ref{thm:clifford-fourth-term-bound} of MT.
\end{proof}

\appsubsection{Proof of Cor.~\ref{lem:clifford-complexity}}
\label{app:clifford-complexity-proof}

We now combine the variance bounds available for local Clifford sampling to derive the optimized Chebyshev-based sufficient-copy bound summarized in Cor.~\ref{lem:clifford-complexity} of MT. By Appx.~\ref{app:clifford-first-three-terms},
\begin{align}
\mathbb{V}^{(1)}_{\Cl}\le 0,
\qquad
\mathbb{V}^{(2)}_{\Cl}\le \frac{3^n}{N_M^2},
\qquad
\mathbb{V}^{(3)}_{\Cl}\le 2\,\frac{N_M-1}{N_M^2}\left(\frac74\right)^n.
\end{align}
For the fourth term, Thm.~\ref{thm:clifford-fourth-term-bound} of MT gives
\begin{align}
\mathbb{V}^{(4)}_{\Cl}\le \left(\frac{N_M-1}{N_M}\right)^2\left(\frac74\right)^n.
\end{align}
Therefore, by Eq.~\eqref{eq:four-term-decomp},
\begin{align}
\Var[X_M]
=
\mathbb{V}^{(1)}_{\Cl}
+\mathbb{V}^{(2)}_{\Cl}
+\mathbb{V}^{(3)}_{\Cl}
+\mathbb{V}^{(4)}_{\Cl}
\le
\frac{3^n}{N_M^2}
+
2\frac{N_M-1}{N_M^2}\left(\frac74\right)^n
+
\left(\frac{N_M-1}{N_M}\right)^2\left(\frac74\right)^n.
\label{eq:appx-cliff-total-variance}
\end{align}
Substituting Eq.~\eqref{eq:appx-cliff-total-variance} into the general Chebyshev sufficient-copy relation Eq.~\eqref{eq:sample complexity-from-variance}, we obtain the sufficient condition
\begin{align}
\label{eq:appx-cliff-sample-bound-exact}
N
\ge
\frac{1}{\delta\ve^2}
\left[
\frac{3^n}{N_M}
+
2\,\frac{N_M-1}{N_M}\left(\frac74\right)^n
+
\frac{(N_M-1)^2}{N_M}\left(\frac74\right)^n
\right].
\end{align}
Since
\begin{align}
\frac{N_M-1}{N_M}\le 1,
\qquad
\frac{(N_M-1)^2}{N_M}\le N_M,
\end{align}
it further suffices to require
\begin{align}
\label{eq:appx-cliff-sample-bound-simple}
N
\ge
\frac{1}{\delta\ve^2}
\left[
\frac{3^n}{N_M}
+
2\left(\frac74\right)^n
+
N_M\left(\frac74\right)^n
\right].
\end{align}
The only $N_M$-dependent leading terms in Eq.~\eqref{eq:appx-cliff-sample-bound-simple} are
\begin{align}
g(N_M):=\frac{3^n}{N_M}+N_M\left(\frac74\right)^n,
\end{align}
and differentiating gives
\begin{align}
g'(N_M)=-\frac{3^n}{N_M^2}+\left(\frac74\right)^n.
\end{align}
The continuous minimizer of the relaxed bound in Eq.~\eqref{eq:appx-cliff-sample-bound-simple} therefore satisfies
\begin{align}
N_M^2=\frac{3^n}{(7/4)^n}=(12/7)^n,
\qquad
N_M=(12/7)^{n/2}.
\end{align}
Equivalently, this is the point where
\begin{align}
\frac{3^n}{N_M}=N_M\left(\frac74\right)^n=(21/4)^{n/2}.
\end{align}
Hence choosing, for example,
\begin{align}
N_M=\lceil(12/7)^{n/2}\rceil
\end{align}
makes the right-hand side of Eq.~\eqref{eq:appx-cliff-sample-bound-simple} at most
\begin{align}
\frac{2(21/4)^{n/2}+3(7/4)^n}{\delta\ve^2}.
\end{align}
Rounding the sufficient count in Eq.~\eqref{eq:appx-cliff-sample-bound-simple} up to the next multiple of $N_M$ costs less than one additional block. Thus, for $0<\ve,\delta<1$, one can choose
\begin{align}
N_{\star,\Cl}
=
\cO\!\left(
\frac{\sqrt{5.25^n}}{\delta\ve^2}
\right),
\end{align}
which is exactly Eq.~\eqref{eq:cliff-sample-main} in MT. 
Since the middle term in Eq.~\eqref{eq:appx-cliff-sample-bound-simple} is independent of $N_M$ and exponentially smaller than $(21/4)^{n/2}$, while the remaining two terms are balanced at $N_M=(12/7)^{n/2}$, the optimal shot scale for this relaxed bound is
\begin{align}
N_M=\Theta((12/7)^{n/2}).
\end{align}
This proves Cor.~\ref{lem:clifford-complexity}.

\appsection{DIPE with single-qubit Haar ensemble}\label{app:haar}

This section focuses on the single-qubit Haar-specific part of the analysis. The first three variance terms were already derived in Appx.~\ref{app:clifford-first-three-terms}, since they depend only on the second and third moments and therefore coincide for the single-qubit Clifford and single-qubit Haar ensembles. We now study the fourth term $\mathbb{V}^{(4)}_{\mathrm{H}}$, prove the twirl statement and universal upper bound summarized in Thm.~\ref{thm:haar-local-twirl-of-clifford}, and then use a direct Haar fourth-moment bound to derive Thm.~\ref{lem:haar-complexity}.

\appsubsection{Calculations of the fourth term}\label{app:haar-fourth-term}

We first derive the operator form of $\mathbb{V}^{(4)}_{\mathrm{H}}$ from Eq.~\eqref{eq:fourth-term}. 

\vspace*{0.1in}
\textbf{Haar fourth-moment operator.}
For a product Haar sample $U=\bigotimes_{\ell=1}^n U_\ell$, define
\begin{align}
O(U)
:=
\sum_{\bm s,\bm t\in\{0,1\}^n}
f(\bm s,\bm t)
\left(U^\dagger\ketbra{\bm s}U\right)
\ox
\left(U^\dagger\ketbra{\bm t}U\right).
\end{align}
Since $f(\bm s,\bm t)=\prod_{\ell=1}^n h(s_\ell,t_\ell)$ and $U$ is a product unitary, we have
\begin{align}
O(U)=\bigotimes_{\ell=1}^n O_1(U_\ell),
\end{align}
where $O_1(\cdot)$ is defined in Eq.~\eqref{eq:O-1-V}. 
Therefore
\begin{align}\label{eq:mu-Haar-rho-sigma}
\mu_U(\rho,\sigma)
:=
\bE_{\bm s,\bm t\mid U}[f(\bm s,\bm t)]
=
\tr\left[O(U)(\rho\ox\sigma)\right],
\end{align}
and hence
\begin{align}
\bE_U[\mu_U(\rho,\sigma)^2]
=
\tr\left[
\bE_U\!\left[O(U)_{A_1B_1}\ox O(U)_{A_2B_2}\right]
(\rho\ox\sigma\ox\rho\ox\sigma)
\right].
\end{align}
By independence of the single-qubit Haar samples, the averaged operator factorizes as
\begin{align}
\bE_U\!\left[O(U)_{A_1B_1}\ox O(U)_{A_2B_2}\right]
=
R_{4,\mathrm{H}}^{\otimes n},
\end{align}
where the single-qubit local Haar fourth-moment operator $R_{4,\mathrm{H}}$ 
is defined as
\begin{align}\label{appx:eq:R-4-Haar-definition}
R_{4,\mathrm{H}}
:=
\bE_V\!\left[O_1(V)_{A_1B_1}\ox O_1(V)_{A_2B_2}\right]
\end{align}
with $V$ a single-qubit Haar random unitary. Combining this with Eq.~\eqref{eq:fourth-term}, we obtain
\begin{align}
\mathbb{V}^{(4)}_{\mathrm{H}}
=
\left(\frac{N_M-1}{N_M}\right)^2
\tr\left[
R_{4,\mathrm{H}}^{\otimes n}
(\rho\ox\sigma\ox\rho\ox\sigma)
\right],
\end{align}
which is exactly the operator form given in Eq.~\eqref{eq:haar-vterms-v4} of MT.

\vspace*{0.1in}
\textbf{Haar commutant expansion.}
For later use, we also record the explicit Haar commutant expansion of $R_{4,\mathrm{H}}$ 
by exploiting the unitary $4$-design property of the Haar ensemble.
With the replica ordering $(A_1,B_1,A_2,B_2)$, let $\Pi_\pi$ denote the permutation operator for $\pi\in S_4$, and write
$\bF_{ij}:=\Pi_{(ij)}$.
Then the single-qubit local Haar fourth-moment operator $R_{4,\mathrm{H}}$ 
admits the permutation-operator expansion
\begin{align}\label{appx:eq:Haar-commutant-expansion}
R_{4,\mathrm{H}}
= \frac15\left(\bI+\bF_{12}+\bF_{34}\right)
+\frac35\left(\bF_{12}\bF_{34}+\bF_{13}\bF_{24}+\bF_{14}\bF_{23}\right)
-\frac3{10}\left(\bF_{13}+\bF_{24}+\bF_{14}+\bF_{23}\right).
\end{align}
Equivalently, in cycle notation,
\begin{align}\label{appx:eq:Haar-commutant-cycle-notation}
R_{4,\mathrm{H}}
&=
\frac15\,\Pi_{\id}
+\frac15\,\Pi_{(12)}
+\frac15\,\Pi_{(34)}
+\frac35\,\Pi_{(12)(34)}
+\frac35\,\Pi_{(13)(24)}
+\frac35\,\Pi_{(14)(23)}\nonumber\\
&\quad
-\frac3{10}\,\Pi_{(13)}
-\frac3{10}\,\Pi_{(24)}
-\frac3{10}\,\Pi_{(14)}
-\frac3{10}\,\Pi_{(23)}.
\end{align}

\vspace*{0.1in}
\textbf{The optimization problem.}
The fourth-moment expression satisfies
\begin{align}
\tr\left[R_{4,\mathrm{H}}^{\otimes n}
(\rho\ox\sigma\ox\rho\ox\sigma)\right]
=
\bE_U[\mu_U(\rho,\sigma)^2],
\end{align}
where $\mu_U(\rho,\sigma)$ is defined in Eq.~\eqref{eq:mu-Haar-rho-sigma}.
We need to solve the following optimization problem in order to obtain the worst-case upper bound
on $\mathbb{V}^{(4)}_{\mathrm{H}}$ over all possible quantum states:
\begin{align}\label{appx:eq:Haar-upper-bound-optimization}
\max_{\rho,\sigma\in\density{(\bC^2)^{\otimes n}}}
\tr\left[
R_{4,\mathrm{H}}^{\otimes n}
(\rho\ox\sigma\ox\rho\ox\sigma)
\right].
\end{align}

Now we show that the maximum in Eq.~\eqref{appx:eq:Haar-upper-bound-optimization} is attained 
by a pair of identical pure states.
The argument closely parallels the proof of Prop.~\ref{prop:cliff-v4-identical-pure}
established for the single-qubit Clifford ensemble; 
for completeness, we present the Haar version here in a self-contained form.

\begin{proposition}
\label{prop:haar-v4-pure}
It holds that
\begin{align}
\max_{\rho,\sigma\in\density{(\bC^2)^{\otimes n}}}
\tr\left[
R_{4,\mathrm{H}}^{\otimes n}
(\rho\ox\sigma\ox\rho\ox\sigma)
\right]
=
\max_{\substack{\ket{\psi}\in(\bC^2)^{\otimes n}\\ \langle\psi|\psi\rangle=1}}
\tr\left[
R_{4,\mathrm{H}}^{\otimes n}
\psi^{\ox4}
\right].
\end{align}
In particular, the maximum is attained by a pair of identical pure states.
\end{proposition}
\begin{proof}
Recall the quantity $\mu_U(\rho,\sigma)$ defined in Eq.~\eqref{eq:mu-Haar-rho-sigma}.
For every fixed $\sigma$ and every Haar sample $U$, $\mu_U(\rho,\sigma)$ is affine in $\rho$. 
Therefore, for $\rho_1,\rho_2\in\density{(\bC^2)^{\otimes n}}$ and $\lambda\in[0,1]$,
\begin{align}
&\; \tr\left[R_{4,\mathrm H}^{\otimes n}
((\lambda\rho_1+(1-\lambda)\rho_2)\ox\sigma
\ox(\lambda\rho_1+(1-\lambda)\rho_2)\ox\sigma)
\right] \\
=&\;
\bE_U\!\left[
\left(
\lambda\mu_U(\rho_1,\sigma)+(1-\lambda)\mu_U(\rho_2,\sigma)
\right)^2
\right] \\
\leq&\; 
\lambda\tr\left[R_{4,\mathrm H}^{\otimes n}(\rho_1\ox\sigma\ox\rho_1\ox\sigma)\right]
+ (1-\lambda)\tr\left[R_{4,\mathrm H}^{\otimes n}(\rho_2\ox\sigma\ox\rho_2\ox\sigma)\right],
\end{align}
because $x\mapsto x^2$ is convex. Hence, for fixed $\sigma$, the Haar
fourth-moment objective is convex in $\rho$ on the compact convex set of
density operators, so its maximum is attained at an extreme point, namely a
pure state. Similarly, for fixed $\rho$, the objective is convex in $\sigma$.

Now let $(\rho_\star,\sigma_\star)$ be a maximizing pair. By convexity in the first argument, there exists a pure state $\ket{\psi}$ such that
\begin{align}
\tr\left[
R_{4,\mathrm H}^{\otimes n}
(\psi\ox\sigma_\star\ox\psi\ox\sigma_\star)
\right]
\ge
\tr\left[
R_{4,\mathrm H}^{\otimes n}
(\rho_\star\ox\sigma_\star\ox\rho_\star\ox\sigma_\star)
\right].
\end{align}
Since $(\rho_\star,\sigma_\star)$ is already globally optimal, equality must hold. Applying the same argument to the second argument yields a pure state $\ket{\phi}$ such that
\begin{align}
\tr\left[
R_{4,\mathrm H}^{\otimes n}
(\psi\ox\phi\ox\psi\ox\phi)
\right]
\ge
\tr\left[
R_{4,\mathrm H}^{\otimes n}
(\psi\ox\sigma_\star\ox\psi\ox\sigma_\star)
\right]
=
\tr\left[
R_{4,\mathrm H}^{\otimes n}
(\rho_\star\ox\sigma_\star\ox\rho_\star\ox\sigma_\star)
\right].
\end{align}
Thus equality must hold, so the global maximum is attained by the pure state pair $(\ketbra{\psi},\ketbra{\phi})$.

It remains to show that the two pure states may be chosen identical. For a fixed Haar sample $U$ and state $\tau$, define
\begin{align}
p_U^\tau(\bm s):=\tr\left[U^\dagger\ketbra{\bm s}U\,\tau\right],
\qquad
\wh p_U^\tau(S):=\sum_{\bm s\in\{0,1\}^n}
(-1)^{\sum_{\ell\in S}s_\ell}p_U^\tau(\bm s).
\end{align}
Here $p_U^\tau$ is the computational basis measurement distribution of $\tau$ after applying $U$, and $\wh p_U^\tau(S)$ is its Walsh coefficient on the subset $S\subseteq[n]$. With this notation, we write
\begin{align}
\mu_U(\rho,\sigma)
=
\sum_{S\subseteq[n]}\frac{3^{|S|}}{2^n}\,\wh p_U^\rho(S)\,\wh p_U^\sigma(S).
\end{align}
Since all weights $3^{|S|}/2^n$ are nonnegative, the weighted Cauchy--Schwarz inequality and the arithmetic-geometric mean inequality give
\begin{align}
\mu_U(\rho,\sigma)^2
\leq \mu_U(\rho,\rho)\,\mu_U(\sigma,\sigma)
\leq \frac{\mu_U(\rho,\rho)^2+\mu_U(\sigma,\sigma)^2}{2}.
\end{align}
Averaging over $U$ yields
\begin{align}
\tr\left[
R_{4,\mathrm H}^{\otimes n}
(\rho\ox\sigma\ox\rho\ox\sigma)
\right]
\le
\frac{
\tr\left[R_{4,\mathrm H}^{\otimes n}\rho^{\ox4}\right]
+
\tr\left[R_{4,\mathrm H}^{\otimes n}\sigma^{\ox4}\right]
}{2}
\le
\max\left\{
\tr\left[R_{4,\mathrm H}^{\otimes n}\rho^{\ox4}\right],
\tr\left[R_{4,\mathrm H}^{\otimes n}\sigma^{\ox4}\right]
\right\}.
\end{align}
Applying this to the maximizing pure state pair $(\ketbra{\psi},\ketbra{\phi})$, we obtain
\begin{align}
\tr\left[
R_{4,\mathrm H}^{\otimes n}
(\psi\ox\phi\ox\psi\ox\phi)
\right]
\le
\max\left\{
\tr\left[R_{4,\mathrm H}^{\otimes n}\psi^{\ox4}\right],
\tr\left[R_{4,\mathrm H}^{\otimes n}\phi^{\ox4}\right]
\right\}.
\end{align}
The left-hand side is the global maximum, while each term on the right-hand side is bounded above by the same quantity. Therefore equality must hold, and at least one of the identical pure state pairs $(\ketbra{\psi},\ketbra{\psi})$ or $(\ketbra{\phi},\ketbra{\phi})$ also attains the global maximum.
\end{proof}

\appsubsection{Proof of Thm.~\ref{thm:haar-local-twirl-of-clifford}}
\label{appendix:sec:haar-local-twirl-of-clifford}

\vspace*{0.1in}
\textbf{Proof of Eq.~\eqref{eq:haar-local-twirl-of-clifford-identity} 
in Thm.~\ref{thm:haar-local-twirl-of-clifford}.}
Recall the definitions of the one-qubit fourth-moment operators
$R_{4,\mathrm H}$ and $R_{4,\Cl}$ in
Eqs.~\eqref{appx:eq:R-4-Haar-definition} and \eqref{eq:cliff-r4-main}. We first prove the
one-qubit twirling identity
\begin{align}\label{eq:haar-local-twirl-of-clifford}
R_{4,\mathrm{H}}
=
\bE_{W\sim\nu}\!\left[
(W^{\otimes 4})^\dagger R_{4,\Cl} W^{\otimes 4}
\right],
\end{align}
where $\nu$ denotes the single-qubit Haar measure.
For a one-qubit Hermitian unitary $M$ with eigenvalues $\pm1$, define
\begin{align}
\mathcal{K}(M)
:=
\left(\frac12\bI+\frac32\,M\ox M\right)_{A_1B_1}
\ox
\left(\frac12\bI+\frac32\,M\ox M\right)_{A_2B_2}.
\end{align}
Since
$O_1(V)=\frac12\bI+\frac32(V^\dagger ZV)\ox(V^\dagger ZV)$, the definition
of $R_{4,\mathrm H}$ in Eq.~\eqref{appx:eq:R-4-Haar-definition} gives
\begin{align}\label{appx:eq-R-4-Haar-K}
R_{4,\mathrm{H}}
=
\bE_V\!\left[\mathcal{K}(V^\dagger ZV)\right],
\end{align}
where $V$ is a single-qubit Haar random unitary.

On the other hand, if $C$ is a uniformly random single-qubit Clifford, then $C^\dagger ZC$ is uniformly distributed over $\{\pm X,\pm Y,\pm Z\}$, 
and the sign disappears inside $M\ox M$. 
Hence it follows from the definition of $R_{4,\Cl}$ in Eq.~\eqref{eq:cliff-r4-main}
that,
\begin{align}\label{appx:eq-R-4-Clifford-K}
R_{4,\Cl}
=
\frac13\sum_{P\in\{X,Y,Z\}}\mathcal{K}(P).
\end{align}

For arbitrary single-qubit unitary $W$ sampled according to the Haar measure $\nu$, we have
\begin{align}
(W^{\otimes 4})^\dagger R_{4,\Cl} W^{\otimes 4}
&= (W^{\otimes 4})^\dagger \left(\frac13\sum_{P\in\{X,Y,Z\}}\mathcal{K}(P)\right) W^{\otimes 4}
\notag\\
&= \frac13\sum_{P\in\{X,Y,Z\}}(W^{\otimes 4})^\dagger \mathcal{K}(P) W^{\otimes 4}
\notag\\
&= \frac13\sum_{P\in\{X,Y,Z\}}\mathcal{K}(W^\dagger P W),
\end{align}
where the last step follows from the definition of $\mathcal{K}$, since
\begin{align}
(W^{\otimes 4})^\dagger\!\left[\left(\frac12\bI+\frac32\,P\ox P\right)_{A_1B_1}\ox\left(\frac12\bI+\frac32\,P\ox P\right)_{A_2B_2}\right]W^{\otimes 4}
=
\mathcal{K}(W^\dagger P W).
\end{align}
Averaging over all single-qubit Haar-random unitary $W$ gives
\begin{align}
\bE_{W\sim\nu}\!\left[(W^{\otimes 4})^\dagger R_{4,\Cl} W^{\otimes 4}\right]
= \frac13\sum_{P\in\{X,Y,Z\}}\bE_{W\sim\nu}\!\left[\mathcal{K}(W^\dagger P W)\right].
\end{align}

For each fixed $P\in\{X,Y,Z\}$, the random observable $W^\dagger P W$ is uniformly distributed over all Bloch-sphere directions
since $W$ is Haar distributed, and hence has the same distribution as $V^\dagger ZV$ for Haar-random $V$. 
Therefore each term in the sum equals $R_{4,\mathrm{H}}$, and so
\begin{align}
\bE_{W\sim\nu}\!\left[(W^{\otimes 4})^\dagger R_{4,\Cl} W^{\otimes 4}\right]
=
\frac13\sum_{P\in\{X,Y,Z\}}R_{4,\mathrm{H}}
=
R_{4,\mathrm{H}}.
\end{align}
This proves Eq.~\eqref{eq:haar-local-twirl-of-clifford}.

Taking the tensor product over $n$ independent qubits and using independence of $W_1,\dots,W_n$ yields
\begin{align}
R_{4,\mathrm{H}}^{\otimes n}
=
\bigotimes_{\ell=1}^n
\bE_{W_\ell\sim\nu}\!\left[
(W_\ell^{\otimes 4})^\dagger R_{4,\Cl} W_\ell^{\otimes 4}
\right]
=
\bE_{\bm W\sim\nu^{\otimes n}}\!\left[
(\bm W^{\otimes 4})^\dagger R_{4,\Cl}^{\otimes n}\bm W^{\otimes 4}
\right],
\end{align}
which is the operator identity in Eq.~\eqref{eq:haar-local-twirl-of-clifford-identity}
of Thm.~\ref{thm:haar-local-twirl-of-clifford}.

\vspace*{0.1in}
\textbf{Proof of Eq.~\eqref{eq:haar-local-twirl-of-clifford-induced-bound} 
in Thm.~\ref{thm:haar-local-twirl-of-clifford}.} For arbitrary
$n$-qubit states $\rho$ and $\sigma$, the operator identity just proved gives
\begin{align}
\tr\left[
R_{4,\mathrm H}^{\otimes n}
(\rho\ox\sigma\ox\rho\ox\sigma)
\right]
&=
\bE_{\bm W\sim\nu^{\otimes n}}
\tr\left[
(\bm W^{\otimes 4})^\dagger
R_{4,\Cl}^{\otimes n}
\bm W^{\otimes 4}
(\rho\ox\sigma\ox\rho\ox\sigma)
\right] \\
&=
\bE_{\bm W\sim\nu^{\otimes n}}
\tr\left[
R_{4,\Cl}^{\otimes n}
\left(
\bm W\rho\bm W^\dagger
\ox
\bm W\sigma\bm W^\dagger
\ox
\bm W\rho\bm W^\dagger
\ox
\bm W\sigma\bm W^\dagger
\right)
\right].
\end{align}
Let $M_{\Cl,n}:=\max_{\tau,\omega}B_{n,\Cl}(\tau,\omega)$ denote the Clifford fourth-moment maximum in Eq.~\eqref{eq:Clifford-fourth-term-bound}, where $\tau$ and $\omega$ range over all $n$-qubit states.
For each fixed $\bm W$, the pair
$(\bm W\rho\bm W^\dagger,\bm W\sigma\bm W^\dagger)$ is again a valid pair of
$n$-qubit density operators. Therefore Thm.~\ref{thm:clifford-fourth-term-bound}
implies the pointwise bound
\begin{align}
\tr\left[
R_{4,\Cl}^{\otimes n}
\left(
\bm W\rho\bm W^\dagger
\ox
\bm W\sigma\bm W^\dagger
\ox
\bm W\rho\bm W^\dagger
\ox
\bm W\sigma\bm W^\dagger
\right)
\right]
\le
M_{\Cl,n}\le\left(\frac74\right)^n.
\end{align}
Averaging over $\bm W$ yields
\begin{align}
\tr\left[
R_{4,\mathrm H}^{\otimes n}
(\rho\ox\sigma\ox\rho\ox\sigma)
\right]
\le
M_{\Cl,n}\le\left(\frac74\right)^n.
\end{align}
Since this holds for every $\rho$ and $\sigma$, taking the maximum over all
state pairs proves Eq.~\eqref{eq:haar-local-twirl-of-clifford-induced-bound}.
This completes the proof of Thm.~\ref{thm:haar-local-twirl-of-clifford}.

\appsubsection{Proof of Thm.~\ref{lem:haar-complexity}}
\label{app:haar-complexity-proof}

For $c=63391/50000$, we prove that every $n$-qubit input pair satisfies
\begin{align}
B_{n,\mathrm H}(\rho,\sigma)
=\tr\!\left[R_{4,\mathrm H}^{\otimes n}
(\rho\ox\sigma\ox\rho\ox\sigma)\right]
\le c^n\le (4/3)^n .
\label{eq:haar-certified-fourth}
\end{align}
The proof bounds the expectation on a product of four replicas. Each replica may be entangled across all $n$ sites.

On the four replicas of one site, let $\Theta_S$ transpose the qubits in $S\subseteq\{1,2,3,4\}$ in the computational basis and conjugate by Pauli $Y$ on those qubits. This map preserves the Hilbert--Schmidt inner product. Applied at every site, it transposes each selected \emph{whole replica} and conjugates it by $Y^{\otimes n}$. Thus $\Theta_S^{\otimes n}$ maps a product of four density matrices to another such product.

Suppose $H$ and $K$ are Hermitian and $H\pm K\succeq0$. Expanding $H^{\otimes n}\pm K^{\otimes n}$ in tensor products of $(H+K)/2$ and $(H-K)/2$ gives a sum of positive semidefinite terms. Hence, if $H_{j+1}\pm\Theta_{S_j}(H_j)\succeq0$, then any replica product $\omega$ obeys
\[
\big|\tr[H_j^{\otimes n}\omega]\big|
\le \tr[H_{j+1}^{\otimes n}\Theta_{S_j}^{\otimes n}(\omega)] .
\]
A sequence from $H_0=R_{4,\mathrm H}$ to $H_9=c\bI$ therefore proves Eq.~\eqref{eq:haar-certified-fourth}.

We give the exact local matrices for this sequence. Write $D_{ab}=2\bF_{ab}-\bI=X_aX_b+Y_aY_b+Z_aZ_b$ and order the basis operators as
\[
(E_0,\ldots,E_9)=(\bI,D_{12},D_{13},D_{14},D_{23},D_{24},D_{34},
                   D_{12}D_{34},D_{13}D_{24},D_{14}D_{23}).
\]
Set $H_j=L^{-1}\sum_\ell h_{j\ell}E_\ell$ with $L=50\,000\,000$. From Eq.~\eqref{appx:eq:Haar-commutant-expansion},
$h_0=(L/20)(5,5,0,0,0,0,5,3,3,3)$, and $h_9=(Lc,0,\ldots,0)$. We obtained the intermediate rows of Tab.~\ref{tab:haar-local-certificate} by solving the local semidefinite inequalities numerically, rounding the coefficients and adding a rational multiple of the identity for the exact verification below. For $j=0,\ldots,8$, the sets $S_j$ in order are
\[
(\{2,3\},\{1,2\},\{1,3\},\{2,3\},\{1,3\},\{1\},\{1,2\},\{2,3\},\{1,2\}).
\]

\begin{table*}[tb]
\centering
\renewcommand{\arraystretch}{1.8}
\setlength{\tabcolsep}{0pt}
\setlength{\heavyrulewidth}{0.3ex}
{\small
\begin{tabular}{@{}*{11}{C{\dimexpr\textwidth/11\relax}}@{}}
\toprule
$\bm{j}$ & $\bm{h_{j0}}$ & $\bm{h_{j1}}$ & $\bm{h_{j2}}$ & $\bm{h_{j3}}$ & $\bm{h_{j4}}$ & $\bm{h_{j5}}$ & $\bm{h_{j6}}$ & $\bm{h_{j7}}$ & $\bm{h_{j8}}$ & $\bm{h_{j9}}$\\
\midrule
1 & 31874613 & -6064890 & -6172442 & -6331487 & -6305363 & -6191016 & -6851193 & 13541958 & 1318271 & 1181575\\
2 & 34388359 & -5267065 & 4887793 & 6480792 & 4522711 & 6282346 & -7428229 & 13187548 & 1669342 & 1277070\\
3 & 52529407 & 1452565 & -2238978 & -1257539 & 492775 & 3000458 & -681364 & 7225502 & -3534988 & 6396440\\
4 & 52529432 & -1452565 & 2238978 & -1257539 & 492775 & -3000458 & 681364 & 7225502 & -3534988 & 6396440\\
5 & 53691186 & 1189158 & 2614067 & 946182 & -946182 & -2614067 & -1189158 & 6839901 & -3154248 & 6014058\\
6 & 57428256 & -3886253 & -1380413 & -3219786 & -3219786 & -1380413 & -3886253 & -4142806 & 1920593 & -3740453\\
7 & 57745640 & -4014530 & 1468060 & 2943056 & 2943056 & 1468060 & -4014530 & -4014530 & 1832945 & -3463724\\
8 & 62703311 & 609772 & 182443 & -260334 & -260334 & 182443 & 609772 & -609772 & 182443 & -260334\\
\bottomrule
\end{tabular}}
\caption{\raggedright\textbf{Certificate for the Haar fourth-moment bound.}
Row $j$ gives the integer coefficients $h_{j\ell}$ of the intermediate four-replica operator
$H_j=L^{-1}\sum_{\ell=0}^{9}h_{j\ell}E_\ell$ used to prove Eq.~\eqref{eq:haar-certified-fourth},
with common denominator $L=50\,000\,000$.}
\label{tab:haar-local-certificate}
\end{table*}

The following determinant test reproduces the local comparisons from the table. Let $T_\ell$ be the replica support of $E_\ell$. Since $\Theta$ takes each one-qubit Pauli matrix to its negative, $\Theta_S(E_\ell)=(-1)^{|S\cap T_\ell|}E_\ell$. For each $j=0,\ldots,8$ and $\eta=\pm1$, form the symmetric integer matrix
\[
G=\sum_{\ell=0}^9
  \big(h_{j+1,\ell}+\eta(-1)^{|S_j\cap T_\ell|}h_{j\ell}\big)E_\ell .
\]
Index the basis by $\sum_{a=1}^4 2^{a-1}b_a$ for replica bits $b_a$. With $\Delta_0=1$, let $\Delta_k$ denote the determinant of the leading $k\times k$ block of $G$. The coefficients in Tab.~\ref{tab:haar-local-certificate} give $\Delta_k\ge 21\Delta_{k-1}$ for $1\le k\le16$ in all 18 matrices. To verify these inequalities in integer arithmetic, start with $G^{(1)}=G$ and take $\Delta_k=G^{(k)}_{kk}$ at step $k$. Update the trailing block by
\[
G^{(k+1)}_{ab}
=\frac{\Delta_k G^{(k)}_{ab}-G^{(k)}_{ak}G^{(k)}_{kb}}{\Delta_{k-1}}
 \qquad(a,b>k).
\]
Sylvester's criterion now gives $H_{j+1}\pm\Theta_{S_j}(H_j)\succ0$ at every step, proving the fourth-moment bound.

For either $a=4/3$ or $a=c$, choose $N_M=\lceil(3/a)^{n/2}\rceil$. Substituting $A_n\le3^n$ and $B_{n,\mathrm H}\le a^n$ into Eq.~\eqref{eq:haar-state-dependent-certificate} bounds its two shot-dependent terms by $(3a)^{n/2}$ and $(3a)^{n/2}+a^n$, respectively. Since $7/4<\sqrt{3a}$ in both cases, these uniform bounds give a sufficient copy count for all input pairs. For $a=4/3$ this yields Eq.~\eqref{eq:haar-sample-main}, while $a=c$ gives the sharper count
\begin{align}
N_{\star,\mathrm H} = \cO\!\left(
\frac{(190173/50000)^{n/2}}{\delta\ve^2}\right).
\label{eq:haar-sharper-bound}
\end{align}
We emphasize that the optimal worst-case Haar exponent remains open.

\appsubsection{Comparison of the Haar and Clifford fourth moments}
\label{appx:sec:universal-upper-bound-not-achievable}

The first three variance terms agree for Haar and Clifford sampling, so the fourth term determines how the variances differ at a fixed shot allocation. We ask whether the local average in Thm.~\ref{thm:haar-local-twirl-of-clifford} can attain the largest Clifford fourth moment. For every $n\ge1$, Eq.~\eqref{eq:haar-certified-fourth} and the Clifford product-stabilizer value give $\max_{\rho,\sigma}B_{n,\mathrm H}(\rho,\sigma)\le(4/3)^n<(3/2)^n\le M_{\Cl,n}$. We derive the condition that equality of the two maxima would impose on a local unitary orbit and compute the exact Haar maxima for one and two qubits.

By Prop.~\ref{prop:haar-v4-pure}, the Haar maximum is attained on a pair of identical pure states.
Applying the twirl identity in Eq.~\eqref{eq:haar-local-twirl-of-clifford-identity} to such a state $\psi=\ketbra{\psi}$ gives the scalar
relation
\begin{align}
\tr\left[
R_{4,\mathrm H}^{\otimes n}\psi^{\ox4}
\right]
&=
\bE_{\bm W\sim\nu^{\otimes n}}
\tr\left[
(\bm W^{\otimes4})^\dagger
R_{4,\Cl}^{\otimes n}
\bm W^{\otimes4}
\psi^{\ox4}
\right]
\notag\\
&=
\bE_{\bm W\sim\nu^{\otimes n}}
\tr\left[
R_{4,\Cl}^{\otimes n}
(\bm W\psi\bm W^\dagger)^{\ox4}
\right],
\end{align}
where $\bm W=\bigotimes_{\ell=1}^n W_\ell$ and each $W_\ell\sim\nu$.
This consequence of the operator identity yields a necessary condition for equality. Indeed, if a pure state $\psi$ were to satisfy
\(\tr\left[R_{4,\mathrm H}^{\otimes n}\psi^{\ox4}\right]=M_{\Cl,n}\), then the random variable
\begin{align}
\tr\left[
R_{4,\Cl}^{\otimes n}
(\bm W\psi\bm W^\dagger)^{\ox4}
\right]
\end{align}
would have expectation equal to its pointwise upper bound $M_{\Cl,n}$. Therefore it must equal $M_{\Cl,n}$ for almost every $\bm W$. Since the local Haar measure has full support and the map
\begin{align}
\bm W\mapsto
\tr\left[
R_{4,\Cl}^{\otimes n}
(\bm W\psi\bm W^\dagger)^{\ox4}
\right]
\end{align}
is continuous, it follows that in fact
\begin{align}\label{eq:haar-saturation-necessary-condition}
\tr\left[
R_{4,\Cl}^{\otimes n}
(\bm W\psi\bm W^\dagger)^{\ox4}
\right]
=
M_{\Cl,n}
\end{align}
for all local unitaries $\bm W$. Thus, Haar saturation would require that the \emph{entire local unitary orbit} of $\psi$ consists of Clifford maximizers. 
For $n=1,2$, we can compare the exact Haar maximum with the value $(3/2)^n$ of a Clifford product stabilizer.

\vspace*{0.1in}
\textbf{The single-qubit case.} 
Let $\psi=\ketbra{\psi}$ be an arbitrary one qubit pure state. 
Since $\psi^{\ox 4}$ is invariant under every replica permutation, then for arbitrary $\pi\in S_4$,
\begin{align}
\Pi_\pi\,\psi^{\ox 4}\,\Pi_\pi^\dagger=\psi^{\ox 4},
\qquad
\tr\left[\Pi_\pi\psi^{\ox 4}\right]=1.
\end{align}
Eq.~\eqref{appx:eq:Haar-commutant-cycle-notation} gives
\begin{align}\label{appx:eq:the-single-qubit-case-Haar}
\max_{\substack{\ket{\psi}\in\bC^2\\ \langle\psi|\psi\rangle=1}}\tr\left[R_{4,\mathrm{H}}\psi^{\ox 4}\right]
= \frac15+\frac15+\frac15 + 3\cdot\frac35 - 4\cdot\frac3{10}
= \frac65
<
\frac32.
\end{align}
A pure Pauli eigenstate gives $3/2$ under Clifford sampling. Thus the Haar maximum is strictly smaller than $M_{\Cl,1}$.

\vspace*{0.1in}
\textbf{The two-qubit case.} 
For two qubits, local unitary invariance allows us to parameterize every bipartite pure state in the Schmidt form
\begin{align}
\ket{\psi_\lambda}
=
\sqrt{\lambda}\ket{00}
+
\sqrt{1-\lambda}\ket{11},
\qquad
0\le \lambda\le 1.
\end{align}
Let $\psi_\lambda:=\ketbra{\psi_\lambda}$. A direct Pauli-basis computation gives
\begin{align}
\wh\psi_\lambda(II) = \wh\psi_\lambda(ZZ) = 1,\quad
\wh\psi_\lambda(ZI) = \wh\psi_\lambda(IZ) = 2\lambda-1,\quad
\wh\psi_\lambda(XX) = - \wh\psi_\lambda(YY) = 2\sqrt{\lambda(1-\lambda)},
\end{align}
with all other Pauli coefficients equal to zero. To evaluate the fourth moment, write $R_{4,\mathrm H}=\sum_\pi c_\pi\Pi_\pi$ as in Eq.~\eqref{appx:eq:Haar-commutant-cycle-notation}. For $a,b\in\{0,1\}^4$, define $m_{ba}:=\langle b|R_{4,\mathrm H}|a\rangle=\sum_{\pi:\,\pi a=b}c_\pi$. Each permutation preserves Hamming weight. Grouping the squared matrix elements by weight gives
\begin{align*}
S_k:=\sum_{|a|=|b|=k}m_{ba}^2,\qquad
25(S_0,S_1,S_2,S_3,S_4)=(36,86,336,86,36).
\end{align*}
Indeed, the weight-zero block has the single entry $6/5$, and each weight-one column has entries $-1/5,4/5,3/10,3/10$ up to order. In the weight-two block of $10m$, columns $0011$ and $1100$ have squared norm $324$, and the other four columns have squared norm $174$. Bitwise complementation gives $S_{4-k}=S_k$. Grouping the two physical qubits separately in $\ket{\psi_\lambda}^{\otimes4}$ gives the basis vectors $\ket a\otimes\ket a$ with amplitudes $\lambda^{(4-|a|)/2}(1-\lambda)^{|a|/2}$. Both sites contribute the same real matrix element $m_{ba}$, so the fourth moment is $\sum_{k=0}^4 S_k\lambda^{4-k}(1-\lambda)^k$. Expanding gives
\begin{align}
\tr\left[
R_{4,\mathrm H}^{\otimes2}\psi_\lambda^{\ox4}
\right]
&=
\frac{36}{25}
-
\frac{58}{25}\lambda(1-\lambda)
+
\frac{236}{25}\lambda^2(1-\lambda)^2.
\end{align}
Let $t:=\lambda(1-\lambda)\in[0,1/4]$, we obtain
\begin{align}
\tr\left[
R_{4,\mathrm H}^{\otimes2}\psi_\lambda^{\ox4}
\right]
&=
\frac{36}{25}
-
\frac{58}{25}t
+
\frac{236}{25}t^2.
\end{align}
This is a convex quadratic function of $t$, so its maximum over $[0,1/4]$ is attained at an endpoint. Therefore
\begin{align}
\max_{\lambda\in[0,1]}
\tr\left[
R_{4,\mathrm H}^{\otimes2}\psi_\lambda^{\ox4}
\right]
=
\max\left\{
\frac{36}{25},
\frac{29}{20}
\right\}
=
\frac{29}{20}.
\end{align}
Hence
\begin{align}
\max_{\substack{\ket{\psi}\in(\bC^2)^{\otimes 2}\\ \langle\psi|\psi\rangle=1}}
\tr\left[
R_{4,\mathrm H}^{\otimes2}\psi^{\ox4}
\right]
=
\frac{29}{20}
<
\left(\frac32\right)^2.
\end{align}
The maximum is attained by a Bell state, whereas product states attain only $(6/5)^2=36/25$. Since identical product stabilizer states attain $(3/2)^2$ under Clifford sampling, the Haar maximum is strictly below $M_{\Cl,2}$.

\appsubsection{State-dependent Haar bounds}
\label{app:haar-state-dependent}

Thm.~\ref{lem:haar-complexity} uses one shot allocation for all input pairs. To determine when fewer copies suffice, we balance the second and fourth variance terms for a given pair. For identical pure product inputs, their coefficients are $A_n=3^n$ and $B_{n,\mathrm H}=(6/5)^n$. More generally, whenever $A_nB_{n,\mathrm H}\le(18/5)^n$, Prop.~\ref{conj:Haar-complexity} gives a state-dependent allocation using $\cO(\sqrt{3.6^n})$ copies per party at fixed accuracy and confidence. We verify this condition for arbitrary product pairs and for entangled families, including the Bell dimers in Fig.~\ref{fig:numerical-leading-scales}(b). For identical pure inputs, the overlap is one, and these examples characterize the variance of the fixed estimator.

\appsubsubsection{Product state benchmark}

Let
\begin{align}
\ket{\psi}=\bigotimes_{\ell=1}^n\ket{\psi_\ell}
\end{align}
be an arbitrary pure product state. Since the Haar fourth-moment operator is
local across qubits, the expectation factorizes:
\begin{align}
\tr\left[R_{4,\mathrm H}^{\otimes n}\psi^{\ox4}\right]
=
\prod_{\ell=1}^n
\tr\left[R_{4,\mathrm H}\psi_\ell^{\ox4}\right].
\end{align}
For every single-qubit pure state $\psi_\ell$, Eq.~\eqref{appx:eq:the-single-qubit-case-Haar}
gives
\begin{align}
\tr\left[R_{4,\mathrm H}\psi_\ell^{\ox4}\right]
=
\frac65.
\end{align}
Therefore every pure product state satisfies
\begin{align}
\tr\left[R_{4,\mathrm H}^{\otimes n}\psi^{\ox4}\right]
=
\left(\frac65\right)^n.
\label{eq:haar-product-fourth-main}
\end{align}

\appsubsubsection{State-dependent shot allocation}
\label{app:haar-special-cases}

For a state pair $(\rho,\sigma)$ and $\mathrm E\in\{\Cl,\mathrm H\}$, write
\begin{align}
A_n(\rho,\sigma)
&:= \tr\left[(2\bI+\bF)^{\otimes n}(\rho\ox\sigma)\right],\label{eq:A-n} \\
C_n(\rho,\sigma)
&:= \tr\left[R_{AA'B}^{\otimes n}(\rho\ox\rho\ox\sigma)
\right] + \tr\left[R_{ABB'}^{\otimes n}(\rho\ox\sigma\ox\sigma)\right],
\label{eq:C-n}\\
B_{n,\mathrm E}(\rho,\sigma)
&:=
\tr\left[
R_{4,\mathrm E}^{\otimes n}
(\rho\ox\sigma\ox\rho\ox\sigma)
\right].\label{eq:B-n}
\end{align}
For $\mathrm E=\mathrm H$, by Eqs.~\eqref{eq:shared-vterms-v2} and
\eqref{eq:haar-vterms-v4}, these quantities are the state-dependent
coefficients of the Haar second, third, and fourth variance terms:
\begin{align}
\mathbb{V}^{(2)}_{\mathrm H}
&= \frac{A_n(\rho,\sigma)}{N_M^2}, \\
\mathbb{V}^{(3)}_{\mathrm H}
&= \frac{N_M-1}{N_M^2}C_n(\rho,\sigma), \\
\mathbb{V}^{(4)}_{\mathrm H}
&=
\left(\frac{N_M-1}{N_M}\right)^2
B_{n,\mathrm H}(\rho,\sigma).
\label{eq:A-B-variance-relation}
\end{align}

\begin{proposition}
\label{conj:Haar-complexity}
Let $\wh X_f$ be the estimator in Eq.~\eqref{eq:estimator} built from the unique kernel in Thm.~\ref{thm:unique-kernel}.
Under single-qubit Haar sampling, if a state pair $(\rho,\sigma)$ satisfies
$A_n(\rho,\sigma)B_{n,\mathrm H}(\rho,\sigma)\le (18/5)^n$, then the guarantee
\begin{align}
\bP(|\wh X_f-\tr\left[\rho\sigma\right]|\ge \ve)\le \delta
\end{align}
holds with a state-dependent sufficient number of copies per party satisfying
\begin{align}
\label{eq:haar-conjectured-sample-main}
N_{\star,\mathrm H}(\rho,\sigma) \leq \cO\!\left(\frac{\sqrt{3.6^n}}{\delta\ve^2}\right).
\end{align}
For $T=(18/5)^{n/2}$, one may choose $N_M=\max\{1,\lceil A_n(\rho,\sigma)/T\rceil\}$.
\end{proposition}

The moment condition thus guarantees the existence of an allocation for each pair. Choosing $N_M$ by this formula requires information about $A_n(\rho,\sigma)$ that the parties generally lack when the inputs are unknown.

\begin{proof}
For Haar sampling, choose $N_U=\max\{1,\lceil\Var[X_M]/(\delta\ve^2)\rceil\}$. By Eqs.~\eqref{eq:sample complexity-from-variance}, \eqref{eq:general-vterms}, and \eqref{eq:haar-vterms-v4}, the sufficient number of copies per party is bounded by
\begin{align}
N_{\star,\mathrm H}(\rho,\sigma;N_M)
\le
N_M+\frac{1}{\delta\ve^2}
\left[
\frac{A_n(\rho,\sigma)}{N_M}
+2\left(\frac74\right)^n
+N_M B_{n,\mathrm H}(\rho,\sigma)
\right]
\label{eq:haar-state-dependent-certificate}
\end{align}
for any positive integer $N_M$. Suppose that
\begin{align}
A_n(\rho,\sigma)B_{n,\mathrm H}(\rho,\sigma)
\le
\left(\frac{18}{5}\right)^n
\label{eq:haar-conjecture-direct-certificate}
\end{align}
and write $T=(18/5)^{n/2}$. Taking $N_M=\max\{1,\lceil A_n/T\rceil\}$ gives $A_n/N_M\le T$ and
\[
N_M B_{n,\mathrm H}\le B_{n,\mathrm H}+A_n B_{n,\mathrm H}/T\le 2T,
\]
since $B_{n,\mathrm H}\le(7/4)^n\le T$. Also $N_M\le1+3^n/T\le1+T$. Substitution in Eq.~\eqref{eq:haar-state-dependent-certificate} proves Eq.~\eqref{eq:haar-conjectured-sample-main} for $0<\ve,\delta<1$.
\end{proof}

The following classes satisfy Eq.~\eqref{eq:haar-conjecture-direct-certificate}.

\vspace*{0.2in}
\appsubsubsubsection{Product state pairs}

\begin{proposition}
\label{prop:haar-product-pair-complexity}
For arbitrary product state pairs
\begin{align}
\rho=\rho_1\ox\cdots\ox\rho_n,
\qquad
\sigma=\sigma_1\ox\cdots\ox\sigma_n,
\end{align}
one has
\begin{align}
A_n(\rho,\sigma)B_{n,\mathrm H}(\rho,\sigma)
\le
\left(\frac{18}{5}\right)^n.
\end{align}
Identical pure product states attain equality in the moment-product bound.
\end{proposition}

\begin{proof}
Both $A_n$ and $B_{n,\mathrm H}$ factorize on product state pairs, so it
suffices to prove the one-qubit bound. Write
\begin{align}
\rho_1=\frac12(\bI+\bm r\cdot\bsigma),
\qquad
\sigma_1=\frac12(\bI+\bm s\cdot\bsigma),
\end{align}
with $\|\bm r\|,\|\bm s\|\le 1$, and set
$x:=\bm r\cdot\bm s$ and $q:=\|\bm r\|\,\|\bm s\|$. Then
\begin{align}
A_1(\rho_1,\sigma_1)
=
2+\tr\left[\rho_1\sigma_1\right]
=
\frac{5+x}{2}.
\end{align}
For the Haar fourth-moment coefficient, measuring along a Haar-random Bloch
direction $\bm n$ gives the conditional kernel expectation
\begin{align}
\frac12+\frac32(\bm n\cdot\bm r)(\bm n\cdot\bm s).
\end{align}
Using the spherical moments
$\bE[n_i n_j]=\delta_{ij}/3$ and
$\bE[n_i n_j n_k n_l]
=(\delta_{ij}\delta_{kl}+\delta_{ik}\delta_{jl}+\delta_{il}\delta_{jk})/15$,
we obtain
\begin{align}
B_{1,\mathrm H}(\rho_1,\sigma_1)
=
\frac14+\frac{x}{2}+\frac{3}{20}q^2+\frac{3}{10}x^2.
\end{align}
Since $q\le 1$ and $-1\le x\le 1$,
\begin{align}
A_1(\rho_1,\sigma_1)B_{1,\mathrm H}(\rho_1,\sigma_1)
\le
\frac{5+x}{2}
\left(
\frac25+\frac{x}{2}+\frac{3x^2}{10}
\right)
\le
\frac{18}{5},
\end{align}
where the last step is an elementary one-variable maximization over
$x\in[-1,1]$, with maximum attained at $x=1$. Equality requires $q=1$ and
$x=1$, hence $\rho_1=\sigma_1$ is pure. Multiplying the one-qubit bounds over
all qubits gives the claim.
\end{proof}

For arbitrary product pairs, the separate one-qubit bounds $A_1\le3$ and $B_{1,\mathrm H}\le6/5$ give $A_n\le3^n$ and $B_{n,\mathrm H}\le(6/5)^n$. Substituting $N_M=\lceil(5/2)^{n/2}\rceil$ into Eq.~\eqref{eq:haar-state-dependent-certificate} therefore gives the copy bound in Eq.~\eqref{eq:haar-conjectured-sample-main} with a single shot allocation for the entire family.

\vspace*{0.2in}
\appsubsubsubsection{\texorpdfstring{$n=1$}{n=1}}

\begin{proposition}
\label{prop:haar-one-qubit-complexity}
For $n=1$, all state pairs satisfy
\begin{align}
A_1(\rho,\sigma)B_{1,\mathrm H}(\rho,\sigma)
\le
\frac{18}{5}.
\end{align}
Identical pure one-qubit states attain equality.
\end{proposition}

\begin{proof}
The one-qubit case is included in Prop.~\ref{prop:haar-product-pair-complexity},
so the upper bound holds for all one-qubit state pairs. Conversely, for any
pure one-qubit state $\psi$,
\begin{align}
A_1(\psi,\psi)
=
\tr\left[(2\bI+\bF)\psi^{\otimes 2}\right]
=
2+\tr\left[\psi^2\right]
=
3.
\end{align}
Using Eq.~\eqref{appx:eq:Haar-commutant-cycle-notation}, every replica
permutation has expectation one on $\psi^{\otimes 4}$, and therefore
\begin{align}
B_{1,\mathrm H}(\psi,\psi)
=
\frac15+\frac15+\frac15
+3\cdot\frac35
-4\cdot\frac3{10}
=
\frac65.
\end{align}
Thus
\begin{align}
A_1(\psi,\psi)B_{1,\mathrm H}(\psi,\psi)
=
\frac{18}{5},
\end{align}
which proves the reverse inequality and the claim.
\end{proof}

\vspace*{0.2in}
\appsubsubsubsection{\texorpdfstring{$n=2$}{n=2}}

\begin{proposition}
\label{prop:haar-two-qubit-identical-pure-complexity}
For identical pure two-qubit state pairs,
\begin{align}
\max_{\substack{\ket{\psi}\in(\bC^2)^{\otimes2}\\ \langle\psi|\psi\rangle=1}}
A_2(\psi,\psi)B_{2,\mathrm H}(\psi,\psi)
=
\left(\frac{18}{5}\right)^2,
\end{align}
and the maximum is attained by product states.
\end{proposition}

\begin{proof}
By local-unitary invariance, every pure two-qubit state
is equivalent to
\begin{align}
\ket{\psi_\lambda}
=
\sqrt{\lambda}\ket{00}
+
\sqrt{1-\lambda}\ket{11},
\qquad
0\le \lambda\le 1.
\end{align}
Set $t:=\lambda(1-\lambda)\in[0,1/4]$. The reduced one-qubit purities are
$1-2t$, and hence
\begin{align}
A_2(\psi_\lambda,\psi_\lambda)=9-8t.
\label{eq:appx-haar-schmidt-A2}
\end{align}
Direct substitution into the local Haar fourth moment gives
\begin{align}
B_{2,\mathrm H}(\psi_\lambda,\psi_\lambda)
=
\frac{36}{25}
-
\frac{58}{25}t
+
\frac{236}{25}t^2.
\label{eq:appx-haar-schmidt-B2}
\end{align}
Therefore
\begin{align}
A_2(\psi_\lambda,\psi_\lambda)B_{2,\mathrm H}(\psi_\lambda,\psi_\lambda)
=
\frac{324}{25}
-
\frac{162}{5}t
+
\frac{2588}{25}t^2
-
\frac{1888}{25}t^3.
\label{eq:appx-haar-schmidt-AB2}
\end{align}
The second derivative of the right-hand side is
\begin{align}
\frac{5176}{25}-\frac{11328}{25}t>0,
\qquad 0\le t\le \frac14.
\end{align}
Hence the cubic in Eq.~\eqref{eq:appx-haar-schmidt-AB2} is convex on $[0,1/4]$ and
attains its maximum at an endpoint. The endpoint values are
\begin{align}
\frac{324}{25}
\qquad\text{and}\qquad
\frac{203}{20}.
\end{align}
The larger value is $324/25$, attained at $t=0$, namely at product states.
This is $(18/5)^2$, proving the claim.
\end{proof}

\vspace*{0.2in}
\appsubsubsubsection{Biseparable states for \texorpdfstring{$n=3$}{n=3}}

\begin{proposition}
\label{prop:haar-three-qubit-biseparable-complexity}
For pure biseparable three-qubit states,
\begin{align}
\max_{\ket{\psi}\ \mathrm{biseparable}}
A_3(\psi,\psi)B_{3,\mathrm H}(\psi,\psi)
=
\left(\frac{18}{5}\right)^3.
\end{align}
The maximum is attained by fully product pure states.
\end{proposition}

\begin{proof}
Any biseparable pure three-qubit state is, up to permutation of the parties,
of the form
\begin{align}
\ket{\psi}=\ket{\varphi_{AB}}\ox\ket{\chi_C},
\end{align}
where $\ket{\varphi_{AB}}$ is a pure two-qubit state and $\ket{\chi_C}$ is a
pure one-qubit state. Since both $A_n$ and $B_{n,\mathrm H}$ are defined by
local tensor-product operators, they factorize across tensor-product states:
\begin{align}
A_3(\psi,\psi)B_{3,\mathrm H}(\psi,\psi)
&=
A_2(\varphi_{AB},\varphi_{AB})B_{2,\mathrm H}(\varphi_{AB},\varphi_{AB})
A_1(\chi_C,\chi_C)B_{1,\mathrm H}(\chi_C,\chi_C).
\end{align}
By Props.~\ref{prop:haar-one-qubit-complexity} and
\ref{prop:haar-two-qubit-identical-pure-complexity}, this is at most
\begin{align}
\left(\frac{18}{5}\right)^2\frac{18}{5}
=
\left(\frac{18}{5}\right)^3.
\end{align}
Equality requires the two-qubit factor to be product, and hence the whole
three-qubit state to be fully product.
\end{proof}

\vspace*{0.2in}
\appsubsubsubsection{GHZ-family}

\begin{proposition}
\label{prop:haar-ghz}
For the standard GHZ family
\begin{align}
\ket{\mathrm{GHZ}_n}:=\frac{\ket{0^n}+\ket{1^n}}{\sqrt2},
\qquad
\rho_{\mathrm{GHZ},n}:=\ketbra{\mathrm{GHZ}_n},
\end{align}
one has
\begin{align}
A_n(\rho_{\mathrm{GHZ},n},\rho_{\mathrm{GHZ},n})
B_{n,\mathrm H}(\rho_{\mathrm{GHZ},n},\rho_{\mathrm{GHZ},n})
\le
\left(\frac{18}{5}\right)^n.
\end{align}
\end{proposition}

\begin{proof}
For the GHZ state, every nontrivial proper subsystem has purity $1/2$, while
the empty subsystem and the full subsystem have purity one. Therefore
\begin{align}
A_n(\rho_{\mathrm{GHZ},n},\rho_{\mathrm{GHZ},n})
=
\frac{3^n+2^n+1}{2}.
\label{eq:haar-ghz-A}
\end{align}
For the fourth moment, write
\begin{align}
\rho_{\mathrm{GHZ},n}
=
\frac12\sum_{a,b\in\{0,1\}}
\left(\ket{a}\!\bra{b}\right)^{\otimes n}.
\end{align}
Then
\begin{align}
B_{n,\mathrm H}(\rho_{\mathrm{GHZ},n},\rho_{\mathrm{GHZ},n})
=
\frac1{16}\sum_{\ba,\bb\in\{0,1\}^4}m(\ba,\bb)^n,
\qquad
m(\ba,\bb):=\bra{\bb}R_{4,\mathrm H}\ket{\ba}.
\end{align}
A direct evaluation of this one-qubit matrix gives
\begin{align}
B_{n,\mathrm H}(\rho_{\mathrm{GHZ},n},\rho_{\mathrm{GHZ},n})
=
\frac58\left(\frac65\right)^n
+\frac12\left(\frac45\right)^n
+\frac34\left(\frac15\right)^n
+\frac12\left(-\frac15\right)^n
+\left(\frac3{10}\right)^n
+\left(-\frac3{10}\right)^n.
\label{eq:haar-ghz-value}
\end{align}
For $n=1$, the GHZ state is a one-qubit pure product state and equality holds.
For $n\ge 2$, divide the product of Eqs.~\eqref{eq:haar-ghz-A} and
\eqref{eq:haar-ghz-value} by $(18/5)^n$. Since
$(2/3)^n\le 4/9$, $3^{-n}\le 1/9$, $(1/6)^n\le 1/36$, and
$(1/4)^n\le 1/16$, we obtain
\begin{align}
\frac{A_n(\rho_{\mathrm{GHZ},n},\rho_{\mathrm{GHZ},n})
B_{n,\mathrm H}(\rho_{\mathrm{GHZ},n},\rho_{\mathrm{GHZ},n})}{(18/5)^n}
&\le
\frac{7}{9}
\left(
\frac58+\frac29+\frac{5}{144}+\frac18
\right)
=
\frac{1015}{1296}
<1.
\end{align}
This proves the claimed bound.
\end{proof}

\vspace*{0.2in}
\appsubsubsubsection{W-family}

\begin{proposition}
\label{prop:haar-w}
For the $W$ family
\begin{align}
\ket{W_n}:=\frac{1}{\sqrt n}\sum_{\ell=1}^n \ket{0^{\ell-1}1\,0^{n-\ell}},
\qquad
\rho_{W,n}:=\ketbra{W_n},
\end{align}
one has
\begin{align}
A_n(\rho_{W,n},\rho_{W,n})B_{n,\mathrm H}(\rho_{W,n},\rho_{W,n})
\le
\left(\frac{18}{5}\right)^n.
\end{align}
\end{proposition}

\begin{proof}
For a subsystem of size $k$, the reduced state of $\rho_{W,n}$ has purity
$((n-k)^2+k^2)/n^2$. Summing the second-term contribution over subsystem sizes gives
\begin{align}
A_n(\rho_{W,n},\rho_{W,n})
=
\frac{1}{n^2}
\sum_{k=0}^n
\binom{n}{k}2^{n-k}\bigl((n-k)^2+k^2\bigr)
=
3^n\frac{5n+4}{9n}.
\label{eq:haar-w-A}
\end{align}
For the fourth moment, write
\begin{align}
\rho_{W,n}
=
\frac1n\sum_{i,j=1}^n \ket{e_i}\!\bra{e_j},
\end{align}
where $\ket{e_i}:=\ket{0^{i-1}1\,0^{n-i}}$. Grouping the eight replica indices
by the number $r$ of active sites gives
\begin{align}
B_{n,\mathrm H}(\rho_{W,n},\rho_{W,n})
=
\frac{1}{n^4}
\sum_{r=1}^4
(n)_r
\left(\frac65\right)^{n-r}
C_r,
\end{align}
where $(n)_r=n(n-1)\cdots(n-r+1)$ counts the injective assignments of $r$ blocks to the $n$ sites and is zero for $r>n$. Label the ket and bra excitations of replica $j$ by $K_j$ and $B_j$, respectively. In a set partition $\pi$ of these eight labels, each block $D$ groups the labels assigned to the same occupied site. Let $a_D,b_D$ be the four-bit strings whose $j$th entries indicate whether $K_j$ and $B_j$ belong to $D$, respectively. The coefficient for $r$ occupied sites is
\begin{align}\label{eq:haar-w-Cr}
C_r=\sum_{\pi:\,|\pi|=r}\ \prod_{D\in\pi}\langle b_D|R_{4,\mathrm H}|a_D\rangle .
\end{align}
Since $R_{4,\mathrm H}$ preserves Hamming weight, a nonzero matrix element requires $|a_D|=|b_D|$. Each block is nonempty, so a contributing block must contain at least one of the four ket labels. Hence $r\le4$. For $r=1$, all eight labels occupy one site and $C_1=\langle1111|R_{4,\mathrm H}|1111\rangle=6/5$. Each unoccupied site contributes $\langle0000|R_{4,\mathrm H}|0000\rangle=6/5$. Direct enumeration gives
\begin{align}
C_1=\frac65,
\qquad
C_2=\frac{254}{25},
\qquad
C_3=\frac{587}{125},
\qquad
C_4=\frac{1561}{2500}.
\end{align}
Therefore
\begin{align}
B_{n,\mathrm H}(\rho_{W,n},\rho_{W,n})
=
\frac{1561n^3+4722n^2+11483n-12582}{5184\,n^3}
\left(\frac65\right)^n.
\label{eq:haar-w-value}
\end{align}
Combining Eqs.~\eqref{eq:haar-w-A} and \eqref{eq:haar-w-value},
\begin{align}
\frac{A_n(\rho_{W,n},\rho_{W,n})B_{n,\mathrm H}(\rho_{W,n},\rho_{W,n})}{(18/5)^n}
=
\frac{(5n+4)(1561n^3+4722n^2+11483n-12582)}{46656\,n^4}.
\end{align}
The right-hand side is equal to one at $n=1$. For $n\ge2$,
\begin{align}
&46656n^4-(5n+4)(1561n^3+4722n^2+11483n-12582)
\notag\\
&\qquad=
(n-1)(38851n^3+8997n^2-67306n-50328)>0.
\end{align}
The cubic factor is positive for $n\ge2$: its derivative is positive on
$[2,\infty)$ and its value at $n=2$ is positive.
Hence the state-dependent product of the second- and fourth-moment coefficients is bounded by $(18/5)^n$, as claimed.
\end{proof}

\vspace*{0.2in}
\appsubsubsubsection{Bell-dimer family}

\begin{proposition}
\label{prop:haar-bell-dimer}
Let
$\ket{\Phi^+}:=(\ket{00}+\ket{11})/\sqrt2$ and
$\rho_{\mathrm{Bell}}:=\ketbra{\Phi^+}$. Define
\begin{align}
\rho_{\mathrm{BD},2m}:=\rho_{\mathrm{Bell}}^{\ox m},
\qquad
\rho_{\mathrm{BD},2m+1}:=\rho_{\mathrm{Bell}}^{\ox m}\ox\ketbra{\phi},
\end{align}
where $\ket{\phi}$ is an arbitrary one-qubit pure state. Then
\begin{align}
A_n(\rho_{\mathrm{BD},n},\rho_{\mathrm{BD},n})
&=
7^{\lfloor n/2\rfloor}3^{n\bmod 2},
\\
B_{n,\mathrm H}(\rho_{\mathrm{BD},n},\rho_{\mathrm{BD},n})
&=
\left(\frac{29}{20}\right)^{\lfloor n/2\rfloor}
\left(\frac65\right)^{n\bmod 2},
\\
A_n(\rho_{\mathrm{BD},n},\rho_{\mathrm{BD},n})
B_{n,\mathrm H}(\rho_{\mathrm{BD},n},\rho_{\mathrm{BD},n})
&\le
\left(\frac{18}{5}\right)^n.
\end{align}
\end{proposition}

\begin{proof}
The Bell-dimer state is a tensor product across Bell pairs and, for odd $n$,
one additional one-qubit pure state. Hence both $A_n$ and $B_{n,\mathrm H}$
factorize over these tensor factors. For a Bell pair, Eqs.~\eqref{eq:appx-haar-schmidt-A2}
and \eqref{eq:appx-haar-schmidt-B2} at $t=1/4$ give
\begin{align}
A_2(\rho_{\mathrm{Bell}},\rho_{\mathrm{Bell}})=7,
\qquad
B_{2,\mathrm H}(\rho_{\mathrm{Bell}},\rho_{\mathrm{Bell}})=\frac{29}{20}.
\end{align}
For a one-qubit pure state,
\begin{align}
A_1(\phi,\phi)=3,
\qquad
B_{1,\mathrm H}(\phi,\phi)=\frac65.
\end{align}
The displayed formulas for $A_n$ and $B_{n,\mathrm H}$ follow by multiplying
these tensor factors. Therefore, with $r:=n\bmod 2$,
\begin{align}
A_n(\rho_{\mathrm{BD},n},\rho_{\mathrm{BD},n})B_{n,\mathrm H}(\rho_{\mathrm{BD},n},\rho_{\mathrm{BD},n})
&=
\left(7\cdot\frac{29}{20}\right)^{\lfloor n/2\rfloor}
\left(3\cdot\frac65\right)^r.
\end{align}
Since
\begin{align}
7\cdot\frac{29}{20}
&=
\frac{203}{20}
<
\left(\frac{18}{5}\right)^2,
\qquad
3\cdot\frac65
=
\frac{18}{5},
\end{align}
the desired bound follows.
\end{proof}

\appsubsubsection{Opposing effects of entanglement}

The certificate in Eq.~\eqref{eq:haar-conjecture-direct-certificate} should be
understood as a competition between two opposing effects of entanglement.

On the one hand,
\begin{align}
A_n\left(\psi,\psi\right) = \sum_{S\subseteq[n]}2^{\,n-\left|S\right|}\tr\left[\psi_S^2\right]
\end{align}
is a weighted sum of subsystem purities. Hence $A_n\left(\psi,\psi\right)\leq 3^n$
with equality if and only if $\ket{\psi}$ is a product state. In this sense,
$A_n$ \emph{penalizes entanglement}: once $\ket{\psi}$ is entangled, 
some reduced states become mixed, and the corresponding purities decrease.

On the other hand, $B_{n,\mathrm H}\left(\psi,\psi\right)$ behaves differently.
It is not simply maximized by product states, and it may even increase for
entangled states. For instance, in the two-qubit Schmidt family
\begin{align}
\ket{\psi_\lambda}
=
\sqrt{\lambda}\ket{00}
+
\sqrt{1-\lambda}\ket{11},
\qquad
t=\lambda\left(1-\lambda\right),
\end{align}
one has
\begin{align}
A_2\left(\psi_\lambda,\psi_\lambda\right)
&=
9-8t,
\\
B_{2,\mathrm H}\left(\psi_\lambda,\psi_\lambda\right)
&=
\frac{36}{25}
-\frac{58}{25}t
+\frac{236}{25}t^2.
\end{align}
Thus
\begin{align}
A_2\big|_{t=0}
&=
9,
&
B_{2,\mathrm H}\big|_{t=0}
&=
\frac{36}{25},
\end{align}
for product states, whereas for the Bell state \(t=1/4\),
\begin{align}
A_2\big|_{t=1/4}
&=
7,
&
B_{2,\mathrm H}\big|_{t=1/4}
&=
\frac{29}{20}.
\end{align}
At the Bell endpoint, $A_2$ is smaller and $B_{2,\mathrm H}$ is slightly larger than at the pure product endpoint. Their product satisfies
\begin{align}
A_2B_{2,\mathrm H}\big|_{t=0}
&=
9\cdot\frac{36}{25}
=
\frac{324}{25}
>
\frac{203}{20}
=
7\cdot\frac{29}{20}
=
A_2B_{2,\mathrm H}\big|_{t=1/4},
\end{align}
so the product state still gives the larger state-dependent product of the second- and fourth-moment coefficients.

For even $n$, the Bell-dimer formula gives $B_{n,\mathrm H}=(29/20)^{n/2}$. At the pure-product allocation $N_M=\Theta((5/2)^{n/2})$, the fourth contribution to $N_M\Var[X_M]$ scales as $\Theta((29/8)^{n/2})$. The second and third contributions are nonnegative, and the first is $-N_M$ for identical pure states. Hence this allocation exceeds the $(18/5)^{n/2}$ scale in the Chebyshev bound on Bell dimers, even though this family satisfies Eq.~\eqref{eq:haar-conjecture-direct-certificate} with a state-dependent allocation.

\appsection{DIPE with independent single-qubit classical shadow}
\label{app:independent-pauli-shadow}

This appendix gives the details behind
Prop.~\ref{prop:independent-pauli-shadow-complexity-main}. The protocol is
based on independent local classical shadows and should be distinguished from
the local randomized-measurement estimator studied in the main text: here Alice
and Bob do not share their single-qubit Pauli bases. Instead, each party forms
an empirical local Pauli shadow state, and the final estimate is the
Hilbert--Schmidt inner product of the two empirical shadows.

\appsubsection{The estimation protocol}

\textbf{Pauli shadow.}
The one-qubit Pauli shadow is obtained by choosing one of the
three Pauli bases $X,Y,Z$ uniformly at random, measuring in that basis, and
then applying the inverse of the corresponding classical-shadow measurement
channel. Equivalently, define the six one-qubit Pauli eigenstate projectors
\begin{align}
\cS_{\mathrm P}
:=
\left\{
\proj{0},
\proj{1},
\proj{+},
\proj{-},
\proj{i+},
\proj{i-}
\right\},
\end{align}
where $\ket{i\pm}:=(\ket{0}\pm i\ket{1})/\sqrt2$. The induced POVM elements
are
\begin{align}
E_\psi=\frac13\psi,\qquad \psi\in\cS_{\mathrm P},
\end{align}
because each of the three bases is selected with probability $1/3$. For an
observed one-qubit projector $\psi$, the inverse channel assigns the snapshot
\begin{align}
s_\psi:=3\psi-\bI_2 .
\end{align}
The six states in $\cS_{\mathrm P}$ form a one-qubit projective $3$-design, and this
measurement is the standard single-qubit Pauli shadow~\cite{huang2020shadows}.
For an $n$-qubit state, the measurement and reconstruction are applied
independently on each qubit. Thus, for an outcome
$\bm\psi=(\psi_1,\ldots,\psi_n)$, the corresponding product shadow snapshot of
$\rho$ is
\begin{align}
\hat{\rho}_{\bm\psi}:=\bigotimes_{\ell=1}^n s_{\psi_\ell}.
\end{align}
If $\bm\psi$ is sampled from the Pauli measurement distribution of $\rho$,
then the tensor-product inverse channel guarantees
\begin{align}
\bE[\hat{\rho}_{\bm\psi}]=\rho.
\end{align}
We use the analogous notation $\hat{\sigma}_{\bm\varphi}$ for Bob's local
Pauli-shadow snapshots sampled from $\sigma$.

\vspace*{0.1in}
\textbf{DIPE with independent Pauli shadows.}
Alice and Bob do not share their single-qubit Pauli bases. 
Instead, Alice constructs the empirical local Pauli shadow
$N^{-1}\sum_i \hat{\rho}_i$, Bob constructs the independent empirical local
Pauli shadow $N^{-1}\sum_j \hat{\sigma}_j$, and the final estimate is the Hilbert--Schmidt
inner product of these two empirical shadows. The corresponding estimation protocol is given in
Alg.~\ref{alg:independent-pauli-shadow-dipe}.

\begin{algorithm}[H]
\caption{DIPE with independent Pauli shadows}
\label{alg:independent-pauli-shadow-dipe}
\begin{algorithmic}[1]
\REQUIRE $N$ copies of $\rho$ held by Alice and $N$ copies of $\sigma$ held by Bob.
\ENSURE An estimate of $f=\tr\left[\rho\sigma\right]$.
\STATE For each $i=1,\ldots,N$, Alice independently measures every qubit in a uniformly random Pauli basis and forms $\hat{\rho}_i=\bigotimes_{\ell=1}^n(3\psi_{i,\ell}-\bI_2)$.
\STATE For each $j=1,\ldots,N$, Bob independently measures every qubit in a uniformly random Pauli basis and forms $\hat{\sigma}_j=\bigotimes_{\ell=1}^n(3\varphi_{j,\ell}-\bI_2)$.
\STATE Return
\[
\wh g_{\mathrm P}:=
\tr\left[
\left(\frac1N\sum_{i=1}^N \hat{\rho}_i\right)
\left(\frac1N\sum_{j=1}^N \hat{\sigma}_j\right)
\right]
=
\frac1{N^2}\sum_{i,j=1}^N\tr\left[\hat{\rho}_i\hat{\sigma}_j\right].
\]
\end{algorithmic}
\end{algorithm}

\appsubsection{Proof of Prop.~\ref{prop:independent-pauli-shadow-complexity-main}}

Since Alice's and Bob's shadows are independent and unbiased,
\begin{align}
\bE[\wh g_{\mathrm P}]
=
\frac1{N^2}\sum_{i,j=1}^N
\tr\left[\bE[\hat{\rho}_i]\bE[\hat{\sigma}_j]\right]
=
\tr\left[\rho\sigma\right].
\end{align}

Write
\begin{align}
h_{ij}:=\tr\left[\hat{\rho}_i\hat{\sigma}_j\right],
\qquad
\wh g_{\mathrm P}=\frac1{N^2}\sum_{i,j=1}^N h_{ij},
\qquad
f:=\tr\left[\rho\sigma\right].
\end{align}
Introduce the three second-moment quantities
\begin{align}
A_{\rho,\sigma}^{\mathrm P}
:=
\bE_{\hat{\rho},\hat{\sigma}}\tr\left[\hat{\rho}\hat{\sigma}\right]^2,\qquad
B_{\rho,\sigma}^{\mathrm P}
:=
\bE_{\hat{\rho}}\tr\left[\hat{\rho}\sigma\right]^2,\qquad
B_{\sigma,\rho}^{\mathrm P}
:=
\bE_{\hat{\sigma}}\tr\left[\rho\hat{\sigma}\right]^2.
\end{align}
Expanding the square gives
\begin{align}
\bE[\wh g_{\mathrm P}^2]
=
\frac1{N^4}
\sum_{i,j,q,r=1}^N
\bE[h_{ij}h_{qr}].
\end{align}
The summands are classified by the same index coincidences as in the global
shadow analysis of Ref.~\cite{anshu2022distributed}. If $i=q$ and $j=r$, there
are $N^2$ terms, each equal to $A_{\rho,\sigma}^{\mathrm P}$. If
$i\neq q$ and $j\neq r$, all four snapshots are independent and each term
equals $f^2$. If $i=q$ and $j\neq r$, conditioning on $\hat{\rho}_i$ gives
$B_{\rho,\sigma}^{\mathrm P}$. If $i\neq q$ and $j=r$, the analogous
contribution is $B_{\sigma,\rho}^{\mathrm P}$. Therefore
\begin{align}
\bE[\wh g_{\mathrm P}^2]
=
\frac{A_{\rho,\sigma}^{\mathrm P}}{N^2}
+
\frac{(N-1)^2}{N^2}f^2
+
\frac{N-1}{N^2}
\left(B_{\rho,\sigma}^{\mathrm P}+B_{\sigma,\rho}^{\mathrm P}\right),
\end{align}
and hence
\begin{align}
\label{eq:pauli-shadow-exact-variance}
\Var[\wh g_{\mathrm P}]
=
\frac{A_{\rho,\sigma}^{\mathrm P}}{N^2}
+
\frac{N-1}{N^2}
\left(B_{\rho,\sigma}^{\mathrm P}+B_{\sigma,\rho}^{\mathrm P}\right)
-
\frac{2N-1}{N^2}f^2.
\end{align}
Dropping the negative term and using $(N-1)/N^2\le 1/N$, we obtain
\begin{align}
\label{eq:pauli-shadow-variance-reduction}
\Var[\wh g_{\mathrm P}]
\le
\frac{A_{\rho,\sigma}^{\mathrm P}}{N^2}
+
\frac{B_{\rho,\sigma}^{\mathrm P}+B_{\sigma,\rho}^{\mathrm P}}{N}.
\end{align}

It remains to bound $A_{\rho,\sigma}^{\mathrm P}$ and
$B_{\rho,\sigma}^{\mathrm P}$. For one qubit define
\begin{align}
\omega_2^{\mathrm P}
:=
\sum_{\psi,\varphi\in\cS_{\mathrm P}}
E_\psi\ox E_\varphi\,\tr\left[s_\psi s_\varphi\right]^2.
\end{align}
Then
\begin{align}
A_{\rho,\sigma}^{\mathrm P}
=
\tr\left[
(\rho\ox\sigma)(\omega_2^{\mathrm P})^{\otimes n}
\right].
\end{align}
For $\psi,\varphi\in\cS_{\mathrm P}$,
\begin{align}
\tr\left[s_\psi s_\varphi\right]
=
\tr\left[(3\psi-\bI_2)(3\varphi-\bI_2)\right]
=
9\tr\left[\psi\varphi\right]-4.
\end{align}
Thus $\tr\left[s_\psi s_\varphi\right]^2$ is $25$ if $\psi=\varphi$, $16$ if
$\psi$ and $\varphi$ are orthogonal states in the same Pauli basis, and
$1/4$ if they belong to different Pauli bases. By one-qubit Clifford
symmetry,
\begin{align}
\omega_2^{\mathrm P}=\alpha\bI_4+\beta\bF.
\end{align}
The two traces determining $\alpha,\beta$ are
\begin{align}
\tr\left[\omega_2^{\mathrm P}\right]
&=
\frac19\left(6\cdot25+6\cdot16+24\cdot\frac14\right)
=28,\\
\tr\left[\omega_2^{\mathrm P}\bF\right]
&=
\frac19\left(6\cdot25+24\cdot\frac14\cdot\frac12\right)
=17.
\end{align}
Since
\begin{align}
\tr\left[\alpha\bI_4+\beta\bF\right]=4\alpha+2\beta,\qquad
\tr\left[(\alpha\bI_4+\beta\bF)\bF\right]=2\alpha+4\beta,
\end{align}
we get $\alpha=13/2$ and $\beta=1$. Hence
\begin{align}
\omega_2^{\mathrm P}
=
\frac{13}{2}\bI_4+\bF
=
\frac{15}{2}\Pi_{\mathrm{sym}}^{(2)}
+
\frac{11}{2}\Pi_{\mathrm{asym}}^{(2)}.
\end{align}
Therefore $\norm{\omega_2^{\mathrm P}}{\infty}=15/2$, and
\begin{align}
\label{eq:pauli-shadow-A-bound}
A_{\rho,\sigma}^{\mathrm P}
\le
\left(\frac{15}{2}\right)^n.
\end{align}
This bound is saturated by identical pure product states. Indeed, if
$\rho=\sigma=\bigotimes_{\ell=1}^n\ketbra{\phi_\ell}$, then each local two-copy
factor lies in the symmetric subspace, and hence
$A_{\rho,\sigma}^{\mathrm P}=(15/2)^n$.

For the one-sided term define
\begin{align}
\omega_3^{\mathrm P}
:=
\sum_{\psi\in\cS_{\mathrm P}}
E_\psi\ox s_\psi\ox s_\psi.
\end{align}
Then
\begin{align}
B_{\rho,\sigma}^{\mathrm P}
=
\tr\left[
(\rho\ox\sigma\ox\sigma)(\omega_3^{\mathrm P})^{\otimes n}
\right].
\end{align}
Using $E_\psi=\psi/3$ and $s_\psi=3\psi-\bI_2$,
\begin{align}
\omega_3^{\mathrm P}
=
\frac13\sum_{\psi\in\cS_{\mathrm P}}
\left(
9\psi\ox\psi\ox\psi
-3\psi\ox\psi\ox\bI_2
-3\psi\ox\bI_2\ox\psi
+\psi\ox\bI_2\ox\bI_2
\right).
\end{align}
The projective $3$-design identities for $\cS_{\mathrm P}$ are
\begin{align}
\sum_{\psi\in\cS_{\mathrm P}}\psi=3\bI_2,\qquad
\sum_{\psi\in\cS_{\mathrm P}}\psi^{\otimes2}=2\Pi_{\mathrm{sym}}^{(2)},\qquad
\sum_{\psi\in\cS_{\mathrm P}}\psi^{\otimes3}=\frac32\Pi_{\mathrm{sym}}^{(3)}.
\end{align}
Substitution gives
\begin{align}
\omega_3^{\mathrm P}
=
\frac92\Pi_{\mathrm{sym}}^{(3)}
-2\Pi_{\mathrm{sym},12}^{(2)}
-2\Pi_{\mathrm{sym},13}^{(2)}
+\bI_2^{\otimes3},
\end{align}
where
\begin{align}
\Pi_{\mathrm{sym},12}^{(2)}=\frac{\bI+\bF_{12}}2,\qquad
\Pi_{\mathrm{sym},13}^{(2)}=\frac{\bI+\bF_{13}}2
\end{align}
are embedded in $(\bC^2)^{\otimes3}$. Direct diagonalization gives
\begin{align}
\operatorname{spec}(\omega_3^{\mathrm P})
=
\left\{
\frac32 \text{ with multiplicity }4,\;
-2 \text{ with multiplicity }2,\;
0 \text{ with multiplicity }2
\right\}.
\end{align}
Thus $\norm{\omega_3^{\mathrm P}}{\infty}=2$. Since
$B_{\rho,\sigma}^{\mathrm P}$ is a square moment, it is nonnegative, and
\begin{align}
\label{eq:pauli-shadow-B-bound}
0\le B_{\rho,\sigma}^{\mathrm P}\le 2^n.
\end{align}
The same bound holds for $B_{\sigma,\rho}^{\mathrm P}$.

Combining Eqs.~\eqref{eq:pauli-shadow-variance-reduction},
\eqref{eq:pauli-shadow-A-bound}, and \eqref{eq:pauli-shadow-B-bound} gives
\begin{align}
\label{eq:pauli-shadow-variance-bound}
\Var[\wh g_{\mathrm P}]
\le
\frac{(15/2)^n}{N^2}
+
\frac{2^{n+1}}{N}.
\end{align}
Chebyshev's inequality then gives
\begin{align}
\bP\left(|\wh g_{\mathrm P}-f|\ge \ve\right)
\le
\frac1{\ve^2}
\left(
\frac{(15/2)^n}{N^2}
+
\frac{2^{n+1}}{N}
\right).
\end{align}
Taking $N=N_{\star,\mathrm P}$ to be the ceiling of four times the maximum in Eq.~\eqref{eq:independent-pauli-shadow-sample-main} bounds the two terms by $\delta/16$ and $\delta/2$, respectively. Their sum is less than $\delta$, proving the stated copy bound.

For fixed $\ve$ and $\delta$, the leading large $n$ scale of this Chebyshev-based sufficient bound is therefore
$\sqrt{7.5^n}$. 
The dominant variance term $A_{\rho,\sigma}^{\mathrm P}$ is saturated by identical pure product states; 
hence this leading variance-based scale cannot be improved within the present Chebyshev analysis.

Finally, replacing the Pauli shadow by a local classical shadow constructed
from the single-qubit Haar ensemble does not improve this sufficient-copy
scaling, because the above analysis uses only the single-qubit projective
$3$-design identities.

\appsection{Additional notes on numerical validation}
\label{app:numerics}

This appendix supplements Sec.~\ref{sec:numerical-simulation} with numerical details that are not included in the core panels of the main text. We pursue two complementary goals. First, we resolve the state-dependent moment coefficients $A_n$, $C_n$, and $B_{n,\mathrm E}$ that determine the positive variance terms after the $N_M$-dependent prefactors are restored. Second, we examine how the fourth moment varies with the depolarization strength and the number of entangling gates.

Unless stated otherwise, all plotted values are obtained from exact state representations together with either exact finite-size tensor-network contraction or Monte Carlo averaging over local measurement settings, depending on the system size. The plotted quantities are coefficients rather than full variance terms, so they do not include the prefactors $1/N_M^2$, $(N_M-1)/N_M^2$, or $((N_M-1)/N_M)^2$. The fitting lines quoted in the main text are extracted from least-squares fits of the logarithm of the relevant quantity versus $n$, using the largest accessible system sizes for each state family.
The plots show numerical values without uncertainty estimates. For comparisons between nearby values, we give exact benchmarks in the accompanying text.

\appsubsection{State-dependent moment coefficients}

We begin with the coefficients entering the three nontrivial positive variance contributions introduced in Eq.~\eqref{eq:general-vterms}. This decomposition is especially useful for explaining why the local Clifford and local Haar ensembles behave so differently in the worst case. We use the state-dependent moment coefficients $A_n$, $C_n$, and $B_{n,\mathrm E}$ defined in Eqs.~\eqref{eq:A-n}--\eqref{eq:B-n}. All three coefficients can depend on entanglement. For fixed inputs, $A_n$ and $C_n$ are common to both ensembles, whereas $B_{n,\mathrm E}$ can depend on the ensemble.
For $\mathrm E\in\{\Cl,\mathrm H\}$, these coefficients satisfy
\begin{align}
\mathbb{V}^{(2)}_{\mathrm E}
&=
\frac{A_n(\rho,\sigma)}{N_M^2},
&
\mathbb{V}^{(3)}_{\mathrm E}
&=
\frac{N_M-1}{N_M^2}C_n(\rho,\sigma),
&
\mathbb{V}^{(4)}_{\mathrm E}
&=
\left(\frac{N_M-1}{N_M}\right)^2
B_{n,\mathrm E}(\rho,\sigma).
\end{align}

Fig.~\ref{fig:appendix-variance-components} resolves $A_n$, $C_n$, and the product $A_nB_{n,\mathrm E}$ for the five benchmark state families used throughout the paper.

For local Haar sampling, the two-qubit Bell state has a larger $B_{2,\mathrm H}$ than the pure product state, but smaller $A_2$ ($7$ versus $9$) and $C_2$ ($17/4$ versus $9/2$). Consequently, the Bell state gives $A_2B_{2,\mathrm H}=203/20$, below the pure-product value $324/25$ (Appx.~\ref{app:haar-state-dependent}).

\begin{figure*}[t]
\centering
\subfloat[Local Clifford: $A_n$.\label{fig:appendix-clifford-A}]{%
\includegraphics[width=0.32\textwidth]{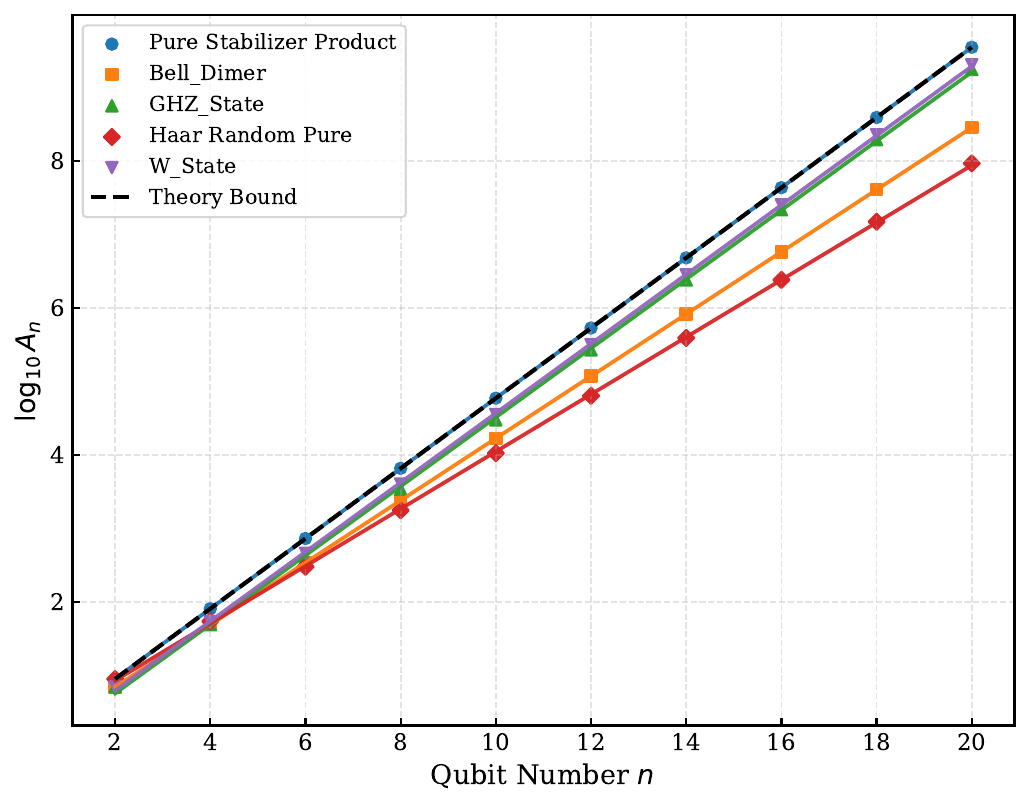}}
\hfill
\subfloat[Local Clifford: $C_n$.\label{fig:appendix-clifford-C}]{%
\includegraphics[width=0.32\textwidth]{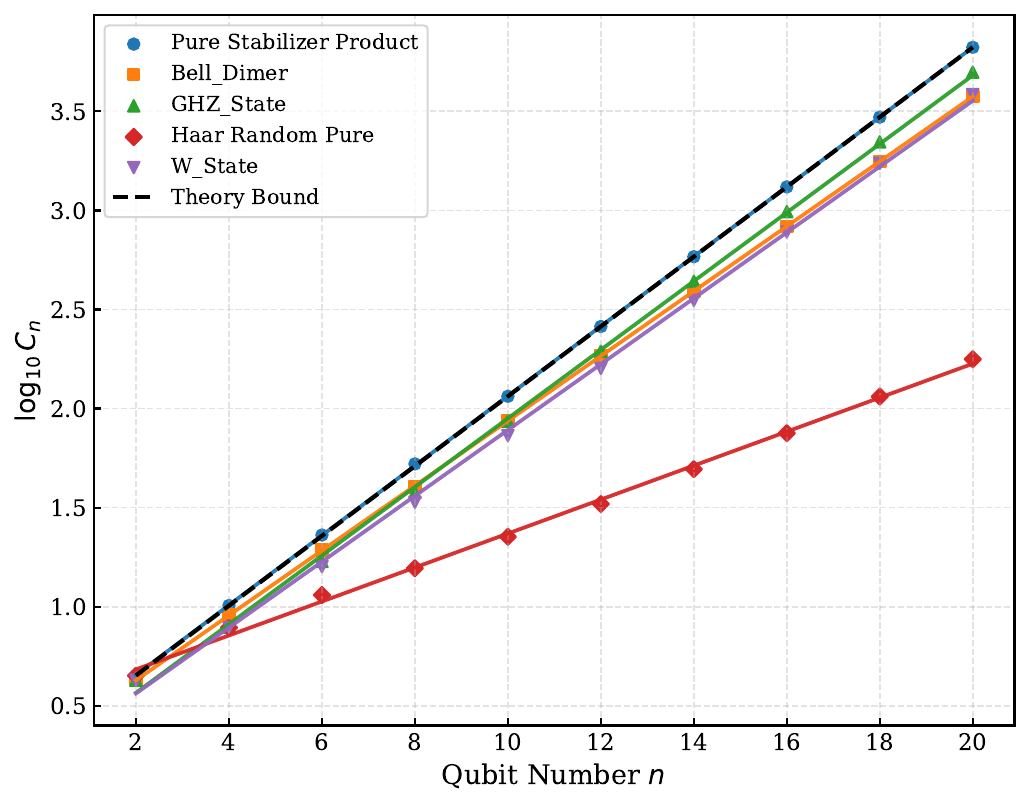}}
\hfill
\subfloat[Local Clifford: $A_nB_{n,\Cl}$.\label{fig:appendix-clifford-AxB}]{%
\includegraphics[width=0.32\textwidth]{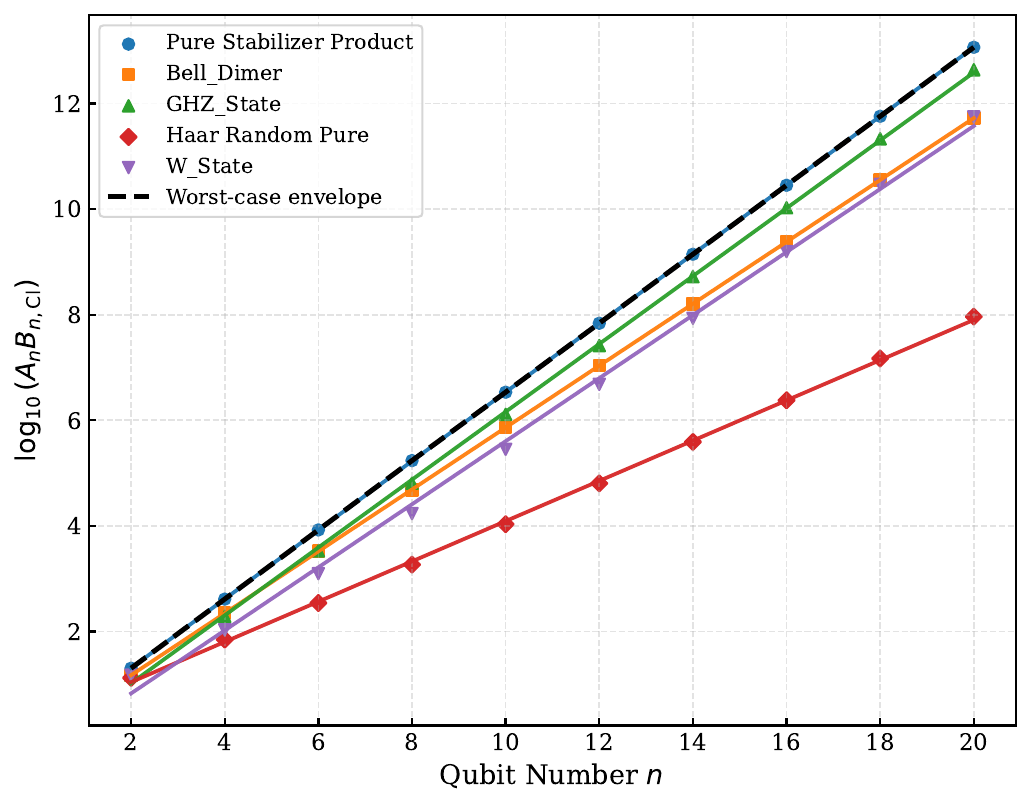}}\\[0.8ex]
\subfloat[Local Haar: $A_n$.\label{fig:appendix-haar-A}]{%
\includegraphics[width=0.32\textwidth]{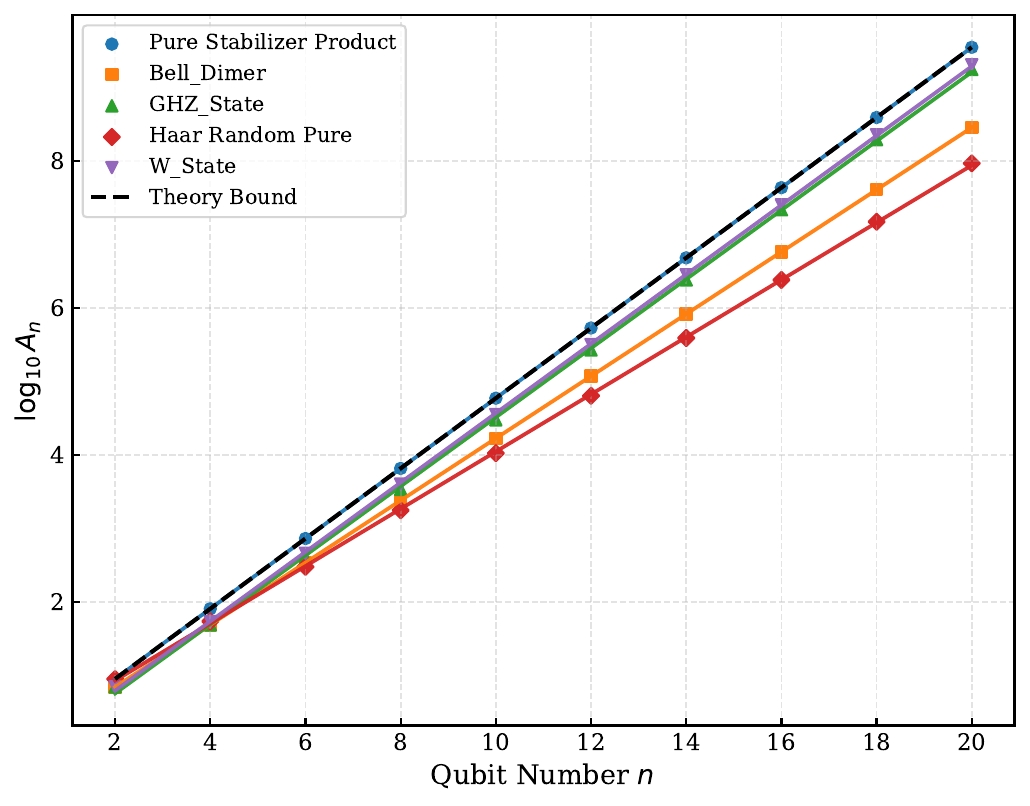}}
\hfill
\subfloat[Local Haar: $C_n$.\label{fig:appendix-haar-C}]{%
\includegraphics[width=0.32\textwidth]{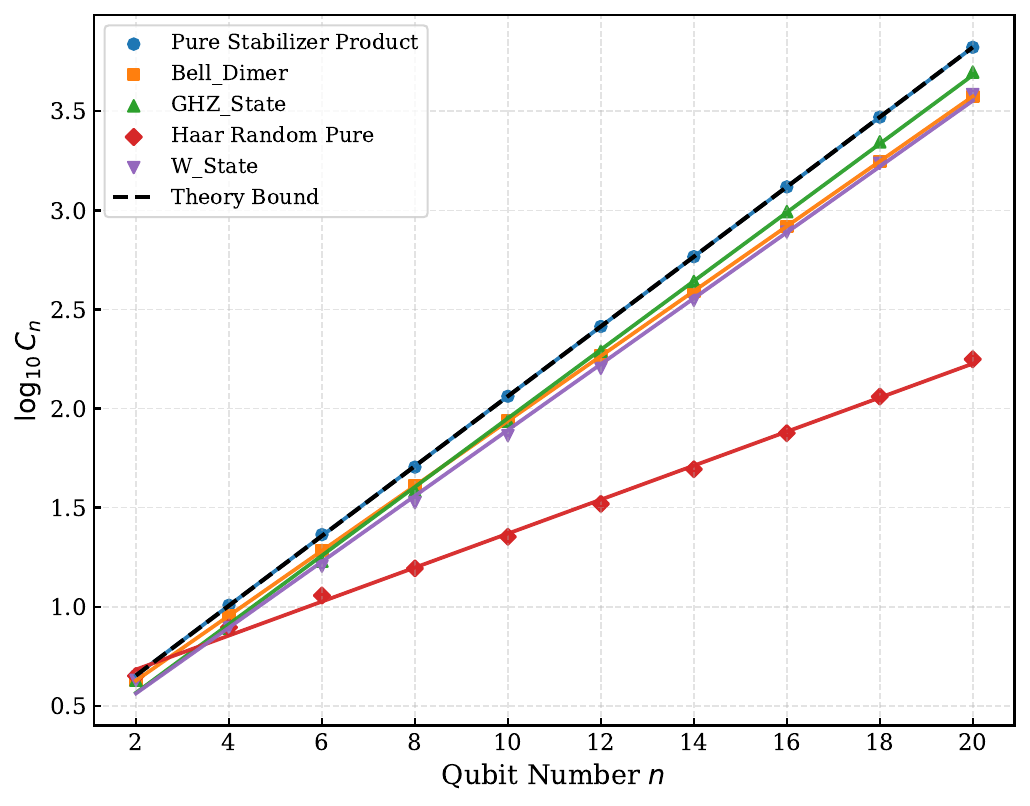}}
\hfill
\subfloat[Local Haar: $A_nB_{n,\mathrm H}$.\label{fig:appendix-haar-AxB}]{%
\includegraphics[width=0.32\textwidth]{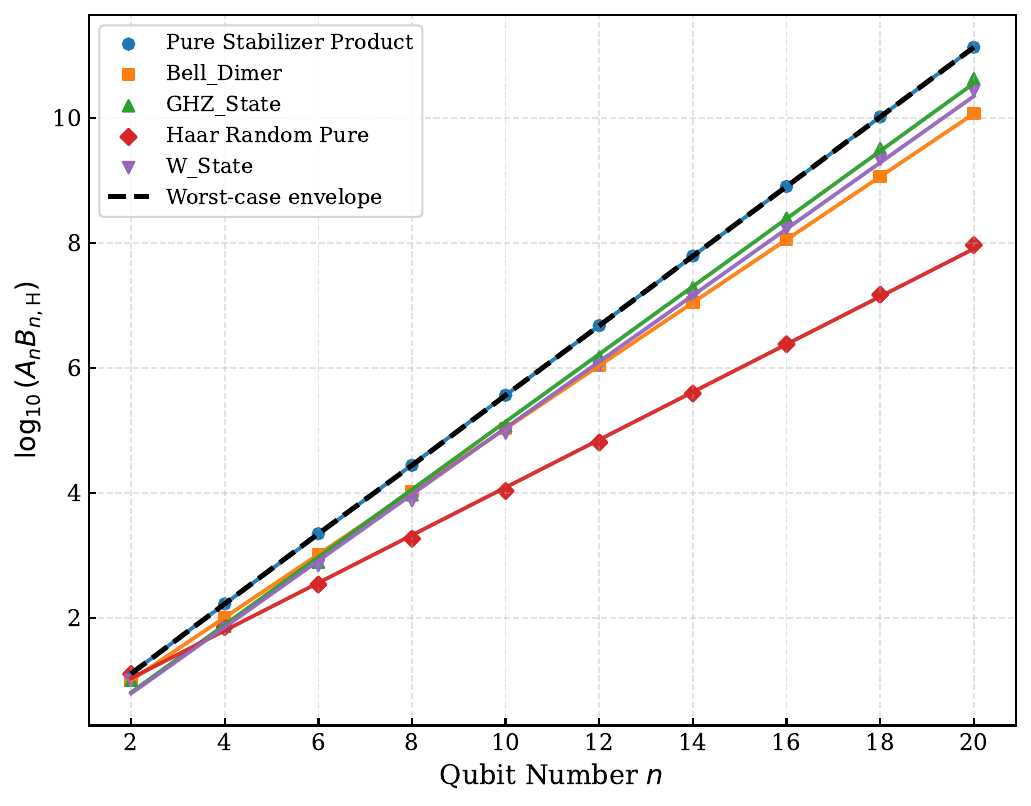}}
\caption{Resolved state-dependent moment coefficients for the benchmark state families. The top row shows local Clifford sampling and the bottom row shows local Haar sampling. For a fixed input state, $A_n$ and $C_n$ have the same exact values under both ensembles. In (a,d), the dashed line gives the all-input bound $A_n\le3^n$. In (b,e), the line labeled ``Theory Bound'' gives the exact value $C_n=2(3/2)^n$ for identical pure product inputs. Prop.~\ref{prop:third-term-bound} gives $C_n\le2(7/4)^n$ for arbitrary pairs. The envelopes for $A_nB_{n,\mathrm E}$ in (c,f) refer to the plotted families.}
\label{fig:appendix-variance-components}
\end{figure*}

\appsubsection{Parameterized state families}

We next vary the purity of a product state through a continuous depolarization parameter and the entanglement of a chain graph state through a discrete number of controlled-$Z$ gates.

\appsubsubsection{Purity deformation}

The first deformation starts from the pure reference state $\rho_0=\ketbra{+}^{\otimes n}$ and applies an independent depolarizing channel to each qubit. For a single-qubit operator $\tau$, we use
\begin{equation}
\mathcal{D}_p(\tau)=(1-p)\tau+p\,\tr\left[\tau\right]\frac{\bI_2}{2},
\qquad p\in[0,1],
\end{equation}
and the noisy $n$-qubit state is $\rho_p=\mathcal{D}_p^{\otimes n}(\rho_0)$. In the simulation, this channel is implemented equivalently at the level of measurement probabilities by applying, independently on each measured bit, the classical transition matrix with flip probability $p/2$. Thus the horizontal axis in Fig.~\ref{fig:appendix-purity-deformation} is the local depolarizing parameter $p$. With $r=1-p$, the exact fourth moments for this product family are
\[
B_{n,\mathrm E}(\rho_p,\rho_p)
=\left(\frac{1+2r^2+\kappa_{\mathrm E}r^4}{4}\right)^n,
\qquad \kappa_{\Cl}=3,\quad \kappa_{\mathrm H}=\frac95 .
\]
Both decrease throughout $0\le p\le1$ and reach $4^{-n}$ at $p=1$, including for curves whose displayed points end earlier. At $n=8$ and $p=0$, the plotted Clifford value exceeds the exact value $(3/2)^8$.

\begin{figure}[t]
\centering
\subfloat[Local Clifford.\label{fig:appendix-clifford-purity}]{%
\includegraphics[width=0.48\columnwidth]{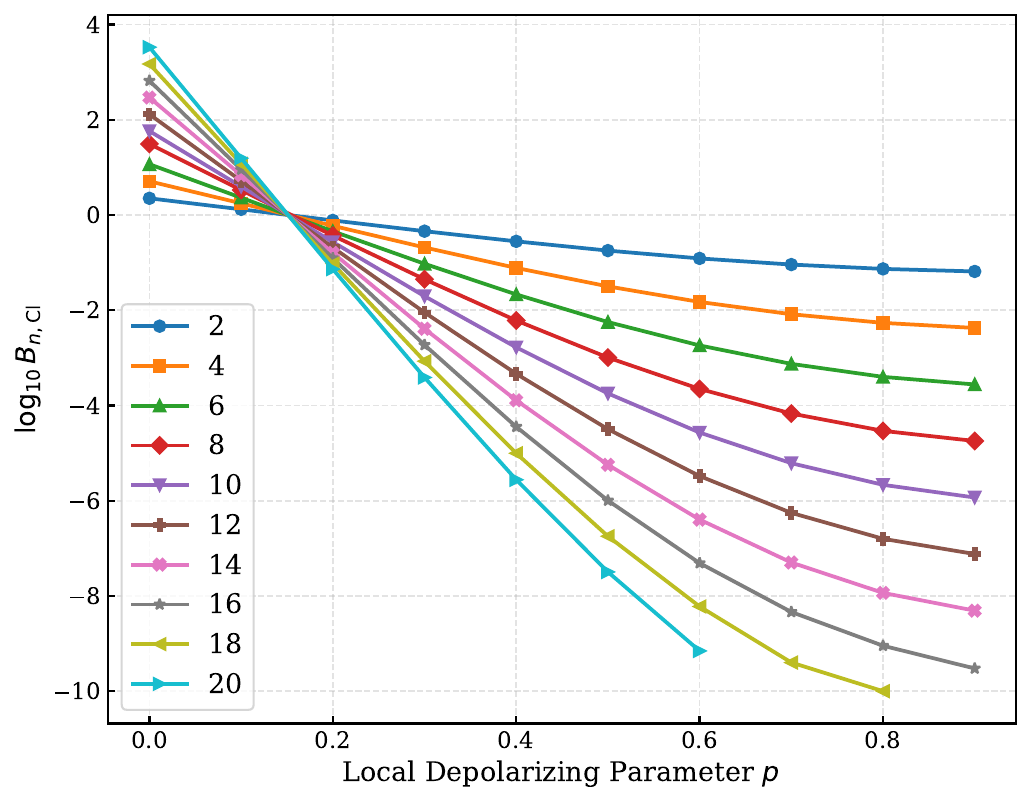}}
\hfill
\subfloat[Local Haar.\label{fig:appendix-haar-purity}]{%
\includegraphics[width=0.48\columnwidth]{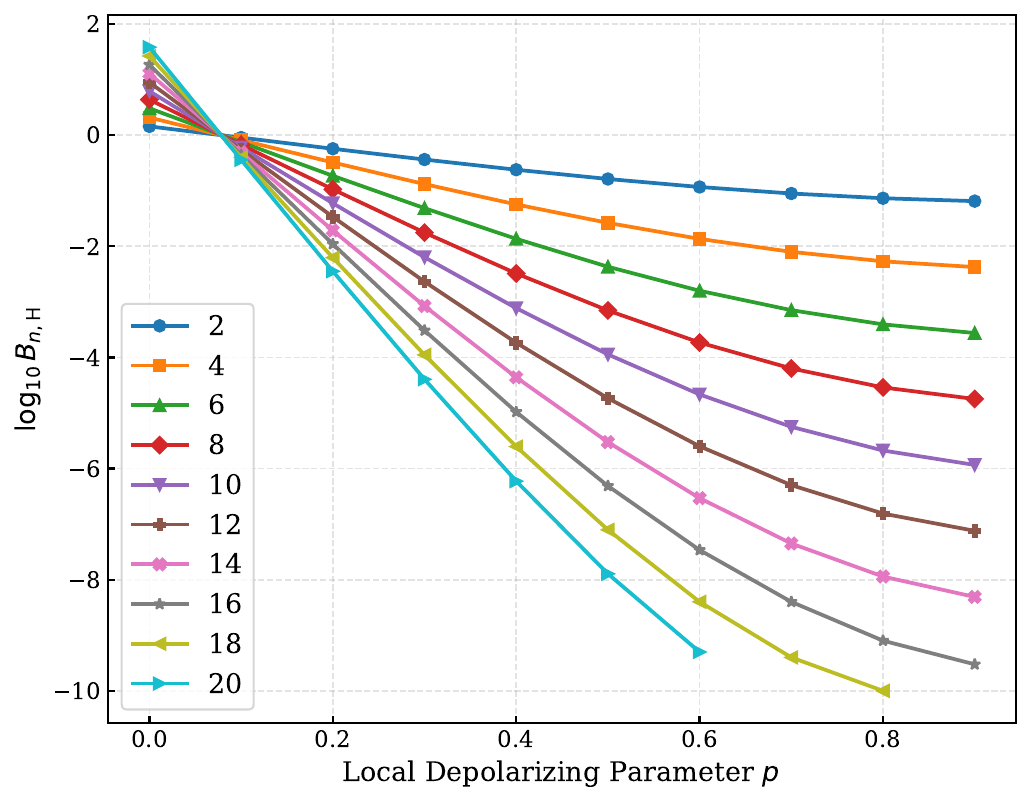}}
\caption{Variation with purity loss induced by local depolarization. Each curve corresponds to a fixed system size $n$, starting from the deterministic product state $\ket{+}^{\otimes n}$ and varying the local depolarizing parameter $p$. The states are not randomly generated; randomness enters only through the sampling of local measurement settings when Monte Carlo estimation is used. Both ensembles show decreasing fourth-moment contributions as the state moves away from the pure-state boundary.}
\label{fig:appendix-purity-deformation}
\end{figure}

\appsubsubsection{Entanglement-structure interpolation}

Finally, we vary the entanglement pattern within a deterministic family of pure stabilizer states. Starting from $\ket{+}^{\otimes n}$, we apply controlled-$Z$ gates along the first $m$ edges of a one-dimensional chain:
\begin{equation}
\ket{\phi_{n,m}}
=
\left(\prod_{j=1}^{m}CZ_{j,j+1}\right)\ket{+}^{\otimes n},
\qquad
0\le m\le n-1 .
\end{equation}
Thus $m=0$ is fully separable, while $m=n-1$ is the one-dimensional cluster state on an open chain. The outcomes are displayed in Fig.~\ref{fig:appendix-entanglement-interpolation}. For the system sizes shown, the Clifford fourth-moment coefficient decreases as edges are added. Every state in this interpolation is a pure stabilizer state and obeys the bound in Prop.~\ref{prop:cliff-stabilizer-bound}. The first $CZ$ gate produces a Bell pair up to a local unitary and leaves the remaining $n-2$ qubits in a product state. The exact Haar fourth moment therefore increases by a factor $(29/20)/(6/5)^2=145/144$ from $m=0$ to $m=1$. The plotted values at $n=6$ and $n=8$ give the opposite ordering for this step.

\begin{figure}[t]
\centering
\subfloat[Local Clifford.\label{fig:appendix-clifford-entanglement}]{%
\includegraphics[width=0.48\columnwidth]{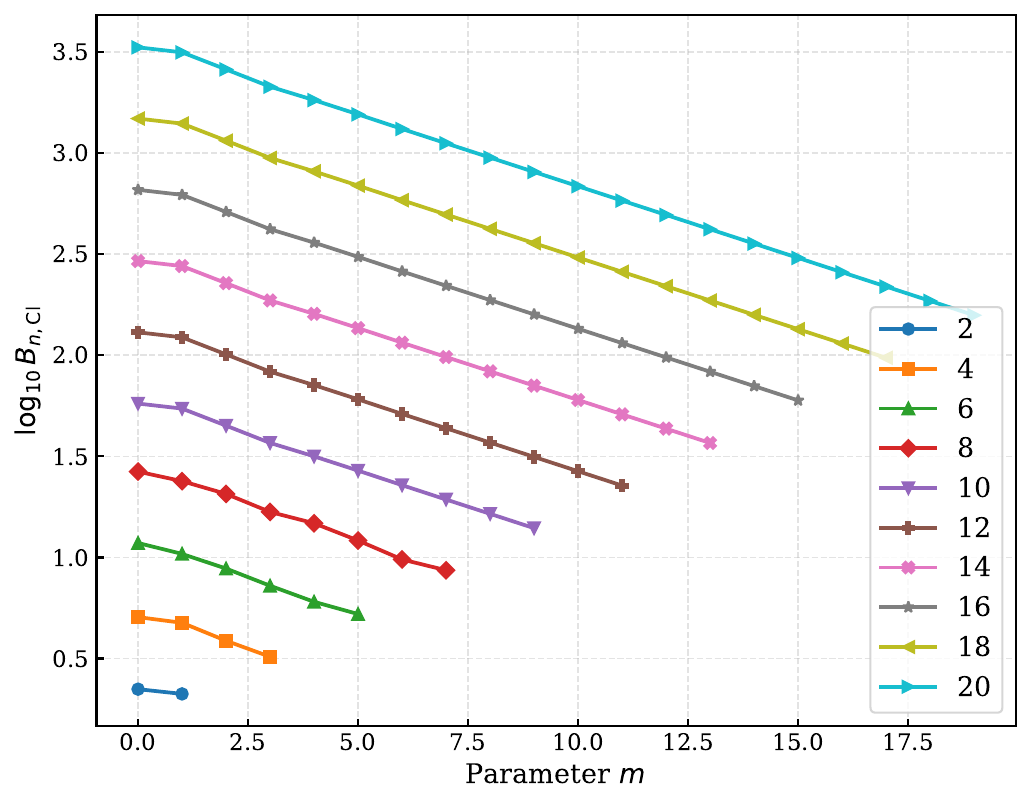}}
\hfill
\subfloat[Local Haar.\label{fig:appendix-haar-entanglement}]{%
\includegraphics[width=0.48\columnwidth]{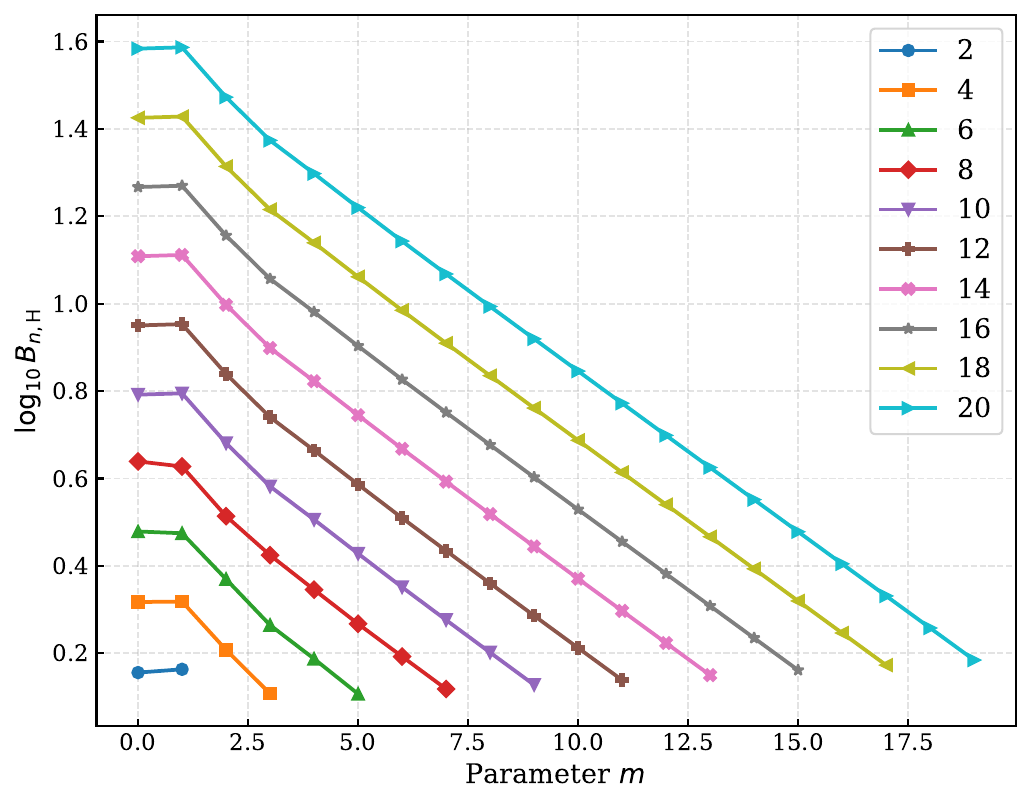}}
\caption{Variation with entanglement structure within the chain graph-state family $\ket{\phi_{n,m}}$. Each curve corresponds to a fixed system size $n$, and the horizontal axis gives the number $m$ of nearest-neighbor $CZ$ edges.}
\label{fig:appendix-entanglement-interpolation}
\end{figure}

Fig.~\ref{fig:appendix-variance-components}(f) gives the corresponding comparison of $A_nB_{n,\mathrm H}$ for the five benchmark families of Sec.~\ref{sec:numerical-simulation}.

\end{document}